\documentclass[12pt,reqno,oneside,british,english]{amsart}
\ifdefined\pdfoutput \pdfoutput=1 \fi
\usepackage{amsmath,amssymb,amsthm}
\usepackage[left=1.0in,right=1.0in,top=1.0in,bottom=1.0in]{geometry}
\usepackage[mathscr]{eucal}
\usepackage{mathrsfs}
\usepackage{bm}
\usepackage[round]{natbib}
\usepackage{babel}
\usepackage[colorlinks=true,linkcolor=blue,citecolor=blue]{hyperref}
\usepackage{float}
\usepackage{multirow}
\usepackage{indentfirst}
\usepackage{mathtools}
\usepackage{graphicx}
\usepackage{pdflscape}
\usepackage{subcaption}
\usepackage{rotating}
\usepackage{calc,xspace}
\usepackage{dcolumn}
\usepackage{mwe}
\usepackage{caption}
\usepackage{natbib}
\usepackage{color}
\usepackage{url}
\usepackage{epstopdf}
\usepackage[shortlabels]{enumitem}
\usepackage[flushleft]{threeparttable}
\usepackage{tikz}
\usetikzlibrary{shapes,arrows,calc,positioning,arrows,decorations.pathreplacing}
\allowdisplaybreaks[3]
\newcolumntype{d}[1]{D{.}{.}{#1}}
\theoremstyle{plain}
\newtheorem{proposition}{Proposition}
\newtheorem{lemma}{Lemma}
\newtheorem{corollary}{Corollary}

\newtheorem{theorem}{Theorem}

\theoremstyle{definition}

\newtheorem{assumption}{Assumption}

\theoremstyle{remark}

\DeclareMathOperator{\diag}{diag}

\DeclareMathOperator{\Var}{Var}

\DeclareMathOperator{\sgn}{sgn}
\DeclareMathOperator{\rank}{rank}

\newcommand{\divby}[2]{#1 \mathord{\left/ \vphantom{#1 #2} \right.}
 \kern-\nulldelimiterspace #2}

\newcommand{\Real}{\mathbb R}

\newcommand{\eps}{\varepsilon}

\newcommand{\norm}[1]{\left\Vert#1\right\Vert}

\renewcommand{\thesection}{\Roman{section}}

\newcounter{year}
\newcommand{\runinhead}[1]{\medskip\noindent\emph{#1.}\ }

\newcommand{\cA}{\mathscr{A}}

\newcommand{\cS}{\mathscr{S}}
\newcommand{\cL}{\mathscr{L}}
\newcommand{\cI}{\mathscr{I}}
\newcommand{\cN}{\mathscr{N}}

\newcommand{\cQ}{\mathscr{Q}}

\makeatletter
\newcommand{\vastsmall}{\bBigg@{3}}
\newcommand{\vast}{\bBigg@{4}}
\newcommand{\Vast}{\bBigg@{5}}
\makeatother

\usepackage{booktabs}
\usepackage{dcolumn}
\begin{document}

\title[Algorithm-Driven SVARs]{Algorithm-Driven SVARs:\\ Navigating the Wilderness of Big Data}
\author{Yucheng Yang and Tao Zha}
\thanks{
\indent Yang: University of Zurich and SFI; Zha: Federal Reserve Bank of Atlanta, Emory University, and NBER. We thank Lukas Hack, Tom Sargent, and participants at various seminars and conferences for helpful discussions, as well as Tong Xu for his initial involvement in this project. Hongyi Fu provided outstanding research assistance. The views expressed herein are those of the authors and do not necessarily reflect those of the Federal Reserve Bank of Atlanta, the Federal Reserve System, or the National Bureau of Economic Research.}
\date{4th August 2026}

\begin{abstract}
Every SVAR result is conditional on two choices: the restrictions that identify the shock and the variables on which they operate. The literature disciplines the first; the second is chosen by hand. We develop a Bayesian methodology that constructs information sets, uses an out-of-sample criterion, and retains the largest system it admits. Under recursive identification, output rises with housing production rather than household credit alone. For monetary policy, an anchor-free joint Bayesian proxy SVAR with multiple instruments strengthens the credit spread channel. A core system augmented with the selected corporate spread identifies expected default risk as a potent transmission margin.
\end{abstract}
\maketitle

\thispagestyle{empty}
\newpage \pagenumbering{arabic}

\section{Introduction}\label{sec:intro}

Structural vector autoregressions (SVARs) produce empirical evidence, under identifying restrictions, about the dynamic effects of economic shocks. That evidence often becomes the object that theory aims to explain. Many models are built to match impulse responses estimated from SVARs \citep{CEE99,sSimstZha06}. Every SVAR result is conditional on two choices: the restrictions that identify the shock and the variables on which those restrictions operate. The first choice is formalized, defended, debated, and, most important, disciplined. The second is usually left to the researcher. The information set behind the impulse responses is therefore not preliminary housekeeping. It is part of the evidence.

Under standard practice, a researcher begins with an economic question, chooses a set of variables by hand, imposes identifying restrictions on the resulting system, and studies the impulse responses. This practice was understandable when data were limited and computation was costly. It is no longer compelling when researchers observe hundreds of macroeconomic, financial, credit, housing, labor market, price, and interest rate series. A hand-built information set is not a neutral starting point. It is an undisciplined model selection decision. Omitting relevant variables can alter the innovations to which identification is applied and the impulse responses that theory is subsequently asked to explain. In the wilderness of big data, hand selection is less credible across researchers, while estimation of a system containing every possible candidate is impractical.

This paper develops an algorithm-driven Bayesian methodology for tackling this issue head-on by constructing the information set of an SVAR. The methodology rests on a clear division of labor. The researcher chooses the economic question, the core variables that define it, and the identifying restrictions. The algorithm constructs the information set. Bayesian out-of-sample (OOS) evidence selects its complexity. The method is not an unsupervised search that asks the data to choose the question or the identification. A study of household credit must begin with household credit and the outcomes the question is meant to address. A study of monetary policy must begin with the policy and transmission variables needed for that question. The data discipline the surrounding information set only after the economic question and its objects have been specified.

The methodology consists of two components. First, the model construction procedure searches over variables and factors outside the core that contain predictive information for the composite disturbances of the current system. It is an iterative procedure. At each iteration, the current SVAR is estimated at its posterior mode and its fitted composite disturbances are computed. An auxiliary optimization step determines whether the contemporaneous or lagged values of variables outside the current system contain information about those disturbances for each value of a model complexity parameter. Selected variables enter the system, the enlarged SVAR is reestimated, and the disturbances are updated. The procedure stops when no remaining variable enters and then delivers a terminal information set for that value of the model complexity parameter. Across the prespecified values of model complexity, it produces terminal systems of different dimensions and does not favor a small or a large model in advance.

Second, the Bayesian OOS criterion determines which complexity to retain. A validation sample that plays no role in model construction is used to evaluate forecasts of the core variables. For each loss, the rule first retains terminal systems that are credibly better than the reference. If none exists, it retains the reference and systems whose loss differences are not credibly distinguishable from zero. It then selects the largest system in the union of the squared and absolute loss comparison sets. The objective is information retention, not mechanical sparsity. The output is an interpretable SVAR with the selected information set for a specified economic question. The methodology produces an algorithm that is reproducible and fast.

We provide the theoretical foundation for this methodology. For every value of the model complexity parameter, the construction path is monotone, terminates after finitely many updates, and produces a unique terminal system. We derive an exact stopping characterization, establish that admission does not depend on measurement units, and show that sufficiently small numerical perturbations do not alter the terminal system (computationally stable). The Bayesian OOS step returns a unique model complexity parameter and hence a unique selected system. Together, these results turn specification search from hidden researcher choices into an open, reproducible algorithm and bring the construction of the information set under the same discipline long applied to identification.

This evolution matters not merely in theory but in practice. We apply the methodology to two of the most studied questions in the SVAR literature: the effects of household credit shocks and the effects of monetary policy shocks. Together, the two applications demonstrate its generality across two classic identification methods: recursive identification and proxy identification.

The first application relates to \citet[MSV hereafter]{MianSufiVerner2017}, who find that output rises and later declines after a household credit shock.  \citet[BPSS hereafter]{BrunnermeierPaliaSastrySims2021} replicate the same initial output boom in a larger monthly system and show that its interpretation depends on the information set. Both systems, however, are chosen by the researchers, and neither contains direct measures of housing production despite the central role of residential construction emphasized by \citet{Leamer2007,Leamer2015}. Using MSV's three variables---output, business credit, and household credit---as the core, our procedure selects a 13-variable system that contains four measures of housing production. No particular housing block is imposed in advance; each housing series competes with every other candidate on the same terms in model construction.

The selected information set changes the central economic conclusion, not merely its precision. A household credit shock raises household credit, but the posterior median of output shows no meaningful initial boom and all four measures of housing production decline. A separate housing production shock raises permits, starts, industrial production, and household credit together. Within the selected system, output rises when housing production expands and does not rise when household credit expands while housing production contracts. This pattern holds under both the recursive identification of MSV and the heteroskedasticity identification of BPSS. It distinguishes household credit from housing production without assigning a primitive cause to the change from the smaller systems.

The second application develops an independent Bayesian method for proxy SVARs. A common two-stage implementation chooses one reduced-form innovation as an anchor and measures every response relative to it. \citet[GK hereafter]{GertlerKaradi2015}, for example, use the innovation in the one-year Treasury rate and one external instrument. With one instrument, admissible anchors recover the same normalized impact vector. With several instruments, however, the estimated responses can depend on the anchor in finite samples, and no economic principle selects one reduced-form innovation over another. The three instruments of \citet{Swanson2021} capture distinct dimensions of monetary policy surprises---the federal funds rate, forward guidance, and large-scale asset purchases. Their joint deployment therefore provides no economic basis for privileging the reduced-form innovation of any single policy indicator as the anchor.

We establish that a vector of external instruments identifies one shock and its equation in a joint SVAR if and only if at least one instrument is correlated with that shock. The method allows any number of instruments for one shock, extends to several shocks and their corresponding instrument vectors, delivers explicit formulas for the identified equation and impulse responses, and estimates the proxy relation and the SVAR under one posterior distribution without a reduced-form innovation anchor. The model construction procedure and the joint proxy SVAR share a common principle. The first replaces a hand-chosen expansion of the information set; the second replaces a hand-chosen anchor. Both turn a consequential modeling choice into a transparent Bayesian procedure.

The empirical part of the second application studies monetary policy transmission. We begin with GK's six core variables and use the three high-frequency instruments of \citet{Swanson2021} jointly in our anchor-free proxy SVAR. The algorithm selects a 19-variable system. Rather than a dense ladder of closely related Treasury rates, it selects distinct labor market, housing, external, commodity, equity valuation, volatility, credit, and liquidity margins. 

In the core system augmented with the selected GZ spread, the difference between the corporate bond spread and the excess bond premium of \citet[GZ hereafter]{GilchristZakrajsek2012} rises, with both the 68\% and 90\% credible bands above zero over the early horizons. Monetary policy tightening therefore affects not only the excess bond premium but also the expected default risk in this reporting system. Joint estimation with multiple instruments and no anchor strengthens GK's central credit spread channel; the selected corporate spread extends the evidence to expected default risk as an integral part of monetary policy transmission.

The remainder of the paper is organized as follows. Section~\ref{sec:literature} relates our contributions to the literature. Sections~\ref{sec:modelconstruction}--\ref{sec:oos} develop the model construction procedure, establish its theoretical foundations, and prove the existence and uniqueness of the Bayesian OOS-selected system. Section~\ref{sec:household} presents the household credit application. Section~\ref{sec:mp} establishes a necessary and sufficient identification theorem for proxy SVARs with multiple external instruments and applies the joint Bayesian method to monetary policy transmission. Section~\ref{sec:conclusion} concludes. 

\section{Related Literature}\label{sec:literature}

\citet{Leamer1978} provides the conceptual foundation for treating model selection and specification search as part of econometric inference rather than as hidden preliminary choices. \citet{GuKellyXiu2020} provide a modern implementation of that principle by using separate training and validation samples to select model complexity in high dimensional predictive models. Our methodology brings this discipline into Bayesian SVAR analysis. The construction sample generates terminal information sets across values of model complexity, while a separate validation sample selects among them using posterior distributions of OOS forecast loss differences. We follow the posterior reporting convention of \citet{SimsZha1999} and the sample splitting logic of \citet{GuKellyXiu2020}, rather than the expanding-window design of a real-time forecasting exercise such as \citet{GiannoneLenzaPrimiceri2021}.

The paper is about model selection, not model averaging. Most SVAR evidence used by economic theory comes from one selected system. An SVAR requires observed economic variables, a concrete information set, and identifying restrictions that deliver clear economic interpretations. Model averaging can be useful, but it answers a different question and can obscure the economic content of the system. Our objective is to select one SVAR for a given economic question. The procedure is designed for repeated empirical use rather than as a single computational exercise. Tractability is therefore part of the methodology, not a secondary convenience. A useful construction method must be fast, reproducible, and feasible for other researchers to replicate, extend, and apply to new economic questions.

\citet{GiannoneReichlin2006} show that omitted information can make shocks recovered from an SVAR nonfundamental. \citet{ForniGambetti2014} build on this insight by using estimated factors to test whether a VAR is informationally sufficient and to amend it when deficiency is detected. \citet{JarocinskiMackowiak2017} use posterior probabilities of Granger causal priority to rank and select variables for a Bayesian VAR. This strand of literature establishes that the information set matters. Big data poses a distinct model construction problem: how to build a question-specific SVAR from hundreds of observed candidate variables. Existing approaches have not tackled this joint construction and selection problem.

Factor methods provide a different response to big data. A factor augmented VAR compresses a large panel into a small number of latent factors and appends them to a VAR \citep{BernankeBoivinEliasz2005}. In unpublished Walras--Bowley Lecture slides, \citet{Reichlin2026} places factor-augmented vector autoregressions (FAVARs) and large Bayesian VARs within the same dense regularization principle. A FAVAR imposes a discrete cutoff at the selected factors, whereas a large Bayesian VAR continuously shrinks all directions and retains information beyond the factor space. The retrospective analysis highlights the empirical success of Bayesian VARs estimated in levels under Minnesota priors and the central role of OOS evaluation in selecting model complexity.

Our methodology follows the Bayesian VAR route but extends it from estimation to model construction. We retain the core variables that define the original economic question, use Bayesian shrinkage to estimate every expanded candidate system, and use OOS evidence to select its complexity. Unlike a FAVAR, the selected object remains a system of observed economic variables whose impulse responses have direct economic meaning. Unlike a full large Bayesian VAR, the method does not require every candidate variable to enter one enormous system. It constructs an information set tailored to the economic question, including persistent and nonstationary series with the data in log levels. Constructed factors in the FAVAR literature may enter the candidate set, but the methodology does not require the information set to be represented by a handful of latent factors.

Our anchor-free joint Bayesian proxy SVAR contributes to the external instrument literature. \citet{MertensRavn2013} formulate covariance restrictions that use proxies correlated with the shocks of interest and orthogonal to other shocks. \citet{StockWatson2018} unify the SVAR and local projection approaches with instrumental variables (IVs) and clarify their identification conditions. GK implement the SVAR-IV approach by using one instrument and the reduced-form innovation in one policy indicator as an anchor. \citet{AriasRubioRamirezWaggoner2021} develop Bayesian proxy SVARs in a framework that pairs a $k$-dimensional proxy vector with a $k$-dimensional shock vector through a nonsingular $k\times k$ proxy--shock covariance matrix and thus does not accommodate multiple proxies for one shock. Their proxy equations may contain contemporaneous and lagged endogenous variables, which makes the exogeneity restrictions nonlinear functions of the structural parameters. Their importance sampler is not designed for the repeated posterior mode calculations required at every iteration and every value of the model complexity parameter. We instead keep the multiple instruments of \citet{Swanson2021}, constructed from high-frequency surprises, outside the SVAR system, consistent with GK and the earlier literature, and impose exogeneity directly on the covariance matrix.

\citet{CaldaraHerbst2019} develop the closest Bayesian predecessor to our joint approach. Their published paper develops a model with one proxy for one shock. In that scalar case, their unrestricted measurement error variance is equivalent to our specification. Their published paper does not develop several proxies for one shock. Their online appendix extends the model to two proxies but assumes independent measurement errors. Applied to our three-instrument setting, the same assumption would force every pair of measurement errors to be uncorrelated. Identification does not require this restriction. Our formulation allows an arbitrary number of instruments and leaves their covariance unrestricted, subject only to the proxy conditions required for identification. It therefore nests both their published one-proxy model and their independent-error extension.

Our contribution establishes the identification theory for this joint approach and generalizes it. We establish a necessary and sufficient population condition under which an instrument vector of arbitrary dimension identifies one shock and its equation, derive explicit formulas for the identified equation and impulse responses, and prove that the remaining equations need not be identified. We further establish population equivalence with the GK two-stage procedure, characterize their finite-sample difference with several instruments, and extend the result to several shocks with corresponding instrument vectors. Moreover, our augmented structural representation permits direct application of the efficient Gibbs sampler of \citet{WZ03b}, which makes the repeated high-dimensional proxy-SVAR estimation required by model construction computationally feasible.

The household credit application connects our results to the literature on housing production, credit allocation, and the business cycle. \citet{Leamer2007,Leamer2015} emphasize the central role of housing production in the business cycle. \citet{DavisHeathcote2005} document that residential investment leads GDP and is more volatile than nonresidential investment. \citet{RognlieShleiferSimsek2018} show how overbuilding durable capital, especially housing, can create an investment hangover and a recession. \citet{MullerVerner2024} show that credit booms disproportionately finance households, construction, real estate, and other nontradable sectors, and that nontradable credit expansions predict subsequent growth slowdowns. Together, these studies establish the economic importance of distinguishing housing production from household credit when interpreting the output boom. Our OOS-selected SVAR makes this distinction: output rises when housing production expands, not when household credit rises alone.

The monetary policy application connects the expected default response to theory and firm-level evidence. In \citet{GomesJermannSchmid2016}, a monetary policy tightening lowers inflation, raises the real burden of long maturity nominal debt, and increases corporate leverage and default. \citet{PalazzoYamarthy2022} find that contractionary monetary surprises raise both the expected loss and risk premium components of firm credit default swap spreads, with larger responses among riskier firms. Our aggregate evidence complements these results. In the core system augmented with the selected GZ spread, expected default risk rises after a monetary policy tightening, with both the 68\% and 90\% credible bands above zero over the early horizons.

\section{Model Construction}\label{sec:modelconstruction}

Let $y_t$ be an $n\times1$ vector of observed variables for $t=1,\ldots,T$. The data generating process is an SVAR with $\cL$ lags,
\begin{equation}
\label{eq:svar_full}
  y_t'A_0
  =
  c'
  +\sum_{\ell=1}^{\cL}y_{t-\ell}'A_\ell
  +\eps_t',
\end{equation}
where $A_0$ is an invertible $n\times n$ matrix of contemporaneous coefficients, $A_\ell$ is an $n\times n$ lag matrix for $\ell=1,\ldots,\cL$, $c$ is an $n\times1$ vector of intercepts, and $\eps_t\sim\mathcal N(0,I_n)$ is serially independent. The normalization $\Var(\eps_t)=I_n$ fixes the scale of the shocks. The identifying restrictions are imposed on $A_0$ and, when specified, on the lag matrices.

The vector $y_t$ contains the full set of observed candidate variables; thus, its dimension $n$ is large. For a given application, the analysis begins with three elements: a specific economic question, a set of core variables needed to address that question, and a set of identifying restrictions. The unresolved problem is which additional variables from the large candidate set should enter the SVAR information set. Selecting a small system by hand risks excluding relevant information, whereas estimating the entire candidate system is infeasible at the dimensions of big data. The obstacle is computational rather than conceptual.\footnote{A VAR with $n$ variables and $\cL$ lags contains $O(n^2\cL)$ lag coefficients; posterior simulation for a large system is costly even in reduced form. Under a nonconjugate shrinkage prior of the Minnesota type, the key matrix factorization can expand from a $K\times K$ system to a full $Kn\times Kn$ system, where $K$ is the number of predictors per equation. The resulting cost per MCMC draw is $O((Kn)^3)$.} We use the prior of \citet[the Sims--Zha prior hereafter]{SZ98a} to regularize the lag dynamics of each fixed system, but the prior does not determine which candidate variables belong in the model's information set.

We formulate the choice of additional variables through the core equations. These equations address the original economic question. Let $y_t^\dagger$ contain the core variables and a subset of additional variables, and let $\tilde y_t$ contain the remaining candidate variables. To formalize the role of the remaining variables, partition
\[
  y_t
  =
  \bigl((y_t^\dagger)',\tilde y_t'\bigr)',
  \qquad
  c
  =
  \bigl(c_1',\tilde c'\bigr)',
  \qquad
  \eps_t
  =
  \bigl((\eps_t^\dagger)',\tilde\eps_t'\bigr)'.
\]
In this partition, $\eps_t^\dagger$ is the first block of the full-system shock vector $\eps_t$. Conformably partitioning each coefficient matrix, Equation~\eqref{eq:svar_full} becomes
\begin{equation}
\label{eq:svar_block}
\begin{split}
  &\begin{bmatrix}(y_t^\dagger)'&\tilde y_t'\end{bmatrix}
  \begin{bmatrix}
    A_{0,11}&A_{0,12}\\
    A_{0,21}&A_{0,22}
  \end{bmatrix}\\
  &\qquad=
  \begin{bmatrix}c_1'&\tilde c'\end{bmatrix}
  +\sum_{\ell=1}^{\cL}
  \begin{bmatrix}(y_{t-\ell}^\dagger)'&\tilde y_{t-\ell}'\end{bmatrix}
  \begin{bmatrix}
    A_{\ell,11}&A_{\ell,12}\\
    A_{\ell,21}&A_{\ell,22}
  \end{bmatrix}
  +\begin{bmatrix}(\eps_t^\dagger)'&\tilde\eps_t'\end{bmatrix}.
\end{split}
\end{equation}
The first block of \eqref{eq:svar_block} can be written as
\begin{equation}
\label{eq:composite_disturbance}
  (y_t^\dagger)'A_{0,11}
  =
  c_1'
  +\sum_{\ell=1}^{\cL}(y_{t-\ell}^\dagger)'A_{\ell,11}
  +(u_t^\dagger)',
\end{equation}
where
\begin{equation*}
  u_t^\dagger
  =
  \eps_t^\dagger+r_t
\end{equation*}
and
\begin{equation*}
  r_t'
  =
  -\tilde y_t'A_{0,21}
  +\sum_{\ell=1}^{\cL}\tilde y_{t-\ell}'A_{\ell,21}.
\end{equation*}
The composite disturbance vector $u_t^\dagger$ combines the first block of the full-system shocks with the contemporaneous and lagged contributions of variables outside the current system. This block representation provides the basis for model construction. The objective is to construct a system for $y_t^\dagger$ in which the remaining variables do not enter its equations:
\begin{equation}
\label{eq:block_restriction}
  A_{0,21}=0
  \qquad\text{and}\qquad
  A_{\ell,21}=0
  \quad\text{for every }\ell=1,\ldots,\cL.
\end{equation}
When \eqref{eq:block_restriction} does not hold, contemporaneous or lagged values of variables outside the current system enter $u_t^\dagger$ through $r_t$. When it holds, $r_t=0$, and hence $u_t^\dagger=\eps_t^\dagger$. The decomposition yields a model construction criterion: the algorithmic procedure evaluates whether contemporaneous or lagged values of the remaining candidates contain predictive information for the fitted composite disturbance vector within the system of core variables.

The algorithm applies this criterion iteratively to construct the model. At each iteration, the current system is estimated and the fitted composite disturbance vector within the system of core variables is constructed. The model construction step evaluates whether contemporaneous or lagged values of the remaining candidates contain predictive information for at least one component of this vector. It admits the candidates chosen by the criterion and reestimates the enlarged system. Reestimation is essential because adding variables alters the fitted equations and the composite disturbance vector used at the next iteration.

Let
\[
  \cN=\{1,\ldots,n\}
\]
index the full candidate set. We begin with an initial set $\cI^{(0)}\subset\cN$ of $n^{(0)}$ core variables chosen by researchers to address a particular economic question and taken as given by the model construction procedure. Let $y_t^{(0)}$ denote the corresponding subvector and set $k=0$. At iteration $k$, let $\cI^{(k)}$ denote the current index set, let $y_t^{(k)}$ denote the corresponding subvector, and define
\[
  n^{(k)}=|\cI^{(k)}|.
\]
The complement
\[
  \widetilde{\cI}^{(k)}
  =
  \cN\setminus\cI^{(k)}
\]
indexes the remaining candidates. Write $\tilde y_t^{(k)}$ for the corresponding $\tilde n^{(k)}\times1$ vector, where $\tilde n^{(k)}=n-n^{(k)}$. For $j=1,\ldots,\tilde n^{(k)}$, let $\iota_j^{(k)}\in\cN$ denote the index in $y_t$ of the $j$th variable in $\tilde y_t^{(k)}$, such that
\[
  \tilde y_{j,t}^{(k)}
  =
  y_{\iota_j^{(k)},t}.
\]
The distinction between $j$ and $\iota_j^{(k)}$ is needed because the set of remaining candidates changes after each update.

At iteration $k$, reorder $y_t$ and the corresponding rows and columns of each coefficient matrix so that the variables indexed by $\cI^{(k)}$ appear first. For $\ell=0,\ldots,\cL$, let
\[
  A_\ell
  =
  \begin{bmatrix}
    A_{\ell,11}^{(k)}&A_{\ell,12}^{(k)}\\
    A_{\ell,21}^{(k)}&A_{\ell,22}^{(k)}
  \end{bmatrix}
\]
denote the corresponding partition, and let $c_1^{(k)}$ denote the first block of the reordered intercept vector. Define $A_\ell^{(k)}=A_{\ell,11}^{(k)}$, $c^{(k)}=c_1^{(k)}$, and
\begin{equation}
\label{eq:outside_component_k}
  (r_t^{(k)})'
  =
  -(\tilde y_t^{(k)})'A_{0,21}^{(k)}
  +\sum_{\ell=1}^{\cL}(\tilde y_{t-\ell}^{(k)})'A_{\ell,21}^{(k)}.
\end{equation}
Let $\eps_{\cI^{(k)},t}$ denote the subvector of $\eps_t$ ordered in the same manner as $y_t^{(k)}$. The composite disturbance associated with the current system is
\begin{equation}
\label{eq:composite_disturbance_k}
  u_t^{(k)}
  =
  \eps_{\cI^{(k)},t}+r_t^{(k)}.
\end{equation}
The current system is therefore written as
\begin{equation}
\label{eq:svar_k}
  (y_t^{(k)})'A_0^{(k)}
  =
  (c^{(k)})'
  +\sum_{\ell=1}^{\cL}(y_{t-\ell}^{(k)})'A_\ell^{(k)}
  +(u_t^{(k)})',
\end{equation}
subject to the maintained identifying restrictions. Equation~\eqref{eq:svar_k} provides the exact block representation underlying the model-construction procedure. The algorithm operationalizes this representation by estimating the current finite-order SVAR and screening the remaining candidates against its fitted composite disturbances.

Collect the parameters of the current system in
\[
  \theta^{(k)}
  =
  \bigl(c^{(k)},A_0^{(k)},A_1^{(k)},\ldots,A_{\cL}^{(k)}\bigr),
\]
let $\Theta_k$ be the restricted parameter space, and write $Y^{(k)}=\{y_t^{(k)}\}_{t=1}^T$. The posterior mode is
\begin{equation}
\label{eq:mode}
  \widehat\theta^{(k)}
  \in
  \arg\max_{\theta\in\Theta_k}
  \left\{
    \log p\bigl(Y^{(k)}\mid\theta\bigr)
    +\log\pi(\theta)
  \right\},
\end{equation}
where $p\bigl(Y^{(k)}\mid\theta\bigr)$ denotes the likelihood for $\theta$ based on $Y^{(k)}$, and $\pi(\theta)$ denotes the prior density appropriate for the dimension of the current system. If the maximizer is not unique, a prespecified deterministic rule selects one element of the argmax set.

Evaluating the fitted equations at the posterior mode gives
\begin{equation}
\label{eq:fitted_composite_disturbances}
  (\widehat u_t^{(k)})'
  =
  (y_t^{(k)})'\widehat A_0^{(k)}
  -(\widehat c^{(k)})'
  -\sum_{\ell=1}^{\cL}(y_{t-\ell}^{(k)})'\widehat A_\ell^{(k)},
  \qquad
  t=\cL+1,\ldots,T.
\end{equation}
Evaluating the current-system equations at the posterior mode yields the fitted composite disturbance vector $\widehat u_t^{(k)}$. By \eqref{eq:composite_disturbance_k}, the underlying composite disturbance combines the shocks associated with the current equations and the contribution $r_t^{(k)}$ from variables outside the current system. The model construction step screens the remaining candidates against these fitted composite disturbances.

For the model construction step, each variable in $y_t$ is standardized before the first iteration using moments computed from the full data set.\footnote{Section~\ref{sec:data} describes the transformations applied to the raw economic variables before standardization. Each series $y_{i,t}$ is measured either in log levels or in percentage points.} For $i=1,\ldots,n$, define
\[
  \bar y_i
  =
  \frac{1}{T}\sum_{t=1}^T y_{i,t},
  \qquad
  \widehat\sigma_i
  =
  \left[
    \frac{1}{T}\sum_{t=1}^T(y_{i,t}-\bar y_i)^2
  \right]^{1/2},
\]
and
\[
  y_{i,t}^{\mathrm{std}}
  =
  \frac{y_{i,t}-\bar y_i}{\widehat\sigma_i}.
\]
We maintain $\widehat\sigma_i>0$ for every candidate variable. The values $\bar y_i$ and $\widehat\sigma_i$ are computed once and held fixed throughout the model construction procedure. For $j=1,\ldots,\tilde n^{(k)}$, define the standardized remaining candidates by
\[
  \tilde y_{j,t,\mathrm{std}}^{(k)}
  =
  y_{\iota_j^{(k)},t}^{\mathrm{std}},
\]
and stack their contemporaneous values and lags as
\[
  x_{t,\mathrm{std}}^{(k)}
  =
  \bigl(
    (\tilde y_{t,\mathrm{std}}^{(k)})',
    (\tilde y_{t-1,\mathrm{std}}^{(k)})',
    \ldots,
    (\tilde y_{t-\cL,\mathrm{std}}^{(k)})'
  \bigr)'.
\]
This standardization applies only to the variables used in the auxiliary optimization and does not alter the variables used to estimate the current system.

The optimization procedure targets a fixed subset $\cQ=\{q_1,\ldots,q_{n^f}\}\subseteq\{1,\ldots,n^{(0)}\}$ of the core equations, where $n^f=|\cQ|$. Let $\widehat u_{\cQ,t}^{(k)}=(\widehat u_{q_1,t}^{(k)},\ldots,\widehat u_{q_{n^f},t}^{(k)})'$ denote the corresponding subvector of fitted composite disturbances. To evaluate the remaining candidates, relate these disturbances to their standardized contemporaneous and lagged values:
\begin{equation}
\label{eq:screen_reg}
  (\widehat u_{\cQ,t}^{(k)})'
  =
  (d^{(k)})'
  +(\tilde y_{t,\mathrm{std}}^{(k)})'D_0^{(k)}
  +\sum_{\ell=1}^{\cL}(\tilde y_{t-\ell,\mathrm{std}}^{(k)})'D_\ell^{(k)}
  +(\eta_t^{(k)})',
  \qquad
  t=\cL+1,\ldots,T,
\end{equation}
where $d^{(k)}$ and $\eta_t^{(k)}$ are $n^f\times1$ vectors, and each $D_\ell^{(k)}$ is a $\tilde n^{(k)}\times n^f$ matrix. The vector $\eta_t^{(k)}$ is the residual in the auxiliary regression. Rows of $D_\ell^{(k)}$ index the remaining candidates, and columns index the equations in $\cQ$. Equation~\eqref{eq:screen_reg} is used only for model construction. Its coefficients measure whether the remaining candidates have predictive content for the fitted composite disturbances.

Let $T_{\cL}=T-\cL$. For equation $q \in \cQ$, let $d_q\in\mathbb R$ and let
\[
  \beta_q
  =
  \bigl(\beta_{0,q}',\beta_{1,q}',\ldots,\beta_{\cL,q}'\bigr)'
  \in
  \mathbb R^{(\cL+1)\tilde n^{(k)}},
  \qquad
  \beta_{\ell,q}\in\mathbb R^{\tilde n^{(k)}}.
\]
Solve
\begin{equation}
\label{eq:l1_obj}
  \min_{d_q,\beta_q}
  \left\{
    \frac{1}{2T_{\cL}}
    \sum_{t=\cL+1}^{T}
    \left[
      \widehat u_{q,t}^{(k)}
      -d_q
      -(x_{t,\mathrm{std}}^{(k)})'\beta_q
    \right]^2
    +\lambda\norm{\beta_q}_1
  \right\},
\end{equation}
where the intercept $d_q$ is not included in the $\ell_1$ term and $\lambda>0$ is the model complexity parameter selected by the out of sample (OOS) procedure described in Section~\ref{sec:oos}.\footnote{The model-construction procedure incorporates machine learning tools used in modern empirical work, including the Lasso, into an auxiliary screening step. For each value of $\lambda$, the iterative procedure generates a terminal candidate information set. The dependent variables in the auxiliary optimization are the fitted composite disturbances evaluated at the posterior mode, and Bayesian posterior simulation evaluates the resulting terminal systems through OOS forecast loss distributions and selects the model complexity parameter $\lambda$.}

Let $(\widehat d_q^{(k)}(\lambda),\widehat\beta_q^{(k)}(\lambda))$ denote a minimizer of \eqref{eq:l1_obj}, and partition
\[
  \widehat\beta_q^{(k)}(\lambda)
  =
  \bigl(
    (\widehat\beta_{0,q}^{(k)}(\lambda))',
    (\widehat\beta_{1,q}^{(k)}(\lambda))',
    \ldots,
    (\widehat\beta_{\cL,q}^{(k)}(\lambda))'
  \bigr)'.
\]
For $\ell=0,\ldots,\cL$, define
\begin{equation}
\label{eq:coefficient_matrices}
  \widehat D_\ell^{(k)}(\lambda)
  =
  \begin{bmatrix}
    \widehat\beta_{\ell,q_1}^{(k)}(\lambda)&
    \widehat\beta_{\ell,q_2}^{(k)}(\lambda)&
    \cdots&
    \widehat\beta_{\ell,q_{n^f}}^{(k)}(\lambda)
  \end{bmatrix}.
\end{equation}
Thus, $\widehat D_\ell^{(k)}(\lambda)$ is a $\tilde n^{(k)}\times n^f$ matrix whose $j$th column contains the coefficients for equation $q_j$ at lag $\ell$, for $j=1,\ldots,n^f$.

The optimization in \eqref{eq:l1_obj} is used only to construct the information set. For $j=1,\ldots,\tilde n^{(k)}$, collect the coefficients associated with the $j$th variable in $\tilde y_t^{(k)}$ across all equations in $\cQ$ and across its contemporaneous and lagged values:
\begin{equation}
\label{eq:coef_block}
  \widehat G_j^{(k)}(\lambda)
  =
  \bigl[
    \widehat D_{0,j\cdot}^{(k)}(\lambda),
    \widehat D_{1,j\cdot}^{(k)}(\lambda),
    \ldots,
    \widehat D_{\cL,j\cdot}^{(k)}(\lambda)
  \bigr].
\end{equation}
Candidate $\iota_j^{(k)}$ is admitted when at least one coefficient in $\widehat G_j^{(k)}(\lambda)$ is nonzero. Thus,
\begin{equation}
\label{eq:selected_set}
  \widetilde{\cS}^{(k)}(\lambda)
  =
  \left\{
    \iota_j^{(k)}:
    j=1,\ldots,\tilde n^{(k)},\
    \norm{\widehat G_j^{(k)}(\lambda)}_F>0
  \right\}.
\end{equation}
The criterion in \eqref{eq:selected_set} admits a remaining candidate when its contemporaneous value or any of its lags has predictive content for at least one component of the fitted composite disturbance. Once admitted, the variable remains in the system at all subsequent iterations.

After solving \eqref{eq:l1_obj}, update the current index set according to
\begin{equation}
\label{eq:update}
  \cI^{(k+1)}
  =
  \cI^{(k)}
  \cup
  \widetilde{\cS}^{(k)}(\lambda).
\end{equation}
If $\widetilde{\cS}^{(k)}(\lambda)\neq\emptyset$, form the enlarged subvector $y_t^{(k+1)}$, estimate the enlarged system, and compute $\widehat u_t^{(k+1)}$ from \eqref{eq:fitted_composite_disturbances}. Reestimation is essential because admitting a variable alters both the fitted dynamics and the fitted composite disturbances against which the remaining candidates are evaluated.

For a fixed model complexity parameter $\lambda$, the construction terminates at iteration $\bar k(\lambda)$ when
\[
  \widetilde{\cS}^{(\bar k(\lambda))}(\lambda)=\emptyset.
\]
No remaining candidate is admitted for this value of $\lambda$. Denote the terminal index set, variable vector, fitted composite disturbance vector, and dimension by
\[
  \cI^\dagger(\lambda)
  =
  \cI^{(\bar k(\lambda))},
  \quad
  y_t^\dagger(\lambda)
  =
  y_t^{(\bar k(\lambda))},
  \quad
  \widehat u_t^\dagger(\lambda)
  =
  \widehat u_t^{(\bar k(\lambda))},
  \quad
  n^\dagger(\lambda)
  =
  |\cI^\dagger(\lambda)|.
\]
The variables in $y_t^\dagger(\lambda)$ comprise the core variables and the additional variables selected at the model complexity parameter $\lambda$, and together form the terminal information set associated with $\lambda$. Section~\ref{sec:theory} establishes finite termination and uniqueness of the mapping $\lambda\mapsto y_t^\dagger(\lambda)$.

To facilitate implementation, the complete procedure is summarized as follows:
\begin{enumerate}[(1)]
  \item \textbf{Initialize.} Begin with the core index set $\cI^{(0)}$ associated with the economic application, fix a model complexity parameter $\lambda$ from the prespecified grid, and set $k=0$.
  \item \textbf{Estimate.} Estimate the current system \eqref{eq:svar_k} under the prior $\pi$ and compute the posterior mode $\widehat\theta^{(k)}$ in \eqref{eq:mode}.
  \item \textbf{Construct the fitted composite disturbances.} Compute $\widehat u_t^{(k)}$ from \eqref{eq:fitted_composite_disturbances}.
  \item \textbf{Evaluate the remaining candidates.} Solve the auxiliary optimization problem \eqref{eq:l1_obj} for the current system and construct $\widetilde{\cS}^{(k)}(\lambda)$ from \eqref{eq:selected_set}.
  \item \textbf{Update.} Apply \eqref{eq:update}. If $\widetilde{\cS}^{(k)}(\lambda)\neq\emptyset$, increment $k$ and return to step (2). Otherwise, set $\bar k(\lambda)=k$ and retain the terminal variable vector $y_t^\dagger(\lambda)$ and the terminal fitted composite disturbance vector $\widehat u_t^\dagger(\lambda)=\widehat u_t^{(\bar k(\lambda))}$. 
  \item \textbf{Compare terminal systems.} Repeat steps (1)--(5) for each value of the model complexity parameter $\lambda$ in the prespecified grid. Select $\widehat\lambda_\alpha$ using the OOS forecast loss criterion in Section~\ref{sec:oos}, where $\alpha$ is the posterior probability level used in the OOS comparison.
  \item \textbf{Selected variable vector.} The final vector $y_t^\dagger(\widehat\lambda_\alpha)$ contains the core variables and the additional variables selected from the large candidate set for the economic application.
\end{enumerate}

During model construction, $\widehat u_t^{(k)}$ denotes the fitted composite disturbance vector for the system at iteration $k$. For each $\lambda$, the terminal fitted composite disturbance vector is
\[
\widehat u_t^\dagger(\lambda)
=
\widehat u_t^{(\bar k(\lambda))}.
\]
After OOS selection, we write $y_t^\dagger\equiv y_t^\dagger(\widehat\lambda_\alpha)$ and refer to the resulting SVAR as the selected system.

The framework separates three components that are often conflated. The economic question, core variables, and identifying restrictions define the initial economic analysis, whose information set may be incomplete. For each value of the model complexity parameter, the iterative procedure constructs a terminal information set. The OOS forecast loss criterion then selects among the resulting terminal systems. The economic question, model construction, and OOS selection therefore play distinct roles.

\section{Theoretical Properties}\label{sec:theory}

Section~\ref{sec:modelconstruction} defines a sequence of current index sets, posterior mode estimates, fitted composite disturbances, auxiliary optimization problems, and index set updates. In this section, we establish that this construction is well defined, terminates after finitely many updates, and produces a unique terminal variable vector. We also derive the exact stopping condition and establish invariance to changes in the units of the remaining candidates and stability with respect to small changes in the fitted composite disturbances. Throughout this section, $\bar k$ denotes the terminal iteration for a fixed value of the model complexity parameter $\lambda$. All results are conditional on the observed data under the prior $\pi(\theta)$.

\begin{assumption} 
\label{ass:well_defined}
The standard deviations $\widehat\sigma_i$ used to standardize the variables are strictly positive for all $i=1,\ldots,n$. For every index set reached by the procedure, the argmax set in \eqref{eq:mode} is nonempty, and the auxiliary optimization problem in \eqref{eq:l1_obj} has a unique minimizer for every $q\in\cQ$.
\end{assumption}

The uniqueness requirement for the auxiliary minimizer concerns the coefficient vector rather than only its fitted values. This distinction is necessary because the update in \eqref{eq:selected_set} depends on which coefficients are zero. The nonempty argmax set and the prespecified deterministic rule stated after \eqref{eq:mode} determine one posterior mode and hence one fitted composite disturbance vector at every iteration. Under Assumption~\ref{ass:well_defined}, the following two propositions establish convergence and uniqueness.

\begin{proposition} 
\label{prop:convergence}
For each fixed $\lambda$, the update rule generates
\[
  \cI^{(0)}\subseteq\cI^{(1)}\subseteq\cdots\subseteq\cN.
\]
Once a variable enters the current system, it remains in every subsequent system. The terminal iteration satisfies $\bar k\leq n-n^{(0)}$, and the procedure estimates at most $n-n^{(0)}+1$ current systems.
\end{proposition}
\begin{proof}
   See Appendix~\ref{sec:proofs}.
\end{proof}

\begin{proposition} 
\label{prop:uniqueness}
Each fixed value of $\lambda$ determines a unique sequence $\{\cI^{(k)}\}_{k=0}^{\bar k}$, a unique terminal index set $\cI^\dagger(\lambda)$, and a unique terminal variable vector $y_t^\dagger(\lambda)$.
\end{proposition}
\begin{proof}
   See Appendix~\ref{sec:proofs}.
\end{proof}

Proposition~\ref{prop:uniqueness} establishes that the mapping $\lambda\mapsto y_t^\dagger(\lambda)$ is well defined for the OOS comparison in Section~\ref{sec:oos}. It does not imply that terminal vectors associated with different model complexity parameters are nested. Each update changes the fitted composite disturbances used in the next auxiliary optimization, and the construction paths may therefore differ across values of $\lambda$.

For the remaining results, let $T_{\cL}=T-\cL$ and define the stacked fitted composite disturbances and the standardized regressor matrix by
\[
  \widehat U^{(k)}
  =
  \begin{bmatrix}
    (\widehat u_{\cL+1}^{(k)})'\\
    \vdots\\
    (\widehat u_T^{(k)})'
  \end{bmatrix},
  \qquad
  X^{(k)}
  =
  \begin{bmatrix}
    (x_{\cL+1,\mathrm{std}}^{(k)})'\\
    \vdots\\
    (x_{T,\mathrm{std}}^{(k)})'
  \end{bmatrix}.
\]
Thus, $\widehat U^{(k)}\in\mathbb R^{T_{\cL}\times n^{(k)}}$ and $X^{(k)}\in\mathbb R^{T_{\cL}\times p_k}$, where $p_k=(\cL+1)\tilde n^{(k)}$. Let $\widehat{\bm u}_q^{(k)}$ denote column $q$ of $\widehat U^{(k)}$, let $\bm 1$ be the $T_{\cL}\times1$ vector of ones, and define
\[
  M
  =
  I_{T_{\cL}}-\frac{1}{T_{\cL}}\bm 1\bm 1'.
\]
The matrix $M$ centers the dependent variable and regressors. The intercept is part of the original regression. After centering, the auxiliary optimization can be written without the intercept, with no change in the estimated coefficients on the regressors.

\begin{lemma} 
\label{lem:centering}
For a fixed equation $q$ and coefficient vector $\beta_q$, the minimizing intercept in \eqref{eq:l1_obj} is
\[
  d_q(\beta_q)
  =
  \frac{1}{T_{\cL}}\bm 1'
  \left(\widehat{\bm u}_q^{(k)}-X^{(k)}\beta_q\right).
\]
Substitution of this intercept into \eqref{eq:l1_obj} yields the equivalent problem
\begin{equation}
\label{eq:centered_auxiliary}
  \min_{\beta_q}
  \left\{
    \frac{1}{2T_{\cL}}
    \left\|
      M\widehat{\bm u}_q^{(k)}-MX^{(k)}\beta_q
    \right\|_2^2
    +\lambda\|\beta_q\|_1
  \right\}.
\end{equation}
\end{lemma}
\begin{proof}
   See Appendix~\ref{sec:proofs}.
\end{proof}

Lemma~\ref{lem:centering} allows the remaining results to use the centered auxiliary optimization without altering the estimated coefficients that determine which candidates are admitted. By \eqref{eq:selected_set}, no remaining candidate is admitted at iteration $k$ when all coefficient blocks are zero. The next proposition states an equivalent condition based on the standardized remaining candidates and the fitted composite disturbances.

\begin{proposition} 
\label{prop:stopping_condition}
Under Assumption~\ref{ass:well_defined}, no remaining candidate is admitted at iteration $k$, that is,
\[
  \widetilde{\cS}^{(k)}(\lambda)=\emptyset,
\]
if and only if
\begin{equation}
\label{eq:stopping_score}
  \max_{q\in\cQ}.
  \left\|
    \frac{1}{T_{\cL}}(X^{(k)})'M\widehat{\bm u}_q^{(k)}
  \right\|_\infty
  \leq
  \lambda.
\end{equation}
When $\tilde n^{(k)}=0$, the left-hand side of \eqref{eq:stopping_score} is defined as zero.
\end{proposition}
\begin{proof}
   See Appendix~\ref{sec:proofs}.
\end{proof}

Proposition~\ref{prop:stopping_condition} gives economic meaning to the zero coefficients returned by the auxiliary optimization. The left-hand side of \eqref{eq:stopping_score} is the largest absolute average cross-product between any standardized remaining candidate, contemporaneously or at any included lag, and any fitted composite disturbance indexed by $\cQ$. When this quantity does not exceed $\lambda$, no remaining candidate contains enough information under the auxiliary criterion to expand the current system. At the terminal iteration, \eqref{eq:stopping_score} holds because termination is defined by $\widetilde{\cS}^{(\bar k)}(\lambda)=\emptyset$. The procedure has therefore exhausted the information among the remaining candidates that is detectable by this criterion at the given value of $\lambda$. This result connects the optimization stopping rule to the economic purpose of constructing the SVAR information set.

The preceding results establish finite termination, uniqueness, and an exact characterization of stopping. The next two propositions establish that the update rule is invariant to the units of the remaining candidates and stable under sufficiently small changes in the fitted composite disturbances.

\begin{proposition} 
\label{prop:unit_invariance}
Under Assumption~\ref{ass:well_defined}, fix an iteration $k$, the current system, and its fitted composite disturbances. Replace each remaining candidate $y_{i,t}$ by $\alpha_i+d_i y_{i,t}$, where $d_i\neq0$, and recompute its standardization moments. The selected set $\widetilde{\cS}^{(k)}(\lambda)$ is unchanged.
\end{proposition}
\begin{proof}
   See Appendix~\ref{sec:proofs}.
\end{proof}

Proposition~\ref{prop:unit_invariance} ensures that, for a given iteration and fixed fitted composite disturbances, the variables admitted by the procedure do not depend on arbitrary choices of measurement units. The next proposition turns from changes in measurement units to changes in the fitted composite disturbances. It gives a finite sample condition under which sufficiently small changes in these disturbances leave the selected set $\widetilde{\cS}^{(k)}(\lambda)$ unchanged for a given iteration.

For equation $q$, let $\widehat\beta_q^{(k)}$ be the unique solution of \eqref{eq:centered_auxiliary}, and define its nonzero coefficient index set by
\[
  \cA_q^{(k)}
  =
  \left\{h:\widehat\beta_{q,h}^{(k)}\neq0\right\}.
\]
Write $\widetilde X^{(k)}=MX^{(k)}$, $\widetilde{\bm u}_q^{(k)}=M\widehat{\bm u}_q^{(k)}$, and
\[
  \widehat{\bm v}_q^{(k)}
  =
  \widetilde{\bm u}_q^{(k)}
  -
  \widetilde X^{(k)}\widehat\beta_q^{(k)}.
\]

\begin{assumption} 
\label{ass:selection_margins}
For every equation $q$, the matrix $T_{\cL}^{-1}(\widetilde X_{\cA_q^{(k)}}^{(k)})'\widetilde X_{\cA_q^{(k)}}^{(k)}$ is positive definite when $\cA_q^{(k)}$ is nonempty. There exist constants $b_q>0$ and $\delta_q>0$ such that
\[
  \min_{h\in\cA_q^{(k)}}
  |\widehat\beta_{q,h}^{(k)}|
  \geq b_q
\]
when $\cA_q^{(k)}$ is nonempty, and
\[
  \max_{h\notin\cA_q^{(k)}}
  \left|
    \frac{1}{T_{\cL}}
    (\widetilde X_h^{(k)})'\widehat{\bm v}_q^{(k)}
  \right|
  \leq
  \lambda-\delta_q
\]
when the complement of $\cA_q^{(k)}$ is nonempty.
\end{assumption}

The constants $b_q$ and $\delta_q$ define the selection margins for equation $q$. When nonzero estimated coefficients are present, the first margin is a positive lower bound on their absolute values and measures how far the selected coefficients are from zero. When unselected regressors are present, the second margin is a positive lower bound on the gap between $\lambda$ and the largest absolute average cross-product between an unselected regressor and the residual $\widehat{\bm v}_q^{(k)}$. It measures how far an unselected regressor is from meeting the threshold for a nonzero coefficient. The following proposition and corollary establish stability when both relevant selection margins are positive.

\begin{proposition} 
\label{prop:update_stability}
Under Assumptions~\ref{ass:well_defined} and~\ref{ass:selection_margins}, there exists $\epsilon_k>0$ such that replacing each fitted composite disturbance vector $\widehat{\bm u}_q^{(k)}$ by $\widehat{\bm u}_q^{(k)}+h_q$ with $\|h_q\|_2<\epsilon_k$ leaves every nonzero coefficient index set $\cA_q^{(k)}$ unchanged. Hence, the selected variable set $\widetilde{\cS}^{(k)}(\lambda)$ and the next index set $\cI^{(k+1)}$ are unchanged.
\end{proposition}
\begin{proof}
   See Appendix~\ref{sec:proofs}.
\end{proof}

\begin{corollary} 
\label{cor:path_stability}
Suppose Assumptions~\ref{ass:well_defined} and~\ref{ass:selection_margins} hold at every iteration of the construction path for a given $\lambda$. For each iteration $k$ and equation $q$, replace $\widehat{\bm u}_q^{(k)}$ by
\[
  \widehat{\bm u}_q^{(k)}+h_q^{(k)},
  \qquad
  \|h_q^{(k)}\|_2<\epsilon_k,
\]
where $\epsilon_k$ is the stability radius in Proposition~\ref{prop:update_stability}. Then the entire sequence of index sets and the terminal variable vector $y_t^\dagger(\lambda)$ remain unchanged.
\end{corollary}
\begin{proof}
   See Appendix~\ref{sec:proofs}.
\end{proof}

Under Assumption~\ref{ass:well_defined}, the posterior mode exists and the auxiliary minimizers are unique. The prespecified deterministic rule stated after \eqref{eq:mode} selects one posterior mode and hence determines one fitted composite disturbance vector at every iteration. Proposition~\ref{prop:convergence} establishes that the procedure terminates after finitely many iterations, and Proposition~\ref{prop:uniqueness} establishes that each value of $\lambda$ determines a unique construction path and terminal variable vector. Proposition~\ref{prop:stopping_condition} gives economic meaning to termination: no remaining candidate contains enough information under the auxiliary criterion to expand the current system. Proposition~\ref{prop:unit_invariance} establishes that the admission decision at a given iteration does not depend on the units of measurement of the remaining candidates.

Under Assumption~\ref{ass:selection_margins}, Proposition~\ref{prop:update_stability} and Corollary~\ref{cor:path_stability} have practical importance for implementing the algorithm. The fitted composite disturbances are computed from posterior mode estimates, which may vary slightly with numerical tolerances, rounding, or optimization routines. Proposition~\ref{prop:update_stability} establishes that such small changes do not alter the variables admitted at a given iteration when the selection margins hold.

Corollary~\ref{cor:path_stability} establishes the corresponding result for the full procedure. A small change at an early iteration could otherwise alter the next system, later fitted disturbances, and the final selected variables. The corollary rules out this propagation when the perturbations remain within the stability radius at each iteration. The complete construction path and the terminal variable vector then remain unchanged. These results establish the local reproducibility and practical reliability of the algorithm, an important property in applied work where numerical tolerances are unavoidable.

\section{OOS Selection of the Model}\label{sec:oos}

In this section, we use the OOS forecast performance of the core variables to select the model complexity parameter $\widehat\lambda_\alpha$ from the candidate values of $\lambda$. Given the terminal map $\lambda\mapsto y_t^\dagger(\lambda)$ established in Sections~\ref{sec:modelconstruction} and~\ref{sec:theory}, each value of $\lambda$ produces a unique terminal variable vector. The model complexity parameter governs admission in the auxiliary optimization; each admission alters the system and the fitted composite disturbances at later iterations. Terminal system size varies with $\lambda$. Our model-construction methodology therefore does not favor small or large systems \emph{a priori}.

\subsection{Terminal systems and the timing convention}
\label{subsec:oos_timing}

Let $\Lambda$ be the finite prespecified grid of model complexity parameter values. For each $\lambda\in\Lambda$, Sections~\ref{sec:modelconstruction} and~\ref{sec:theory} deliver the terminal index set $\cI^\dagger(\lambda)$, the terminal variable vector $y_t^\dagger(\lambda)$, and its dimension $n^\dagger(\lambda)$. Proposition~\ref{prop:uniqueness} ensures that these objects are uniquely determined. The terminal systems generated by different values of $\lambda$ need not be nested, and their sizes need not be monotone in $\lambda$. We compare them by the OOS forecast loss of the core variables.

Every terminal system is constructed from the common estimation sample through $T$. The validation period does not enter the construction of any terminal index set. \citet{GuKellyXiu2020} use forecast performance over a validation sample to tune their hyperparameters. We follow the same training--validation logic to select our model complexity parameter $\widehat\lambda_\alpha$ and hence the terminal system.\footnote{Unlike an expanding-window exercise that re-estimates the model as new data arrive \citep[for example]{GiannoneLenzaPrimiceri2021}, we obtain the posterior draws from data through $T$ and hold them fixed throughout the validation period. The forecasts are updated recursively as new observations become available. Our purpose is to compare terminal systems, not to mimic a real-time forecaster who re-estimates the model as new data arrive.}

Let $\theta_\lambda$ denote the parameters of the terminal system indexed by $\lambda$, and let $Y_T^\dagger(\lambda)$ denote its observations through $T$. The posterior distribution
\[
  p_\lambda\!\left(\theta_\lambda\mid Y_T^\dagger(\lambda)\right)
\]
is computed for each terminal system. Let $s=1,\ldots,S$ index posterior draws $\theta_\lambda^{(s,T)}$. For each posterior draw, we construct the OOS forecast path recursively using the data available in the validation period and compute the corresponding OOS forecast loss.

In our paper, ``out of sample (OOS)'' refers to the forecasts. The posterior distribution and the terminal system remain fixed after $T$. At each forecast origin, the forecast uses the data available at that time. Once the realization is observed, it is used to compute the forecast error and enters the information set for the next forecast. The validation period serves one purpose: the resulting OOS forecast losses select $\widehat\lambda_\alpha$. 

Let $T+1,\ldots,T+R$ be the validation period, let $H$ be the maximum forecast horizon, and let $\mathcal O\subseteq\{T,\ldots,T+R-H\}$ be a prespecified set of $F=|\mathcal O|$ forecast origins. At an origin $\tau\in\mathcal O$, forecasts condition on observations through $\tau$, but the posterior distribution remains based on the estimation sample ending at $T$. The validation observations update the conditioning values in the forecast recursion and provide the realizations used to compute forecast errors. They do not change the posterior distribution. Every terminal system is evaluated at the same forecast origins, horizons, and realizations.

\subsection{OOS forecast loss for the core variables}
\label{subsec:oos_loss}

Let $J=n^{(0)}$ be the number of core variables, placed first in $y_t$, and let $y_{j,t}$ denote core variable $j$ for $j=1,\ldots,J$. For a core variable represented in log levels, the target is its growth rate. For a core variable represented in levels, the target remains in levels. Let $\mathcal J_{\log}\subseteq\{1,\ldots,J\}$ index the core variables represented in log levels and define
\begin{equation}
\label{eq:oos_target}
  z_{j,t}
  =
  \begin{cases}
    y_{j,t}-y_{j,t-1}, & j\in\mathcal J_{\log},\\
    y_{j,t}, & j\notin\mathcal J_{\log}.
  \end{cases}
\end{equation}
The same transformation is applied to realized values and forecasts for every terminal system. Forecast errors are therefore evaluated in prespecified economic units: growth rates for variables represented in log levels and levels for all other variables, including percentage points for rates and spreads.

For posterior draw $s$, let $y^{(s,T)}_{j,\tau+h\mid\tau}(\lambda)$ denote the $h$ step ahead forecast of core variable $j$ from the terminal system indexed by $\lambda$. Set $y^{(s,T)}_{j,\tau\mid\tau}(\lambda)=y_{j,\tau}$ and define
\begin{equation*}
  z^{(s,T)}_{j,\tau+h\mid\tau}(\lambda)
  =
  \begin{cases}
    y^{(s,T)}_{j,\tau+h\mid\tau}(\lambda)-y^{(s,T)}_{j,\tau+h-1\mid\tau}(\lambda), & j\in\mathcal J_{\log},\\
    y^{(s,T)}_{j,\tau+h\mid\tau}(\lambda), & j\notin\mathcal J_{\log}.
  \end{cases}
\end{equation*}
Define the corresponding realization by
\begin{equation*}
  z_{j,\tau+h}
  =
  \begin{cases}
    y_{j,\tau+h}-y_{j,\tau+h-1}, & j\in\mathcal J_{\log},\\
    y_{j,\tau+h}, & j\notin\mathcal J_{\log}.
  \end{cases}
\end{equation*}
The squared forecast error is
\begin{equation*}
  \mathrm{SE}^{(s,T)}_{j,h,\tau}(\lambda)
  =
  \left[
    z^{(s,T)}_{j,\tau+h\mid\tau}(\lambda)-z_{j,\tau+h}
  \right]^2.
\end{equation*}
The absolute forecast error is
\begin{equation*}
  \mathrm{AE}^{(s,T)}_{j,h,\tau}(\lambda)
  =
  \left|
    z^{(s,T)}_{j,\tau+h\mid\tau}(\lambda)-z_{j,\tau+h}
  \right|.
\end{equation*}

The OOS forecast loss for posterior draw $s$ averages the squared forecast errors over core variables, horizons, and forecast origins:
\begin{equation}
\label{eq:oos_loss}
  \cL^{(s,T)}_{\mathrm{OOS}}(\lambda)
  =
  \frac{1}{JHF}
  \sum_{j=1}^{J}
  \sum_{h=1}^{H}
  \sum_{\tau\in\mathcal O}
  \mathrm{SE}^{(s,T)}_{j,h,\tau}(\lambda).
\end{equation}
For each $\lambda$, the draws $\{\cL^{(s,T)}_{\mathrm{OOS}}(\lambda)\}_{s=1}^{S}$ form the empirical posterior distribution of the OOS forecast loss. We also compute the OOS forecast loss using absolute forecast errors. In that case, we replace $\mathrm{SE}^{(s,T)}_{j,h,\tau}(\lambda)$ with $\mathrm{AE}^{(s,T)}_{j,h,\tau}(\lambda)$ in \eqref{eq:oos_loss}.

\subsection{The reference system and posterior OOS forecast loss differences}\label{subsec:oos_diff}

The comparison begins with a reference terminal system. For each $\lambda\in\Lambda$, define the posterior mean OOS forecast loss by
\begin{equation*}
  \overline{\cL}^{(T)}_{\mathrm{OOS}}(\lambda)
  =
  \frac{1}{S}
  \sum_{s=1}^{S}
  \cL^{(s,T)}_{\mathrm{OOS}}(\lambda).
\end{equation*}
The reference model complexity parameter satisfies
\begin{equation}
\label{eq:benchmark}
  \lambda^\ast
  \in
  \arg\min_{\lambda\in\Lambda}
  \overline{\cL}^{(T)}_{\mathrm{OOS}}(\lambda).
\end{equation}
A prespecified ordering of $\Lambda$ resolves ties in \eqref{eq:benchmark}. The posterior mean of the OOS forecast loss determines only the reference system.\footnote{The posterior mean is the Bayesian expected-loss criterion for choosing the reference system, provided that the posterior expectation exists. This criterion applies to OOS forecast losses based on either squared or absolute forecast errors.} The final choice of $\lambda$ uses the posterior distribution of OOS forecast loss differences as follows.

For $\lambda\neq\lambda^\ast$, let $s=1,\ldots,S$ index independent pairs of posterior draws, with one draw from the terminal system indexed by $\lambda$ and one draw from the reference system indexed by $\lambda^\ast$. The independent pairing defines a prespecified product-posterior comparison convention for models with different parameter spaces. It is not implied by a joint posterior over the models. Define
\begin{equation}
\label{eq:loss_diff}
  \Delta^{(s,T)}(\lambda)
  =
  \cL^{(s,T)}_{\mathrm{OOS}}(\lambda)
  -
  \cL^{(s,T)}_{\mathrm{OOS}}(\lambda^\ast),
  \qquad
  s=1,\ldots,S.
\end{equation}
The paired differences
\[
  \left\{
    \Delta^{(s,T)}(\lambda)
  \right\}_{s=1}^{S}
\]
form the empirical posterior distribution of the OOS forecast loss difference. For every $\lambda\in\Lambda$ and for each loss, let $\widehat F_\lambda$ be the empirical distribution function of the $S$ posterior OOS forecast loss differences. Define
\[
  q_{(1-\alpha)/2}(\lambda)
  =
  \inf\left\{
    x:
    \widehat F_\lambda(x)\geq\frac{1-\alpha}{2}
  \right\},
  \;
  q_{(1+\alpha)/2}(\lambda)
  =
  \inf\left\{
    x:
    \widehat F_\lambda(x)\geq\frac{1+\alpha}{2}
  \right\}
\]
for $\alpha\in(0,1)$. Write
\[
  \mathrm{PI}_\alpha[\Delta(\lambda)]
  =
  \left[
    q_{(1-\alpha)/2}(\lambda),
    q_{(1+\alpha)/2}(\lambda)
  \right]
\]
for the equal-tailed posterior probability interval of level $\alpha$. The final selection rule compares terminal systems through the posterior distribution of their OOS forecast loss differences relative to $\lambda^\ast$.

The resulting probability interval is conditional on the terminal construction paths, the loss-specific reference system, the realized validation outcomes, and the independent product coupling. We use it as an algorithmic posterior comparison device, not as a posterior probability over model space.

\subsection{Posterior comparison and the selection rule}
\label{subsec:oos_rule}

At probability level $\alpha$, define the sets of terminal systems that are credibly better than the reference system and credibly indistinguishable from it in forecast performance by
\[
  \Lambda_\alpha^{-}
  =
  \left\{
    \lambda\in\Lambda\setminus\{\lambda^\ast\}:
    q_{(1+\alpha)/2}(\lambda)<0
  \right\},
\]
\[
  \Lambda_\alpha^{0}
  =
  \{\lambda^\ast\}
  \cup
  \left\{
    \lambda\in\Lambda\setminus\{\lambda^\ast\}:
    q_{(1-\alpha)/2}(\lambda)
    \leq 0 \leq
    q_{(1+\alpha)/2}(\lambda)
  \right\}.
\]
A system in $\Lambda_\alpha^{-}$ is classified as credibly better than the reference system at probability level $\alpha$. A system in $\Lambda_\alpha^{0}$ is classified as credibly indistinguishable from the reference because its equal-tailed interval contains zero. A system whose interval lies above zero is classified as credibly worse and is not retained.

The comparison set gives priority to terminal systems that are credibly better than the reference system. 
\begin{equation}
\label{eq:equiv_set}
  \Lambda_\alpha
  =
  \begin{cases}
    \Lambda_\alpha^{-}, & \Lambda_\alpha^{-}\neq\emptyset,\\
    \Lambda_\alpha^{0}, & \Lambda_\alpha^{-}=\emptyset.
  \end{cases}
\end{equation}
The explicit inclusion of $\lambda^\ast$ in $\Lambda_\alpha^{0}$ reflects that the OOS forecast loss difference between the reference system and itself is identically zero. By construction, therefore, $\Lambda_\alpha$ is nonempty.

The credibly-indistinguishable classification depends on the loss function. Squared error places greater weight on large forecast misses, whereas absolute error is less sensitive to them \citep{Gneiting2011}. We therefore apply the preceding construction separately under squared and absolute forecast errors. This application yields the comparison sets $\Lambda_{\alpha}^{\mathrm{SE}}$ and $\Lambda_{\alpha}^{\mathrm{AE}}$, each relative to its loss-specific reference system. Define the combined comparison set by
\[
  \Lambda_{\alpha}^{\cup}
  =
  \Lambda_{\alpha}^{\mathrm{SE}}
  \cup
  \Lambda_{\alpha}^{\mathrm{AE}}.
\]
The two losses provide alternative admissibility criteria. We take their union because our objective is to retain the largest terminal system that satisfies the posterior comparison criterion under either loss.

Define
\[
  n_{\alpha,\cup}^{\max}
  =
  \max_{\lambda\in\Lambda_{\alpha}^{\cup}}
  n^\dagger(\lambda),
  \qquad
  \Lambda_{\alpha,\cup}^{\max}
  =
  \left\{
    \lambda\in\Lambda_{\alpha}^{\cup}:
    n^\dagger(\lambda)=n_{\alpha,\cup}^{\max}
  \right\}.
\]
Let $\prec$ be a prespecified strict total order on the grid $\Lambda$. The selected model complexity parameter is
\begin{equation}
\label{eq:selection_rule}
  \widehat\lambda_\alpha
  =
  \min_{\prec}\Lambda_{\alpha,\cup}^{\max}.
\end{equation}
If $\Lambda_{\alpha,\cup}^{\max}$ is a singleton, its sole member is selected. If several systems share the largest dimension, the prespecified ordering $\prec$ resolves the tie. Because terminal systems need not be nested, the largest system in this rule is the system with the largest dimension, not a system that contains every other retained system.

We set the probability level to $\alpha=0.68$ in the empirical analysis, following the reporting convention of \citet{SimsZha1999}. The 68\% intervals provide sharper discrimination than the 90\% and 95\% intervals.\footnote{We report the 90\% and 95\% intervals as diagnostics. At these wider levels, the sign of the loss difference is unresolved for nearly all terminal systems in our applications.} Our zero-containing 68\% comparison set parallels the band-based logic of the one-standard-error rule of cross validation. That rule selects the most parsimonious model whose estimated prediction error lies within one standard error of the minimum \citep{BreimanFriedmanOlshenStone1984,HastieTibshiraniFriedman2009}. Our rule instead gives priority to credible improvement and otherwise retains the largest system credibly indistinguishable from the reference, because our objective is information retention rather than predictive parsimony.

\subsection{Existence and uniqueness of the OOS selection}
\label{sec:oos_theory}

This subsection establishes that the OOS forecast loss criterion returns exactly one model complexity parameter, and that the mapping from the data to the selected system exists and is unique.

\begin{assumption} 
\label{ass:oos_regularity}
The posterior probability level satisfies $\alpha\in(0,1)$, and the following conditions hold.
\begin{enumerate}[(i)]
  \item The grid $\Lambda$ is finite and nonempty, and is endowed with a strict total order $\prec$.
  \item For every $\lambda\in\Lambda$, the terminal variable vector $y_t^\dagger(\lambda)$ is unique, so the terminal dimension $n^\dagger(\lambda)=|\cI^\dagger(\lambda)|$ is well defined.
  \item The posterior of each terminal system is based on data through $T$. The posterior draws are maintained over the validation period, and each squared-error or absolute-error OOS forecast loss in \eqref{eq:oos_loss} has a finite value for every draw $s=1,\ldots,S$ and every $\lambda\in\Lambda$.
   \item For every $\lambda\in\Lambda$ and for each loss, the $S$ posterior draws uniquely determine the empirical quantiles $q_{(1-\alpha)/2}(\lambda)$ and $q_{(1+\alpha)/2}(\lambda)$ used as the interval endpoints.  
\end{enumerate}
\end{assumption}

Condition~(ii) is part of Proposition~\ref{prop:uniqueness}. Conditions~(i), (iii), and~(iv) are part of the implemented procedure: a finite grid with a fixed order, a fixed set of posterior draws, finite OOS forecast loss values, and a fixed rule for reading quantiles off those draws.

The OOS selection is constructed one loss at a time. The next lemma establishes that the entire per-loss construction---the reference system, the credibly better set, the credibly indistinguishable set, and the retained set $\Lambda_\alpha$ defined in \eqref{eq:equiv_set}---is well defined, unique, and nonempty. Proposition~\ref{prop:welldefined} then applies this construction to the two losses and combines the resulting retained sets.

\begin{lemma} 
\label{lem:perloss}
Let $g\in\{\mathrm{SE},\mathrm{AE}\}$ denote either the squared or the absolute forecast loss. Let $\lambda^{\ast}_{g}$, $\Lambda_{\alpha}^{g,-}$, and $\Lambda_{\alpha}^{g,0}$ denote the corresponding loss-specific objects, and let $\Lambda_{\alpha}^{g}$ denote the retained set obtained by applying the rule in \eqref{eq:equiv_set} to these two sets. Under Assumption~\ref{ass:oos_regularity}, $\lambda^{\ast}_{g}$ and the sets $\Lambda_{\alpha}^{g,-}$, $\Lambda_{\alpha}^{g,0}$, and $\Lambda_{\alpha}^{g}$ exist and are unique. Moreover, $\lambda^{\ast}_{g}\in\Lambda_{\alpha}^{g,0}$, and $\Lambda_{\alpha}^{g}$ is nonempty.
\end{lemma}
\begin{proof}
   See Appendix~\ref{sec:proofs}.
\end{proof}

\begin{proposition} 
\label{prop:welldefined}
Under Assumption~\ref{ass:oos_regularity}, the loss-specific retained sets $\Lambda_{\alpha}^{\mathrm{SE}}$ and $\Lambda_{\alpha}^{\mathrm{AE}}$, the combined retained set $\Lambda_{\alpha}^{\cup}$, the maximal dimension $n_{\alpha,\cup}^{\max}$, the maximal set $\Lambda_{\alpha,\cup}^{\max}$, and the selected model complexity parameter $\widehat\lambda_{\alpha}=\min_{\prec}\Lambda_{\alpha,\cup}^{\max}$ in \eqref{eq:selection_rule} exist and are unique. Consequently, the selected terminal system $y_t^\dagger(\widehat\lambda_{\alpha})$ exists and is uniquely determined by the data and the prespecified procedure.
\end{proposition}
\begin{proof}
   See Appendix~\ref{sec:proofs}.
\end{proof}

Proposition~\ref{prop:welldefined} completes the link between model construction and OOS selection. For each loss, the reference system minimizes the posterior mean OOS forecast loss. The posterior distribution of OOS forecast loss differences then identifies terminal systems that are credibly better than the reference system or credibly distinguishable from it. The two loss-specific retained sets are combined, and the selection rule retains the terminal system with the largest dimension in their union. Together with Propositions~\ref{prop:convergence} and~\ref{prop:uniqueness}, this result ensures that each model complexity parameter produces one terminal system and that the OOS selection returns one model complexity parameter and hence one selected terminal system.

We next apply the methodology to two of the most studied questions in the SVAR literature: the effects of credit shocks and the effects of monetary policy shocks. We demonstrate the methodology under two classic identification methods: recursive identification in the household credit application and proxy identification in the monetary policy application. The reported impulse-response bands are conditional on the system used to estimate each response. They incorporate posterior parameter uncertainty within that system but do not average across alternative information sets or reporting systems.

\section{Household Credit and Housing Production}\label{sec:household}

This application asks whether the boom-and-decline response of output to a household credit shock survives the systematic construction of the SVAR information set. MSV find in a panel of 30 countries that output initially rises after a household credit shock and later declines. Their proxy SVAR gives the shock a credit supply interpretation. Using a 10-variable model with monthly U.S. data, BPSS confirm the same boom-and-decline finding but attribute the later decline to endogenous monetary policy tightening.

The two papers therefore agree on the central impulse response but differ in their interpretation of the later decline. BPSS argue that the interpretation of a household credit shock depends on the information set. Their larger system includes interest rates, money, commodity prices, and credit spreads, but its composition is chosen by the researcher. This leaves open the model construction question: which additional variables belong in the SVAR?

We address this question under MSV's recursive identification. Output, business credit, and household credit form the core system because they are common to the MSV and BPSS analyses. Starting from this core, the procedure constructs terminal systems from a large candidate set, and OOS forecast performance selects the final information set. The selected system adds several measures of housing production and separates a household credit shock from a housing production shock.

\subsection{Data}\label{sec:data}

Our monthly candidate data set contains about 145 series used in this application and in the monetary policy transmission application discussed in Section~\ref{sec:mp}. Most series are obtained directly from Federal Reserve Economic Data (FRED), maintained by the Federal Reserve Bank of St. Louis. The remaining series are obtained from the public website of the Federal Reserve Board. We retain only series available over the full sample required by each application. This requirement keeps the estimation sample fixed as variables enter the model during construction. For both applications, we report the FRED series ID whenever one is available.

The data transformations follow standard practice in macroeconomic VAR analysis and place the variables on scales suitable for estimation and economic interpretation. Variables measured as indexes, counts, flows, or stocks enter in log levels. These variables include incomes, production, employment, price indexes, credit, money, housing permits, and housing starts. Log transformations reduce differences in scale across these series, and their impulse responses can be interpreted approximately as percentage changes. Interest rates, term and credit spreads, unemployment rates, capacity utilization rates, and other variables measured as rates or shares enter in percentage points. Their impulse responses retain those units.

\subsection{The economic question and the core evidence}\label{sec:core}

MSV begin with a 3-variable recursive VAR in log real GDP, nonfinancial firm debt in percent of lagged GDP, and household debt in percent of lagged GDP. The recursively identified household credit shock raises household debt. GDP initially rises and then falls, even though household debt remains above its pre-shock level throughout the ten-year horizon. MSV state that this first exercise describes the dynamic relation and is not intended to establish a causal relationship. They then use a low mortgage spread as an external instrument for a household credit supply shock. The proxy SVAR produces the same boom and decline of output. The central finding therefore does not rest on the external instrument. It is already present under the recursive identification, which assigns the rise and fall of output to a household credit shock.

BPSS examine whether the same finding survives in a richer monthly U.S. system. Their 10-variable SVAR includes industrial production, prices, household and business credit, money, the federal funds rate, commodity prices, a term spread, and two measures of financial stress. They identify the shocks through heteroskedasticity. Their household credit shock produces a small initial rise in industrial production and prices, followed by a later decline in industrial production and an increase in the federal funds rate. Neither measure of financial stress widens persistently. BPSS therefore attribute the later output decline to endogenous monetary tightening rather than to an inevitable reversal of credit supply or to financial stress. This interpretation does not overturn MSV's credit supply interpretation. It identifies endogenous monetary policy as an additional channel behind the later decline.

Figure~\ref{fig:hhc_qje} replicates the MSV evidence for the United States using different data, a different sample, and monthly rather than annual frequency. We use monthly data from 1960M1 through 2012M12, which match the calendar span of MSV's annual panel, 60 monthly lags to preserve the five-year lag span of their VAR, and the monthly Sims--Zha prior.\footnote{MSV measure outstanding household and nonfinancial firm credit using annual data from the Bank for International Settlements, including bank and nonbank loans and debt securities, and scale these stocks by GDP. BPSS use monthly commercial bank loan stocks in log levels: consumer and real estate loans for household credit (HHC) and commercial and industrial loans for business credit (BC). Their HHC series includes commercial real estate loans that cannot be separated over the full sample. Our HHC and BC follow BPSS. The MSV measures provide broader sectoral coverage, while the BPSS measures provide the monthly frequency required for our application. Both measure credit stocks rather than flows of new lending.} We follow the BPSS variable definitions because our analysis compares the variables selected by our procedure with the variables chosen by BPSS. 

Household credit remains above its pre-shock level throughout the forecast horizon. The posterior median of industrial production rises for about two years and then declines. Although the data, sample, frequency, and credit measures differ from those in MSV's VAR analysis, the same initial output boom and later decline continue to hold. This is the core evidence that survives BPSS's larger system despite their measured challenge to MSV's interpretation. The remaining question is whether both information sets omit a real activity margin that moves with household credit and helps generate the initial output boom.

\subsection{Identification and model construction}\label{sec:bb_identification}

We maintain the recursive identification used by MSV throughout our model construction and OOS selection. The three core variables are ordered as industrial production, business credit, and household credit, corresponding to MSV's ordering of real GDP, nonfinancial firm debt, and household debt. The household credit shock is the third shock identified by the recursive ordering. The HHC equation may respond contemporaneously to IP and BC, but the HHC shock cannot affect IP or BC on impact.\footnote{For the monthly systems, we use the BPSS sample, 1973M1--2015M6, and their lag length of 10 months. We use updated vintages of the underlying data, truncated to the same sample period. We also use the monthly Sims--Zha prior, which is closely related to the prior used by BPSS. This common setup allows us to assess whether the variables chosen by our methodology overlap with those in BPSS.}

All variables admitted beyond the core are ordered after IP, BC, and HHC. The expanded systems therefore preserve MSV's recursive ordering of the three core variables. Shocks to the added variables cannot affect the core variables on impact, although their lags enter the core equations and alter subsequent dynamics. Conversely, the core shocks can affect the added variables on impact. Adding variables also changes the fitted core equations, innovations, and impulse responses. These changes show the importance of model construction. Model construction evaluates the remaining candidates against the disturbances of output, business credit, and household credit, while OOS selection evaluates forecasts of these same variables.

At each iteration, the procedure in Section~\ref{sec:modelconstruction} estimates the current system and constructs the fitted composite disturbances for the three core equations. A remaining candidate is admitted when its contemporaneous value or included lags contain predictive information for at least one of these disturbances. Admission does not alter the maintained contemporaneous restrictions: admitted variables remain ordered after the three core variables, while their lags are allowed to enter the core equations and affect subsequent dynamics. The enlarged system is then reestimated because the fitted core disturbances change when the information set expands. For each value of $\lambda$, this process continues until no remaining candidate is admitted. The OOS criterion in Section~\ref{sec:oos} compares the resulting terminal systems using forecasts of the same three core variables (i.e., $J=n^{(0)}=3$ in \eqref{eq:oos_loss}).

\subsection{The selected system}\label{sec:bb_selection}

The procedure selects a 13-variable terminal system at $\widehat\lambda_\alpha=0.1129$ (Table~\ref{tab:mp_model11_variables}). In addition to IP, BC, and HHC, the system contains four housing variables: housing starts in the South and new private housing permits for the Northeast, the West, and the nation. It contains three labor market indicators: average weekly hours in goods producing industries, average weekly hours in manufacturing, and initial claims for unemployment insurance. It also contains three financial indicators: the one year Treasury rate relative to the federal funds rate, the excess bond premium, and the GZ credit spread.\footnote{The construction path begins at iteration 0 with 3 variables: \{\texttt{IP}, \texttt{BC}, \texttt{HHC}\}. Iteration 1 adds \{\texttt{HRS-G}, \texttt{HS-S}, \texttt{PER-NE}, \texttt{T1Y-FF}, \texttt{EBP}, \texttt{GZS}\}, producing a 9-variable system. Iteration 2 adds \{\texttt{HRS-M}, \texttt{PER}\}, producing an 11-variable system. Iteration 3 adds \{\texttt{UI}, \texttt{PER-W}\}, producing the selected 13-variable system. See the detailed variable descriptions in Table~\ref{tab:mp_model11_variables}. Computation used MATLAB R2023b under Windows 11 Enterprise on an Intel Core Ultra 7 165U processor at 2.10 GHz with 32 GB of RAM. MATLAB had access to 12 CPU cores, although computation was serial. Model construction across all values of $\lambda$ took 75 seconds. OOS selection based on 2,000 posterior draws took 1,900 seconds, or about 32 minutes. Increasing the number of posterior draws does not change the results.}

The selected system includes variables that capture two dimensions emphasized by BPSS. The one year Treasury spread measures the relation between a market rate and the policy rate, while the excess bond premium and the GZ spread measure corporate credit conditions. The selected system differs more sharply from BPSS on the real side. It adds three labor market indicators and four housing variables. The housing variables form the largest coherent block among the selected additions and provide the new economic margin. Neither MSV nor BPSS includes housing activity in its SVAR.

No housing block is imposed in model construction. Each housing series competes with all other candidates on the same terms. The procedure chooses housing starts in the South, but not national starts or starts in other regions. Among the housing starts series, southern starts are the one retained by the admission criterion at the selected value of $\lambda$. By contrast, the selected permit series cover the nation, Northeast, and West, indicating that housing adjustments differ across regions and over time. The geographic labels carry no separate structural interpretation. Together, these series measure several distinct margins of the housing production cycle and contribute predictive information for the fitted composite disturbances in the three core equations.

The selection of housing starts and permits is consistent with a broader literature on housing and credit allocation. \citet{Leamer2007,Leamer2015} place residential construction and housing quantities at the center of the U.S. business cycle. \citet{RognlieShleiferSimsek2018} develop a model in which overbuilding of durable capital such as housing can generate a subsequent demand-driven recession. \citet{MullerVerner2024} show more broadly that the macroeconomic consequences of credit expansions depend on the sector receiving the credit, as nontradable sectors rely more heavily on real-estate-secured borrowing. This literature supports the economic relevance of the housing production margin discovered by our procedure.

\subsection{Household credit is leverage without production}\label{sec:bb_hhc}

Figure~\ref{fig:hhc_aerchol_hhc} reports the responses to a household credit shock in the selected system. Household credit rises on impact, and its posterior median remains positive throughout the horizon. The 90\% credible band excludes zero for most of the horizon. Output does not follow. Its posterior median rises only trivially on impact, and this increase is negligible relative to its later decline. The household credit expansion therefore occurs without a material expansion in output.

The force of this result lies in the central response, not merely in the credible bands. MSV and BPSS interpret the boom and subsequent decline from the estimated response itself. If the selected system merely widened the bands while leaving the initial rise in place, their economic interpretation would survive with less precision. It does not. The posterior median shows no economically meaningful rise in output. This does more than weaken the evidence for an output boom. It removes the boom from the estimated response. The selected information set therefore changes the economic interpretation, not just its statistical precision. Household credit expands without an output boom.

The housing activity responses explain the change in interpretation. The posterior medians of national permits, western permits, northeastern permits, and southern housing starts are negative. The 90\% credible bands for national permits, western permits, and southern housing starts lie below zero for substantial portions of the horizon. The response of northeastern permits is less precise, but its posterior median declines. The GZ spread rises only briefly before returning toward zero. The household credit shock therefore produces a contraction, rather than an expansion, in housing production and no persistent deterioration in corporate credit conditions.

The increase in household credit is instead accompanied by lower residential construction. The aggregate data do not identify the exact use of the additional credit. It may reflect purchases of existing homes, refinancing, home equity extraction, consumer borrowing, or other balance sheet adjustments. These activities can increase outstanding household credit without financing new residential construction---more leverage without more production. Household credit rises, housing production falls, and output shows no meaningful initial increase.

The later output decline is imprecisely estimated in all three exercises. In our selected system, both the 68\% and 90\% credible bands for the IP response contain zero. In BPSS, the corresponding bands are somewhat tighter but also contain zero during the later decline. In MSV, the 95\% confidence interval contains zero over almost the entire ten-year horizon in which the GDP estimate turns negative. The later decline is therefore a feature of the central responses, not a precisely estimated effect.

This imprecision limits what the impulse responses can establish about the source of the later decline. The deleveraging and weak demand interpretation emphasized by MSV and the endogenous monetary tightening interpretation emphasized by BPSS remain consistent with our evidence. Our interpretation does not require output to decline when housing production falls. A decline in housing production works against an output boom, but other components of activity may offset that effect. We therefore remain agnostic about the sign of the later output response. Our main result concerns the initial response: once the selected information set accounts for housing production, a household credit shock no longer produces an output boom in the point estimate, measured by the posterior median.

\subsection{Output rises with housing production}\label{sec:bb_housing}

We refer to the recursive innovation associated with national building permits as the housing production shock. The label summarizes its empirical content and does not assign a primitive cause. National permits rise on impact. Permits in the Northeast and West and housing starts in the South rise with them (Figure~\ref{fig:hhc_aerchol_housing}). These joint responses support interpreting the shock as a broad expansion in prospective and current residential construction.

Industrial production rises after the housing production shock. Its 90\% credible band remains above zero for about two years, and its 68\% band remains above zero longer. Household credit also rises persistently, with both bands above zero over most of the horizon. The posterior median of business credit rises as well, although its response is less precise. Permits signal construction to come, starts record construction under way, and industrial production captures aggregate real activity. Housing activity and output rise together after the shock and later decline together. This comovement is consistent with the emphasis of \citet{Leamer2007,Leamer2015} on housing activity in the U.S. business cycle.

The contrast with the household credit shock is sharp. The posterior median of industrial production rises substantially after the housing production shock but only trivially after the household credit shock. Household credit rises after both shocks and therefore does not by itself distinguish their output responses. Housing production rises after the housing production shock and falls after the household credit shock. Within the selected system, output rises only in the first case. The comparison shows that the output response changes with the information set; it does not isolate the contribution of any single added variable.

\subsection{Three interpretations}\label{sec:bb_channel}

The response to a household credit shock admits of three complementary interpretations. MSV interpret the initial rise in output as a credit supply driven consumption boom and associate the later decline with debt overhang, deleveraging, and macroeconomic frictions. BPSS leave the source of the initial rise largely open, noting that credit and output generally move together, and attribute the later decline to endogenous monetary tightening. Our selected system adds a third interpretation: housing production is an omitted real activity margin that can explain the initial boom in systems that exclude housing starts and permits. These three interpretations can coexist.

MSV and BPSS differ mainly in their explanations of the later output decline. In MSV, the 95\% confidence band contains zero over nearly all of the horizon on which the estimated GDP response is negative. In BPSS, both credible bands contain zero during the later decline. Yet both papers develop economic interpretations of that decline. Their interpretations therefore rest on the central impulse responses. In their argument, wide bands reduce precision but do not alter the economic interpretation. In our selected system, the posterior median of output also declines later. This response remains compatible with the debt overhang and deleveraging interpretation of MSV, the endogenous monetary tightening interpretation of BPSS, or both. Our selected system does not distinguish between them, and our main result does not require the later decline to be statistically different from zero.

The same focus on central responses also shapes their interpretation of the initial output rise. MSV interpret its positive central response as a consumption boom driven by credit supply. BPSS also obtain a positive initial response and relate it to the general positive comovement between credit and output, but they do not identify the spending or production margin behind the initial rise. Neither information set includes housing starts or permits. These systems therefore cannot determine whether the initial output rise comes from household credit itself or from residential construction that moves with credit.

Our selected system separates these margins. At the point estimate, measured by the posterior median, the initial response of industrial production to the household credit shock is near zero and economically negligible. This is not a case in which wider bands surround an output boom. The boom is absent from the estimated response itself. Housing production also falls after the household credit shock. This finding does not imply that falling housing production must cause aggregate output to decline, nor does it rule out support for consumption. It shows that any such support is insufficient to produce a meaningful aggregate output boom when residential construction contracts. Our contribution therefore concerns the source of the initial boom rather than the source of the later decline. Models of household credit and the business cycle should allow residential construction to move separately from household credit. In smaller systems that do not separate these margins, the identified household credit shock may combine movements in credit with movements in housing production.

\subsection{Robustness}\label{subsec:bb_robustness}

The next two sets of exercises establish two distinct results. First, the initial output boom in the original researcher-chosen systems is unaffected by differences in data frequency, variable definitions, updated vintages, and the Sims--Zha prior. Second, the selected 13-variable system produces no initial output boom under the heteroskedasticity identification of BPSS.

First, Figure~\ref{fig:hhc_qje} replicates MSV's VAR evidence using data that differ from theirs in both frequency and variable definition, as discussed in Section~\ref{sec:core}. We use monthly rather than annual data and the monthly Sims--Zha prior, whereas MSV do not use a Bayesian prior. Figure~\ref{fig:hhc_aer_dup} replicates the IP, BC, and HHC responses to the household credit shock reported by BPSS. We use their 10-variable system and sample period, but updated vintages of the monthly data. We also use the Sims--Zha prior, which is close to but not identical to the BPSS prior.\footnote{The full panel of impulse responses from this 10-variable BPSS system is reported in Appendix~\ref{sec:full_irfs}.} In both cases, the posterior median of output rises initially and later declines. Different data, updated vintages, and the Sims--Zha prior do not remove the initial output boom. The point estimate of the initial rise remains sizable.

Second, the selected system yields the same separation under BPSS's identification through shock heteroskedasticity. We apply their identifying method to our 13-variable system. Figure~\ref{fig:hhc_t_shk_aer_hhc} shows that a household credit shock raises household credit without producing an output boom, while all four housing activity variables decline.\footnote{These results hold when we augment the selected system with six BPSS variables absent from that system: the price level, M1, the federal funds rate, the 10Y--3M Treasury spread, the TED spread, and the commodity price index. In particular, there is no initial output boom after the household credit shock.} The GZ spread remains essentially unchanged. Because this spread measures corporate credit conditions, its muted response indicates that the household credit expansion is not accompanied by a broader deterioration in corporate credit markets.

Figure~\ref{fig:hhc_t_shk_aer_housing} shows that a housing production shock raises housing permits, housing starts, industrial production, household credit, and business credit. The same separation therefore obtains under BPSS's identification through shock heteroskedasticity and is not an artifact of the recursive identification. Across both identifying methods, household credit rises after both shocks, but output rises only when housing production expands. This result supports the housing production channel.

\section{Monetary Policy Transmission}\label{sec:mp}

This application asks how systematic construction of the SVAR information set changes the estimated transmission of monetary policy under proxy identification. We apply the methodology of Sections~\ref{sec:modelconstruction}--\ref{sec:oos} to GK's proxy SVAR. The application further develops a necessary and sufficient identification result for multiple instruments identifying one shock and an anchor-free joint Bayesian method. We first describe GK's two-stage procedure in our notation and then show how the proxy restrictions identify the monetary policy shock within the general SVAR framework.

The joint approach is not merely an extension or a broader formulation. It is needed for coherent and stable economic results.
Both GK's procedure and ours identify the same monetary policy shock and the same impulse responses. They differ in how they recover these responses from the data. GK's procedure singles out the residual of one variable as an anchor and measures every response relative to it. Our approach estimates the responses jointly and eliminates the need for such a choice. With a single instrument, the anchor does not matter, and GK use one instrument. With several instruments, however, the anchor generally matters in finite samples. The estimated responses and their credible bands can move with the variable chosen as the anchor, and no economic argument selects one variable over another.

Moving from GK's single instrument to the multiple instruments of \citet{Swanson2021} makes the joint method necessary for avoiding arbitrary finite-sample dependence on the choice of anchor. The empirical subsection uses the six variables from GK's Figure~2 as the core variables and the three external instruments of \citet{Swanson2021} to show how the joint approach operates within our system-selection methodology and what it implies for the transmission of monetary policy.

\subsection{Proxy identification within one SVAR system}\label{subsec:mp_proxy}
To avoid burdensome system-specific notation, throughout this section we reuse $n$, $A_0$, $A_\ell$, $\eps_t$, $\Sigma$, and related symbols for the SVAR system under consideration. With this slight abuse of notation, the system may represent the original core-variable system, a current system generated during model construction, or the selected terminal system, depending on the context. The dimensions and values of these objects are understood to conform to the system under consideration. The proxy identification argument below applies to the system being analyzed.

For every system under consideration, the monetary policy equation is placed first, and $\eps_{1t}$ is the monetary policy shock. Let $v_t$ denote the reduced-form innovation of that system, where the symbol $v$ is used because $u$ is reserved in Section~\ref{sec:modelconstruction} for the composite disturbance:
\begin{equation}
\label{eq:mp_reduced_form}
  v_t'
  =
  \eps_t'A_0^{-1},
  \qquad
  \Sigma
  \equiv
  \mathbb E(v_tv_t')
  =
  A_0^{-1\,\prime}A_0^{-1}.
\end{equation}
Let $e_1$ be the first column of $I_n$ and define
\begin{equation}
\label{eq:mp_a1_s}
  a_1
  \equiv
  A_0e_1,
  \qquad
  s
  \equiv
  A_0^{-1\,\prime}e_1,
  \qquad
  s'
  =
  e_1'A_0^{-1}
  =
  a_1'\Sigma.
\end{equation}
The vector $a_1$ contains the coefficients in the monetary policy equation. The row vector $s'$ gives the impact responses to the monetary policy shock.

Let $m_t$ be an $n_m\times1$ vector of centered external instruments, where $n_m\geq1$. Following the literature, we maintain the following proxy restrictions.

\begin{assumption} 
\label{ass:mp_proxy}
The instrument vector satisfies
\begin{equation}
\label{eq:mp_proxy_moments}
  \mathbb E(m_t\eps_{1t})
  =
  \gamma_m
  \neq
  0_{n_m\times1},
  \qquad
  \mathbb E(m_t\eps_{jt})
  =
  0_{n_m\times1},
  \quad
  j=2,\ldots,n.
\end{equation}
Let $\Omega_m\equiv\mathbb E(m_tm_t')$. The matrix $\Omega_m$ is positive definite.
\end{assumption}

Assumption~\ref{ass:mp_proxy} requires the instruments to have predictive content for the policy shock and to be orthogonal to every other structural shock. Define the reduced-form proxy covariance matrix
\begin{equation}
\label{eq:mp_G}
  G
  \equiv
  \mathbb E(v_tm_t')
  \in
  \Real^{n\times n_m}.
\end{equation}
For $i=1,\ldots,n$, let $g_i\in\Real^{n_m}$ denote the transpose of the $i$th row of $G$, that is, $g_i=\mathbb E(m_tv_{i,t})$, and let $W\equiv\Omega_m^{-1}$.

\runinhead{The GK two-stage procedure} GK first estimate the reduced-form VAR and obtain its residuals. They then choose one reduced-form residual as the policy indicator used in the first stage. Let $p\in\{1,\ldots,n\}$ index this anchor residual. An anchor $p$ is defined as admissible if and only if $g_p\neq0_{n_m\times1}$. Assuming $p$ is admissible, the two-stage least-squares coefficient from regressing $v_{i,t}$ on $v_{p,t}$, using $m_t$ as the instrument vector, is
\begin{equation}
\label{eq:mp_gk_ratio}
  \beta_{i\mid p}
  =
  \frac{g_p'Wg_i}{g_p'Wg_p}.
\end{equation}
Following the external-instrument argument of GK, Theorem~\ref{thm:mp_identification} below expresses their identifying conditions in matrix form. Assumption~\ref{ass:mp_proxy} implies $G=s\gamma_m'$, hence $g_i=\gamma_ms_i$ and
\[
  \beta_{i\mid p}
  =
  \frac{s_i}{s_p}
\]
whenever $s_p\neq0$. Collect these coefficients in the anchor-specific ratio vector 
\[
d_p=(\beta_{1\mid p},\ldots,\beta_{n\mid p})',
\] 
which satisfies $\beta_{p\mid p}=1$ and $d_p=s/s_p$. The two-stage procedure recovers this ratio representation and then uses the unit-variance normalization of the policy shock to recover its scale:
\begin{equation}
\label{eq:mp_gk_scale}
  s
  =
  \pm
  \frac{d_p}{\sqrt{d_p'\Sigma^{-1}d_p}}.
\end{equation}
The sign is chosen by an economic normalization, such as requiring a contractionary policy shock to raise the interest rate on impact.

GK choose the reduced-form residual for the one-year government bond rate as the admissible anchor, estimate the impact ratios relative to that residual, and then recover the scale of the impact vector. They do not examine how the ratio representation changes when a different residual is used as the admissible anchor. The next proposition shows that different anchors change this intermediate representation but recover the same impact vector and the same impulse responses after normalization.

\begin{proposition} 
\label{prop:mp_anchor}
Under Assumption~\ref{ass:mp_proxy}, every admissible anchor $p$ yields the ratio vector
\[
  d_p
  =
  \frac{s}{s_p}.
\]
For any two admissible anchors $p$ and $r$,
\[
  d_r
  =
  \frac{s_p}{s_r}d_p.
\]
Thus, different anchors yield different ratio vectors that are proportional to one another. After the scale and sign normalization in \eqref{eq:mp_gk_scale}, every admissible anchor recovers the same impact vector $s$ and the same impulse responses at every horizon.
\end{proposition}

\begin{proof}
   See Appendix~\ref{sec:mp_appendix}.
\end{proof}

Proposition~\ref{prop:mp_anchor} has a sharp finite-sample implication. With one instrument, the estimate $\widehat G$ has one column and therefore has rank one whenever it is nonzero. The sample analogue of the proposition then holds exactly in the data: different admissible anchors yield proportional ratio vectors, the same normalized impact vector, and the same impulse responses. With several instruments, Assumption~\ref{ass:mp_proxy} still makes $G$ rank one in population, but $\widehat G$ need not be rank one and is generally of higher rank in finite samples. Different anchors can then yield ratio vectors that are not proportional, causing the estimated impulse responses and their error bands to depend on the chosen anchor. The joint method introduced next estimates the entire system together. Its likelihood and posterior do not depend on any choice of anchor and therefore avoid this arbitrary finite-sample dependence.

\runinhead{The joint proxy SVAR} Our method estimates the SVAR and the proxy relation jointly under one posterior distribution. The contemporaneous second moments of the joint system satisfy
\begin{equation}
\label{eq:mp_joint_covariance}
  \mathbb E
  \left(
  \begin{bmatrix}
    \eps_t\\
    m_t
  \end{bmatrix}
  \begin{bmatrix}
    \eps_t' & m_t'
  \end{bmatrix}
  \right)
  =
  \begin{bmatrix}
    I_n & e_1\gamma_m'\\
    \gamma_m e_1' & \Omega_m
  \end{bmatrix},
\end{equation}
where $\Omega_m$ is positive definite. The block structure imposes the proxy moments in \eqref{eq:mp_proxy_moments} and preserves the unit-variance normalization in \eqref{eq:svar_full}. Under joint normality, we represent \eqref{eq:mp_joint_covariance} as an augmented structural system for $\widetilde y_t=[y_t',m_t']'$. The proxy exogeneity conditions become linear, non-cross-equation block restrictions on the augmented coefficient matrices, and the augmented shock vector is normally distributed with mean zero and identity covariance matrix. We place the Sims--Zha prior on the augmented system and draw from its posterior using the efficient Gibbs sampler of \citet{WZ03b}, which applies directly and without modification. Because each Gibbs draw has an invertible augmented contemporaneous matrix, the implied covariance matrix in \eqref{eq:mp_joint_covariance} is positive definite, and the Schur complement condition \eqref{eq:mp_joint_pd} below holds automatically.

\citet{CaldaraHerbst2019} develop a published model with one proxy for one shock. Their online appendix considers two proxies for one shock but assumes that the two measurement errors are independent. In our notation, this assumption requires $\Omega_m-\gamma_m\gamma_m'$ to be diagonal. It therefore excludes residual covariance between the proxies and requires all covariance between them to arise through their common loading on the shock. Identification does not require this restriction. Applied to our three instruments, it would impose three zero restrictions on the off-diagonal elements of the measurement error covariance matrix. When written as an augmented structural system, their measurement equation also couples the proxy equation to the policy equation through cross-equation proportionality restrictions. 

These restrictions fall outside the non-cross-equation class studied by \citet{WZ03b}. The posterior simulator of \citet{CaldaraHerbst2019} is therefore a Metropolis-within-Gibbs algorithm with tuned proposals. When the proxy is highly informative, the proposal based on the VAR-only posterior can differ sharply from the joint posterior, causing low acceptance rates. They consequently use a sequential Monte Carlo sampler for their reported applications, which they describe as stable but computationally demanding. Our specification instead leaves $\Omega_m$ unrestricted within the class of positive definite matrices and imposes only proxy exogeneity and the Schur complement condition in \eqref{eq:mp_joint_pd} below; identification further requires $\gamma_m\neq0$, as in GK and the earlier external-instrument literature. This unrestricted covariance specification nests both the published scalar model of \citet{CaldaraHerbst2019} and their independent-error multiple-proxy extension. Moreover, Proposition~\ref{prop:mp_gk_equivalence} below establishes the population equivalence between our joint Bayesian approach and the GK two-stage procedure.

To interpret the remaining restriction, consider the linear projection
\begin{equation}
\label{eq:mp_proxy_projection}
  \eps_{1t}
  =
  m_t'\Omega_m^{-1}\gamma_m+w_{1t},
  \quad
  \mathbb E(m_tw_{1t})
  =
  0_{n_m\times1}.
\end{equation}
Since $\Var(\eps_{1t})=1$, the fraction of the policy shock variance explained by the instruments is
\[
  \gamma_m'\Omega_m^{-1}\gamma_m,
\]
and the variance of the projection residual is
\[
  \Var(w_{1t})
  =
  1-\gamma_m'\Omega_m^{-1}\gamma_m.
\]
The covariance matrix in \eqref{eq:mp_joint_covariance} is positive definite when
\begin{equation}
\label{eq:mp_joint_pd}
  \gamma_m'\Omega_m^{-1}\gamma_m
  <
  1.
\end{equation}
The projection residual has positive variance if and only if the Schur complement inequality in \eqref{eq:mp_joint_pd} holds. Together with $\gamma_m\neq0$, this inequality implies that the instruments provide information about the policy shock without perfectly predicting it. Given that $\Omega_m$ is positive definite, it also ensures that the joint covariance matrix in \eqref{eq:mp_joint_covariance} is positive definite.

Utilizing \eqref{eq:mp_proxy_projection} in the first equation of \eqref{eq:svar_full} gives
\begin{equation}
\label{eq:mp_joint_first_equation}
  y_t'a_1
  =
  c_1
  +
  \sum_{\ell=1}^{\cL}y_{t-\ell}'a_{\ell,1}
  +
  m_t'\delta_m
  +
  w_{1t},
  \quad
  a_{\ell,1}
  =
  A_\ell e_1,
\end{equation}
where $\delta_m =\Omega_m^{-1}\gamma_m$. The instruments enter the first equation through the linear projection of the policy shock on $m_t$. For $j=2,\ldots,n$, Assumption~\ref{ass:mp_proxy} implies $\mathbb E(m_t\eps_{jt})=0_{n_m\times1}$. The linear projection of each remaining shock on the instruments is therefore zero, and no instrument term enters the remaining equations.

The external instrument vector $m_t$ joins the SVAR system under consideration through the first equation. By construction, $\mathbb E(m_tw_{1t})=0_{n_m\times1}$. For $j=2,\ldots,n$,
\[
  \mathbb E(w_{1t}\eps_{jt})
  =
  \mathbb E(\eps_{1t}\eps_{jt})
  -
  \delta_m'\mathbb E(m_t\eps_{jt})
  =
  0,
\]
because the model shocks are contemporaneously uncorrelated and Assumption~\ref{ass:mp_proxy} makes the instruments uncorrelated with every remaining shock. Thus, $w_{1t}$ is contemporaneously uncorrelated with the shocks in the remaining equations.

This augmented specification is what we have referred to as the joint system. The coefficient $\delta_m$ is estimated together with all other parameters in that system. Theorem~\ref{thm:mp_identification} below establishes that the joint system identifies all coefficients in the first equation of the SVAR system under consideration and therefore identifies the impulse responses to the first shock at every horizon.\footnote{Following Theorem~4 of \citet{tZ99}, Proposition~\ref{prop:mp_normalization} in Appendix~\ref{sec:mp_appendix} shows that an orthogonal transformation of columns $2,\ldots,n$ leaves the joint covariance restriction, the first equation, and the impulse responses to the first shock unchanged. The identification result is related to the partial-identification framework of \citet{jRdWtZ10}, whose Theorem~2 gives a general sufficient rank condition for identifying one equation. Here the proxy moments yield a necessary and sufficient condition for identifying the first equation and provide explicit formulas for its coefficients and impulse responses.} No anchor is chosen, and no ratio calculation is needed.

\begin{theorem} 
\label{thm:mp_identification}
Assume that $n\geq2$, the SVAR system under consideration has an invertible $A_0$, and \eqref{eq:mp_reduced_form} holds. Let
\[
  \gamma_m
  \equiv
  \mathbb E(m_t\eps_{1t}),
  \qquad
  \Omega_m
  \equiv
  \mathbb E(m_tm_t'),
\]
where $\Omega_m$ is positive definite, and maintain the exclusion restrictions
\[
  \mathbb E(m_t\eps_{jt})
  =
  0_{n_m\times1},
  \qquad
  j=2,\ldots,n.
\]
In the absence of additional identifying restrictions on the SVAR system, the joint system identifies all coefficients in the first equation up to their common sign if and only if
\[
  \gamma_m
  \neq
  0_{n_m\times1}.
\]
Equivalently, identification holds if and only if $G\neq0_{n\times n_m}$, or, equivalently,
\[
  \rank(G)
  =
  \rank(\Sigma^{-1}G)
  =
  1.
\]
When these equivalent conditions hold,
\begin{equation}
\label{eq:mp_rank_one}
  G
  =
  s\gamma_m',
  \qquad
  \Sigma^{-1}G
  =
  a_1\gamma_m'.
\end{equation}
For any $q\in\Real^{n_m}$ such that $Gq\neq0_{n\times1}$,
\begin{equation}
\label{eq:mp_identified_objects}
  a_1
  =
  \pm
  \frac{\Sigma^{-1}Gq}
  {\sqrt{(\Sigma^{-1}Gq)'\Sigma(\Sigma^{-1}Gq)}},
  \qquad
  s
  =
  \Sigma a_1
  =
  \pm
  \frac{Gq}
  {\sqrt{(Gq)'\Sigma^{-1}(Gq)}}.
\end{equation}
Thus, the joint system identifies $c_1$, $a_1$, $a_{\ell,1}$ for $\ell=1,\ldots,\cL$, and $\delta_m$, up to the common sign normalization of the first equation. It therefore identifies the impulse responses to the first shock at every horizon. The remaining equations are not identified without additional restrictions and are not needed for these impulse responses.
\end{theorem}
\begin{proof}
See Appendix~\ref{sec:mp_appendix}.
\end{proof}

The identification result in Theorem~\ref{thm:mp_identification} applies beyond a single shock. The label ``first'' is merely a notational convention; any shock can be placed first. More important, identification need not proceed one equation at a time. Multiple shocks can be identified simultaneously. Let $\mathcal J\subseteq\{1,\ldots,n\}$ index a set of shocks, and for each $j\in\mathcal J$ let $m_t^{(j)}$ be a vector of instruments satisfying 
\[
  \mathbb E\bigl(m_t^{(j)}\eps_{jt}\bigr)
  =
  \gamma_j
  \neq
  0,
  \qquad
  \mathbb E\bigl(m_t^{(j)}\eps_{kt}\bigr)
  =
  0,
  \quad
  k\neq j.
\]
Let $m_t$ stack the instrument vectors $m_t^{(j)}$, let $\Omega=\mathbb E(m_tm_t')$, and let $C=\mathbb E(\eps_t m_t')$. For population identification, assume $\Omega\succ0$ together with the relevance and exclusion restrictions above. For estimation under one regular joint Gaussian posterior, assume that
\[
\begin{bmatrix}
I_n & C\\
C' & \Omega
\end{bmatrix}
\succ0,
\]
or equivalently $I_n-C\Omega^{-1}C'\succ0$ and $\Omega-C'C\succ0$. Under these conditions, Theorem~\ref{thm:mp_identification} applies to each instrument vector and simultaneously identifies the corresponding shock and equation, the impact vector $s_j=A_0^{-1\prime}e_j$, and the corresponding impulse responses at every horizon, up to a separate sign normalization for each shock. The joint method makes this extension straightforward: the instrument vectors for different shocks are stacked, the corresponding restrictions are imposed on the joint covariance block, and estimation proceeds under one posterior distribution as in a standard SVAR with several identified shocks and the corresponding equations.

We now return to the notation of Theorem~\ref{thm:mp_identification}, where the shock of interest is labeled ``first.'' The vector $q$ is only an algebraic device for displaying the identified direction. Under Assumption~\ref{ass:mp_proxy}, $Gq$ is proportional to $s$ for every $q$ such that $Gq\neq0_{n\times1}$. The joint posterior is defined using the full instrument vector and does not depend on any choice of $q$. The identified impact vector and the reduced-form dynamics determine the impulse responses to the shock of interest.

Define the reduced-form coefficient matrices
\begin{equation}
\label{eq:mp_B}
  B_\ell
  \equiv
  A_\ell A_0^{-1},
  \qquad
  \ell=1,\ldots,\cL,
\end{equation}
and define the moving-average matrices by
\begin{equation}
\label{eq:mp_Psi}
  \Psi_0
  =
  I_n,
  \qquad
  \Psi_h
  =
  \sum_{\ell=1}^{\min(\cL,h)}
  \Psi_{h-\ell}B_\ell,
  \quad
  h\geq1.
\end{equation}
The horizon-$h$ impulse-response row to the first shock is
\begin{equation}
\label{eq:mp_irf}
  \mathrm{IRF}_1(h)
  =
  s'\Psi_h
  =
  a_1'\Sigma\Psi_h.
\end{equation}
Theorem~\ref{thm:mp_identification} therefore identifies $\mathrm{IRF}_1(h)$ up to the common sign normalization for every $h\geq0$.

\begin{proposition} 
\label{prop:mp_gk_equivalence}
Under Assumption~\ref{ass:mp_proxy}, the GK two-stage procedure based on any admissible anchor and the joint proxy SVAR identify the same impact vector $s$ after the common scale and sign normalization. They therefore identify the same first equation and the same impulse responses at every horizon.
\end{proposition}
\begin{proof}
   See Appendix~\ref{sec:mp_appendix}.
\end{proof}   

Proposition~\ref{prop:mp_gk_equivalence} confirms that the joint method preserves the population object identified by GK. The methods differ in finite-sample estimation. The joint method requires no anchor, estimates the SVAR and the proxy relation under one posterior distribution, and carries uncertainty about both components jointly. It also avoids division by an anchor loading. With one instrument, the normalized GK procedure is anchor invariant even in finite samples. With multiple instruments, the sample covariance matrix generally need not have rank one, and the GK estimates can depend on the residual chosen as the anchor. The joint method remains anchor free. Proposition~\ref{prop:mp_efficiency} in Appendix~\ref{sec:mp_appendix} shows that, under its stated regularity conditions, the joint estimator attains the parametric efficiency bound and weakly dominates any regular two-stage estimator of the same object under the same model.

\subsection{Policy surprises and financial conditions}\label{subsec:mp_finance}
The preceding subsection establishes that the joint proxy SVAR identifies the monetary policy shock, its equation, and its impulse responses without choosing an anchor. We now combine this anchor-free identification with model construction from the large data set described in Section~\ref{sec:household} to revisit GK's central question: how monetary policy surprises pass through financial markets to prices and real activity. GK use the 6-variable specification in their Figure~2 as the baseline: the one-year Treasury rate, CPI, industrial production, the excess bond premium, the mortgage spread, and the commercial paper spread (see Table~\ref{tab:mp_model17_variables} for data details). The one-year Treasury rate summarizes both the current policy stance and the expected path of short rates. CPI and industrial production measure the aggregate price and real activity responses. The three spreads trace distinct credit-cost margins in corporate bond, mortgage, and short-term business credit markets.

These 6 variables define the economic question and therefore serve as the core variables throughout model construction and OOS selection. Our methodology then determines which additional variables enlarge the information set. This approach preserves direct comparability with GK while replacing their researcher-chosen extensions with our model-construction procedure and OOS forecast loss criterion. Identification and information selection perform distinct tasks: the three high-frequency instrument series constructed by \citet{Swanson2021} are used jointly as one instrument vector to identify the monetary policy shock, while the model-construction procedure and OOS forecast loss criterion determine which variables are included to trace its transmission through financial markets to prices and real activity. The model-construction sample runs from 1988:02 through 2012:06. The end date matches GK's sample endpoint, while the start date is set by the earliest common availability of the three instrument series. The OOS validation sample covers the next 10 years, from 2012:07 through 2022:06.

We treat the three series of \citet{Swanson2021} as one instrument vector for a single monetary policy shock under Assumption~\ref{ass:mp_proxy}. Their federal funds rate, forward-guidance, and large-scale-asset-purchase labels describe distinct dimensions of announcement surprises, not separate shocks in our model.

\subsubsection{The economic question and the core evidence}
At the center of GK is a fundamental question about monetary policy transmission. In a frictionless benchmark, monetary policy affects private borrowing rates only by altering current and expected future short rates. Term premia and credit spreads do not respond. A modest movement in the path of safe rates should therefore produce a similarly modest movement in private credit costs. With financial frictions, however, a monetary tightening can widen credit spreads. The resulting increase in private borrowing costs can lead to declines in economic activity. The 6-variable core system that GK use as their baseline directly addresses the central question of their paper: ``how credit costs respond to exogenous surprises in monetary policy.'' 

Figure~\ref{fig:mp_core} provides our empirical estimates for this core system. The excess bond premium and commercial paper spread rise sharply after the monetary policy shock, with credible bands above zero at both the 68\% and 90\% levels, and then decline toward zero. The mortgage spread also rises, although its credible bands lie above zero only at the 68\% level and in the short run. By contrast, the one-year Treasury rate rises only slightly on impact, and its credible bands include zero at both the 68\% and 90\% levels. Private credit conditions therefore tighten more clearly than the policy indicator itself. The response of the economy cannot be summarized by a large and precisely estimated movement in the one-year rate. CPI falls persistently, with credible bands below zero at the 90\% level for almost two years. The median response of industrial production also turns negative, but its credible bands lie below zero only at the 68\% level, between one and two years after the shock.

The broad economic response supports GK's central mechanism. A monetary policy shock raises credit spreads across corporate bond, mortgage, and short-term business credit markets while prices and output decline. The rise in borrowing spreads represents a transmission margin beyond the safe rate. It increases the cost of external finance faced by firms and households and puts downward pressures on prices and output. The evidence is particularly strong for the excess bond premium and the commercial paper spread, which identify transmission through both long-term corporate credit and short-term business finance.

At the same time, our results differ from GK in informative respects. In their Figure~2, the one-year rate rises by roughly 20 basis points on impact, industrial production falls significantly after several months, CPI declines but not significantly at the 95\% level, and all three credit spreads rise significantly. In our Figure~\ref{fig:mp_core}, the decline in CPI is more precisely estimated than the decline in industrial production, and the response of the one-year rate is not credibly different from zero. The main agreement between our results and GK's is the increase in credit spreads across several financial markets. The precise response of the policy indicator and the relative strength of the price and output responses are less robust across the two empirical specifications.

This difference reflects design differences in several respects: the samples differ, the instruments differ, the estimation methods differ, and the error bands differ. GK identify the shock by instrumenting the reduced-form innovation in the one-year Treasury rate with FF4, the high-frequency surprise in the three-month-ahead federal funds futures rate. The one-year rate is the anchor of their construction. We instead use the three high-frequency series of \citet{Swanson2021} jointly as one instrument vector and estimate the proxy relation and the SVAR under one posterior distribution without any anchor. The three series contain policy news associated with the federal funds rate, forward guidance, and large-scale asset purchases. 

Our credible bands are wider than those reported by GK. The sample, instruments, model, prior, estimation method, and interval construction all differ, so this cross-study comparison does not identify the source of the difference and is not an efficiency comparison. Within our model, posterior uncertainty about the common proxy direction, the instrument loadings, and the SVAR parameters enters every impulse response. The wider credible bands do not contradict the asymptotic efficiency result in Proposition~\ref{prop:mp_efficiency} of Appendix~\ref{sec:mp_appendix}, which shows that the joint estimator attains the efficiency bound and weakly dominates regular two-stage estimators under the same model.

Yet, despite these differences in design and statistical uncertainty, the central transmission mechanism through credit markets holds. Our design reproduces the central GK finding: a monetary policy surprise widens credit spreads. The result holds under a different identification, with three instruments jointly and no privileged reduced-form innovation, and without a large or precise response of the one-year rate. Even without such a response, the excess bond premium and commercial paper spread rise credibly, the mortgage spread is credibly positive at the 68\% level, and prices and output are under downward pressures from these tightening credit markets. GK stress that modest movements in the one-year safe rate can produce large movements in credit costs. In that sense, our results strengthen their message: even without a credible movement in the one-year safe rate, contractionary monetary policy still tightens credit markets. Thus, our results establish the crucial fact behind the GK interpretation: monetary policy surprises tighten private credit conditions in ways that are not summarized by the response of the one-year safe rate. This fact highlights an important financial amplification margin, as emphasized by GK.

\subsubsection{Selection}

Our model-construction procedure concerns only the monetary policy shock, because that is the single shock identified. It therefore targets one component of the composite disturbance, the policy-equation component $\widehat u_{1,t}^{(k)}$, which is scalar. As discussed in Sections~\ref{sec:modelconstruction} and~\ref{sec:theory}, this component changes as variables are admitted. Each admitted variable enlarges the current system. The enlarged system is then reestimated by the joint method of Section~\ref{subsec:mp_proxy}, and $\widehat u_{1,t}^{(k)}$ changes accordingly. The OOS forecast loss criterion selects a 19-variable system at $\widehat\lambda_\alpha=0.0881$: the 6 GK core variables and the 13 additional variables reported in Table~\ref{tab:mp_model17_variables}.\footnote{The 19-variable system at $\widehat\lambda_\alpha=0.0881$, selected by the OOS forecast loss criterion, reaches its terminal system after one update. The procedure begins with the 6 GK core variables at iteration 0 and adds all 13 additional variables in the first and only update. Across the candidate values of $\lambda$, most terminal systems require two updates, some require one, and a few require more than two.}

Our methodology selects neither a ladder of two-year, five-year, and ten-year Treasury rates nor several closely related Treasury-minus-federal-funds spreads. This is intuitive because these highly correlated variables add little distinct information for predicting the innovation in the policy equation. The fact that the 19-variable selected system is not nested in GK's sequence of interest rate extensions is economically informative. The procedure does not merely choose which maturity to append to their baseline. It instead adds variables that measure distinct margins of labor adjustment, housing activity, asset valuation, volatility, credit pricing, liquidity conditions, exchange rates, and commodity prices, none of which can be recovered by moving along a Treasury yield ladder. The selected system therefore expands the empirical content of monetary policy transmission rather than only its interest rate dimension.

The selected variables pass two distinct procedures. They first enter because their contemporaneous values or lags have predictive content for the one fitted composite disturbance of the policy equation in the current system. The resulting terminal system is then compared with the other terminal systems through the OOS forecast loss of the core variables.\footnote{Computations were performed in MATLAB R2023b under Windows 11 Enterprise on a computer with an Intel Core Ultra 7 165U processor at 2.10 GHz and 32 GB of RAM. MATLAB had access to 12 CPU cores, although the computations were serial. The complete model-construction process across all values of $\lambda$ required 290 seconds of wall-clock time. The OOS selection based on 2,000 posterior draws for each model required 840 seconds (14 minutes). Increasing the number of posterior draws did not alter the results.} This separation matters. The variables are not chosen because they produce preferred impulse responses or because their economic labels fit a prior story. They enter through model construction, and the corresponding terminal system is evaluated by the common OOS criterion before the impulse responses are interpreted.

Model construction and impulse-response reporting serve different purposes. The 13 additional variables enter jointly in the optimization procedure, and their conditional predictive content for the fitted composite disturbance determines their admission. At the terminal iteration, they form the selected 19-variable information set in which the policy equation and monetary policy shock are estimated.

Identification is system specific. Given any core-plus-one reporting system, Theorem~\ref{thm:mp_identification} identifies the policy shock and its equation for that system under the maintained relevance and exclusion restrictions; the 19-variable dimension is not an identification condition for that given system. The selected information set determines which transmission margins are reported. We therefore follow GK and estimate each core-plus-one system separately because the selected 19-variable system produces credible bands too wide to distinguish the individual transmission margins clearly. Each reported response is conditional on the corresponding reporting system, not on the selected 19-variable system.

\subsubsection{Transmission margins beyond the GK core}
The 13 additional variables expand the GK core along six economically distinct margins: labor market adjustment, housing activity, external and commodity conditions, equity valuation and volatility, credit and liquidity conditions, and expected default risk. We take them in turn. Each figure adds one variable from the selected information set to the six-variable core, holds the core fixed, and reports posterior medians with 68\% and 90\% credible bands.\footnote{In each 7-variable system, the responses of the six core variables are essentially unchanged from the core system in Figure~\ref{fig:mp_core}. We therefore report only the response of the added variable.} We move from the real-activity and external margins, where the responses show downward pressures with limited precision, to the financial margins, where the responses are credible at the 90\% level over the early horizons.

\runinhead{Labor market adjustment} The selected system adds short-duration unemployment, hours in goods-producing industries, and hours in manufacturing (Figure~\ref{fig:mp_labor}). Hours are the intensive margin of labor demand, which firms can vary without changing employment or capacity. Short-duration unemployment measures recent job separations and hiring conditions. GK include neither margin. Their extensions to the baseline add interest rates one at a time, not labor variables. A monetary tightening can weaken demand, raise financing costs, and tighten working capital, and firms can respond on the hours margin. 

The evidence is mixed. Short-duration unemployment shows no clear response; its posterior median is small and its bands are wide. Hours in both sectors decline, with the largest reductions in the first year. The 68\% bands lie below zero over the medium run, while the 90\% bands include zero. A monetary tightening therefore puts downward pressures on hours worked. GK measure activity through industrial production alone. These labor margins complement that measure and respond to the same contractionary shock. This group also shows why selection and response strength are distinct. Short-duration unemployment is selected for its predictive content, yet its own response to the identified shock is weak.

\runinhead{Housing activity} The selected system adds three regional housing series: starts in the Northeast and Midwest and permits in the Northeast (Figure~\ref{fig:mp_housing}). This group extends the mortgage-spread channel already in the GK core. The mortgage spread measures the price of housing finance. Permits and starts measure the quantity response, at the planning and construction stages. Housing is among the most interest-sensitive sectors of the economy, and the sector where the policy rate, long rates, mortgage spreads, collateral values, and construction decisions meet. The regional labels carry no structural interpretation. These indicators are selected for their informational content in the model-construction procedure. The selection of more than one regional series indicates that regional housing adjustments are not synchronized and that this heterogeneity carries predictive content for the fitted composite disturbance in the policy equation.

Housing starts and permits have negative posterior medians. The 68\% credible band for housing permits includes zero. For housing starts, the 68\% credible bands lie strictly below zero over the one-year horizon in the Midwest and the two-year horizon in the Northeast, whereas the 90\% bands include zero.  A monetary tightening therefore puts downward pressures on residential construction. These quantity responses complement the positive mortgage-spread response in the core system: tighter housing finance is followed by weaker building. The selected additions show that the quantity side of the same channel belongs in the information set.

\runinhead{External and oil markets} The selected system adds the U.S.--U.K. exchange rate and the oil price (Figure~\ref{fig:mp_external_commodity}). Both contain fast-moving information about global financial conditions and cost pressures. The exchange rate reflects differences in monetary conditions across countries and changes in global risk. Oil prices respond to global demand and supply and also affect inflation. In the selected system, the exchange rate helps distinguish domestic policy movements from changes in foreign monetary conditions and global risk, while the oil price helps distinguish policy surprises from global demand and oil supply disturbances. Together, these variables help prevent external and commodity-market disturbances from contaminating the monetary policy shock.

Other U.S. dollar exchange rates are highly correlated with this bilateral rate. As with the regional housing series, the bilateral exchange rate is an indicator selected by the model-construction procedure, not a claim that the United Kingdom has a unique role in U.S. policy transmission. The exchange rate declines persistently. Because the series is measured in U.S. dollars per U.K. pound, the decline is a dollar appreciation after the monetary policy tightening. Its 68\% band lies below zero over the medium run, and the upper edge of the 90\% band is only barely above zero, so the posterior places overwhelming probability on a decline. The oil price falls on impact and then reverts, but its bands are wide throughout. A monetary tightening therefore puts downward pressures on the oil price, consistent with weaker global demand, though the oil response is only suggestive. GK include no exchange rate or commodity price. These margins condition the system on global conditions that move alongside policy surprises.

\runinhead{Equity valuation and volatility} The selected system adds the price-earnings ratio and the VIX (Figure~\ref{fig:mp_equity_uncertainty}). The two are related but distinct. The price-earnings ratio is a valuation ratio that reflects discount rates, expected earnings growth, and risk compensation. The VIX is a forward-looking measure of expected volatility. These are the first noncore margins credible at the 90\% level. The price-earnings ratio falls sharply during the first year, with both bands below zero over the early horizons. Its median later turns positive, but its 68\% credible band includes zero. The VIX rises immediately and then decays, with both bands above zero over the early horizons. A monetary policy tightening therefore operates through both valuation and volatility. Lower valuations and higher volatility affect investment, collateral values, risk-bearing capacity, and the willingness of intermediaries to supply credit. GK measure neither margin. These responses show a rapid deterioration in market risk conditions that the core credit spreads do not summarize.

\runinhead{Credit and liquidity conditions} The selected system retains the three core credit-cost measures and adds the Baa-minus-federal-funds (BAA-FF) spread, the TED spread, and the GZ corporate bond spread (Figure~\ref{fig:mp_credit_liquidity}). These cover different parts of the credit system. The BAA-FF spread combines the transmission of the policy rate with term and corporate-risk components. The TED spread measures the premium on unsecured wholesale dollar funding relative to safe Treasury bills, and captures funding stress and flight to safety. The GZ spread measures the total corporate bond spread. Together with the mortgage and commercial paper spreads, these variables trace credit conditions from borrowers back to the funding markets that support lending. 

Across the three core-plus-one reporting systems, the GZ spread rises sharply and is credible at the 90\% level over the early horizons, the strongest response of the three. The TED spread rises on impact and is credible at the 90\% level at the short end, then decays quickly. The BAA-FF response is hump-shaped and imprecise, not even credible at the 68\% level in the medium run. Together, these reporting systems extend GK's central result across corporate bond pricing, wholesale bank funding, and short-term business credit. The credit channel is not confined to a single spread.

\runinhead{Expected default risk} The GZ spread and the excess bond premium are both in the selected information set, and their difference isolates the spread component associated with expected default risk (Figure~\ref{fig:mp_expected_default}). GK emphasize that the excess bond premium removes the expected-default component and isolates the part of corporate spreads driven by risk-bearing capacity. The difference $\texttt{GZS}-\texttt{EBP}$ rises sharply after the shock, and both the 68\% and 90\% bands exclude zero over the early horizons. Monetary tightening therefore widens the corporate bond spread through two channels: the excess bond premium that GK emphasize and the expected-default component that they remove. This is the sharpest extension of their central result. The total rise in corporate borrowing costs reflects both a fall in risk-bearing capacity and a deterioration in the expected creditworthiness of borrowers. Expected default, therefore, is an independently essential part of the credit channel.

\runinhead{Summary} Taken together, the six margins convert GK's researcher-chosen sequence of interest-rate extensions into one selected information set of economically distinct margins. The selection adds no dense ladder of Treasury yields. It adds labor adjustment, housing variables, equity valuation, volatility, corporate credit, wholesale liquidity, the exchange rate, and the oil price. The strongest evidence lies in the financial margins: the responses of the price-earnings ratio, VIX, GZ spread, TED spread, and expected-default risk are credible at the 90\% level over the early horizons. The BAA--FF response is imprecise. The real activity and external margins show downward pressures with less precision.

The expected-default result extends GK's central finding most directly. In the core system augmented with the selected GZ spread, the component they remove from the corporate spread responds materially to monetary policy. The selected information set is therefore a disciplined enlargement of GK. It preserves their external-instrument strategy and emphasis on credit costs, replaces a single anchor and one instrument with anchor-free joint estimation from multiple instruments, and replaces researcher-chosen model expansion with iterative construction and OOS selection. The selected information set is larger than the GK baseline, narrower than the full candidate data set, and organized around distinct margins that trace monetary policy transmission through financial conditions.

\section{Conclusion}\label{sec:conclusion}

SVAR evidence is conditional on an information set. When hundreds of candidate series are available, the surrounding information set cannot be chosen credibly by hand. This paper formalizes that choice through an algorithm-driven Bayesian procedure. The researcher specifies the economic question, the core variables, and the identifying restrictions. For each value of the model complexity parameter, the construction step admits remaining candidates with predictive content for the fitted composite disturbances of the current system, reestimates the enlarged SVAR, and iterates to a terminal information set. For each loss, the Bayesian OOS criterion retains systems credibly better than the reference; when none exist, it retains systems credibly indistinguishable from the reference. The final rule selects the largest system in the union of the retained sets under squared and absolute loss. We establish finite termination, uniqueness of the construction path and terminal system, an exact stopping characterization, invariance of the auxiliary admission rule at a given iteration to the measurement units of the remaining candidates, stability under small numerical perturbations, and existence and uniqueness of the OOS-selected system.

The two applications show that the information set can alter the central economic conclusion. In the household credit application, the procedure selects four measures of housing production that neither MSV nor BPSS include. A household credit shock raises household credit while housing production falls and the posterior median of output shows no initial boom. A separate housing production shock raises construction, output, and household credit together. Within the selected system, output rises with expanding residential construction, not with household credit alone. Models of household credit and the business cycle should therefore allow housing production to move separately from household credit.

In the monetary policy application, joint estimation with three instruments and no anchor strengthens GK's central credit spread channel. In the core system augmented with the selected GZ spread, monetary policy tightening raises both the excess bond premium and expected default risk. The application delivers an independent theoretical result. Under the maintained proxy restrictions, a vector of external instruments identifies one shock and its equation if and only if at least one instrument is correlated with that shock. The theorem allows any number of instruments for one shock, extends to several shocks and their corresponding instrument vectors under the joint covariance-admissibility condition, and gives explicit formulas for the coefficients of the identified equation and its impulse responses without a reduced-form innovation anchor.

The framework is designed to be modular. The auxiliary optimization is one implementation of the construction step; alternative optimization criteria could use different regularizers, group a variable with its lags, or target several equations with different tuning parameters. The Bayesian selection step could use other loss functions or proper scoring rules for the full predictive distribution, and the candidate set could include constructed latent factors along with observed variables. The architecture, however, remains the same: the researcher defines the economic question and identifying restrictions, the algorithm constructs candidate information sets, Bayesian estimation incorporates parameter uncertainty, and OOS evidence selects model complexity. The specific theoretical guarantees and empirical performance would have to be established for each alternative. Our theoretical results provide those guarantees for the procedure studied here and establish a well-defined, reproducible, tractable, and economically informative benchmark for evaluating such alternatives.

The paper offers a broader lesson. Identification restrictions are stated, defended, and debated; yet the variables on which they operate are often chosen informally. When the set of candidate variables is vast, this asymmetry is no longer defensible. An explicit algorithm makes the information set reproducible and open to inspection, and the applications show that this discipline can alter the economic interpretation of an identified shock. Identification has long been disciplined. In the age of big data, the information set deserves the same discipline. This paper provides it.

\clearpage

\begin{table}[!htbp]
\centering
\caption{Variables in the selected household credit system}
\label{tab:mp_model11_variables}
\footnotesize
\renewcommand{\arraystretch}{1.08}
\begin{tabular}{@{}p{0.11\textwidth}p{0.29\textwidth}p{0.51\textwidth}@{}}
\toprule
\textbf{Short} & \textbf{Original ticker} & \textbf{Economic description}\\
\midrule
\texttt{IP} & \texttt{INDPRO} & Industrial production index.\\
\texttt{BC} & \texttt{BUSLOANS} & Commercial and industrial loans.\\
\texttt{HHC} & \texttt{HH} & Consumer and real estate loans at all commercial banks.\\
\midrule
\texttt{HRS-G} & \texttt{CES0600000007} & Average weekly hours in goods-producing industries.\\
\texttt{HS-S} & \texttt{HOUSTS} & Housing starts in the South.\\
\texttt{PER-NE} & \texttt{PERMITNE} & New private housing permits in the Northeast, at a seasonally adjusted annual rate.\\
\texttt{T1Y-FF} & \texttt{T1YFFM} & One-year Treasury rate minus the federal funds rate.\\
\texttt{EBP} & \texttt{EBP\_OA} & Option-adjusted excess bond premium of GZ---the corporate bond spread component not attributable to expected default risk.\\
\texttt{GZS} & \texttt{GZ\_SPR} & GZ corporate bond spread, comprising the component associated with expected default risk and the excess bond premium.\\
\texttt{HRS-M} & \texttt{AWHMAN} & Average weekly hours in manufacturing.\\
\texttt{PER} & \texttt{PERMIT} & New private housing permits, at a seasonally adjusted annual rate.\\
\texttt{UI} & \texttt{CLAIMSx} & Initial claims for unemployment insurance.\\
\texttt{PER-W} & \texttt{PERMITW} & New private housing permits in the West, at a seasonally adjusted annual rate.\\
\bottomrule
\end{tabular}
\par
\vspace{0.4em}
\noindent
\begin{minipage}{0.95\textwidth}
\footnotesize
\raggedright
\textit{Note:} The table reports the 13 variables included in the selected system at $\lambda=0.1129$. It gives the short abbreviations used in the discussion in the first column, the original series names in the second column, and their economic meanings in the third column. The original tickers are FRED series IDs whenever available.
\end{minipage}
\end{table}

\begin{table}[!htbp]
\centering
\caption{Variables in the selected monetary policy system}
\label{tab:mp_model17_variables}
\footnotesize
\renewcommand{\arraystretch}{1.08}
\begin{tabular}{@{}p{0.11\textwidth}p{0.29\textwidth}p{0.51\textwidth}@{}}
\toprule
\textbf{Short} & \textbf{Original ticker} & \textbf{Economic description}\\
\midrule
\texttt{T1Y} & \texttt{GS1}  & One-year Treasury rate.\\
\texttt{CPI} & \texttt{CPIAUCSL}  & Consumer price index, all items.\\
\texttt{IP} & \texttt{INDPRO}  & Industrial production index.\\
\texttt{EBP} & \texttt{EBP\_OA} & GZ option-adjusted excess bond premium, the corporate bond spread component not attributable to expected default risk.\\
\texttt{MGS} & \texttt{MortgageSpliceGS10}  & Mortgage spread relative to the ten-year Treasury rate.\\
\texttt{CPS} & \texttt{CP3M\_TB3MS}  & Three-month commercial paper spread relative to the three-month Treasury bill rate.\\
\midrule
\texttt{U5} & \texttt{UEMPLT5}  & Civilians unemployed for less than 5 weeks.\\
\texttt{HRS-G} & \texttt{CES0600000007}  & Average weekly hours in goods-producing industries.\\
\texttt{HRS-M} & \texttt{AWHMAN}  & Average weekly hours in manufacturing.\\
\texttt{HS-NE} & \texttt{HOUSTNE}  & Housing starts in the Northeast.\\
\texttt{HS-MW} & \texttt{HOUSTMW}  & Housing starts in the Midwest.\\
\texttt{PER-NE} & \texttt{PERMITNE}  & New private housing permits in the Northeast, at a seasonally adjusted annual rate.\\
\texttt{PE} & \texttt{S\&P PE ratio}  & Price-earnings ratio for the S\&P composite common stock index.\\
\texttt{BAA-FF} & \texttt{BAAFFM}  & Moody's Baa corporate bond yield minus the federal funds rate.\\
\texttt{FX-UK} & \texttt{EXUSUKx}  & U.S.--U.K. foreign exchange rate.\\
\texttt{OIL} & \texttt{OILPRICEx}  & Crude oil price, spliced WTI and Cushing series.\\
\texttt{VIX} & \texttt{VIXCLSx}  & CBOE implied volatility index.\\
\texttt{TED} & \texttt{MED3\_TB3MS}  & Three-month Eurodollar deposit rate minus the three-month Treasury bill rate.\\
\texttt{GZS} & \texttt{GZ\_SPR} & GZ corporate bond spread, comprising the component associated with expected default risk and the excess bond premium.\\
\bottomrule
\end{tabular}
\par
\vspace{0.4em}
\noindent
\begin{minipage}{0.95\textwidth}
\footnotesize
\raggedright
\textit{Note:} The table reports the variables included in the selected monetary policy system. It gives the short abbreviations used in the discussion in the first column, the original series names in the second column, and their economic meanings in the third column. The original tickers are FRED series IDs whenever available.
\end{minipage}
\end{table}

\begin{figure}[htbp]
\centering
\includegraphics[width=\textwidth]{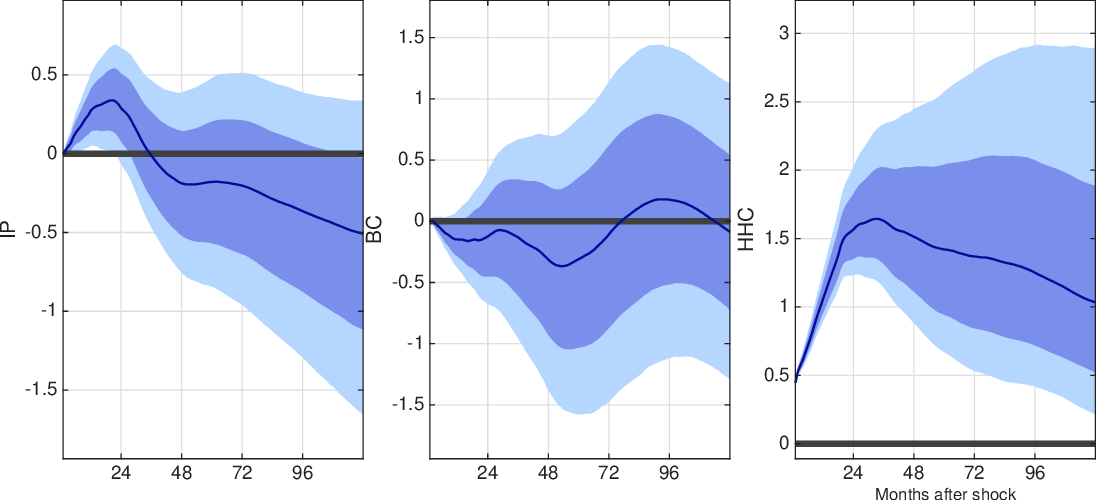}
\caption{Household credit shock and its dynamic impacts. This figure replicates MSV's Figure I using updated monthly data and the monthly Sims--Zha prior. All responses are expressed in percent. Unless otherwise stated, throughout this paper we follow GK and use ``percent'' as shorthand: variables entered in log levels are reported as approximate percentage changes, whereas rates, spreads, and shares are reported in percentage points. The dark band shows the 68\% posterior credible band, and the light band the 90\% posterior credible band. See Table~\ref{tab:mp_model11_variables} for complete variable descriptions.}
\label{fig:hhc_qje}
\end{figure}

Variables entered in log levels are reported as approximate percentage changes; rates, spreads, and shares are reported in percentage points.

\begin{figure}[htbp]
\centering
\includegraphics[width=\textwidth]{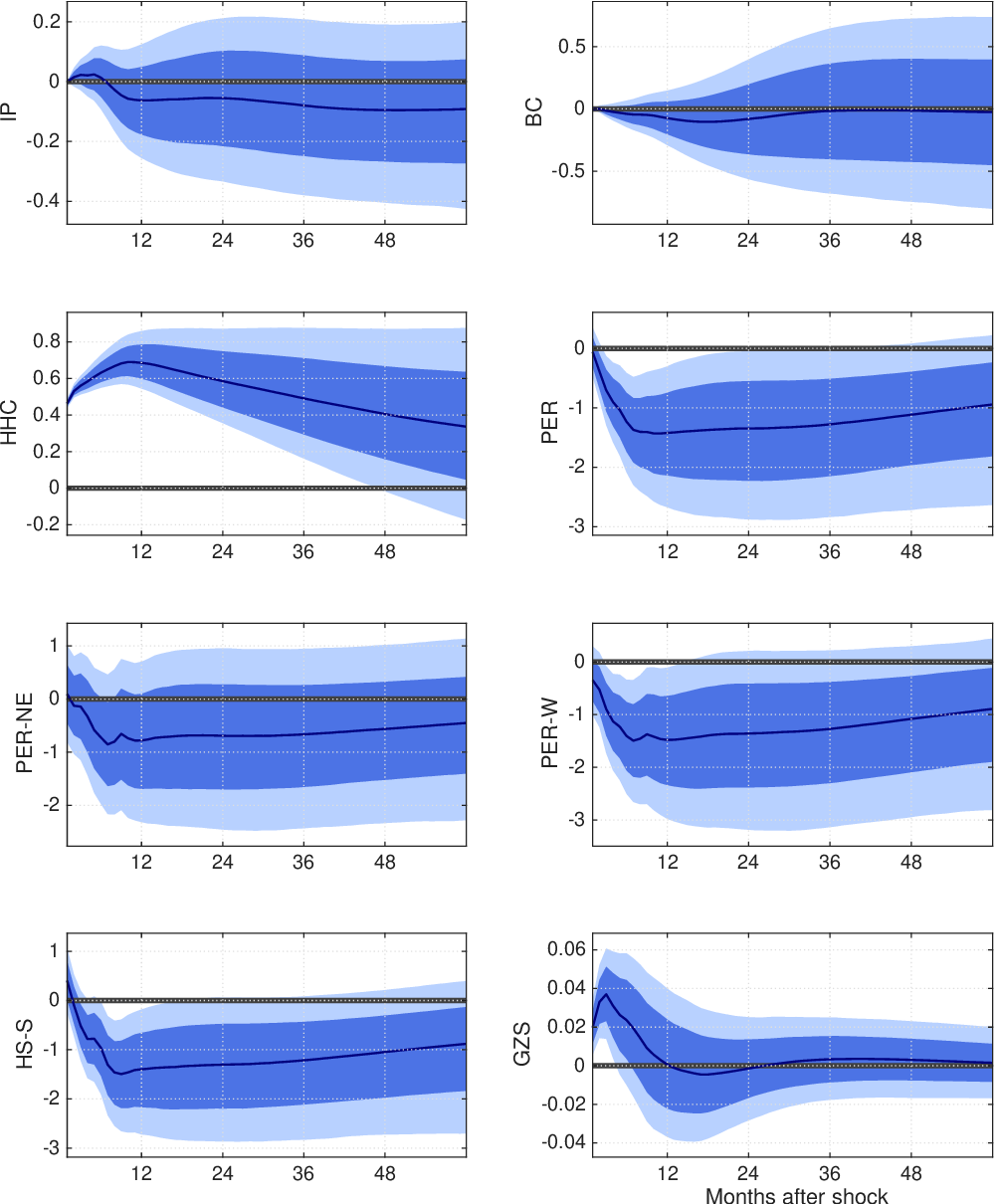}
\caption{Household credit shock and its dynamic effects. The impulse responses are estimated from the selected 13-variable system under the recursive identification used by MSV. All impulse responses are expressed in percent. The dark band shows the 68\% posterior credible band and the light band the 90\% posterior credible band. See Table~\ref{tab:mp_model11_variables} for complete variable descriptions.}
\label{fig:hhc_aerchol_hhc}
\end{figure}

\begin{figure}[htbp]
\centering
\includegraphics[width=\textwidth]{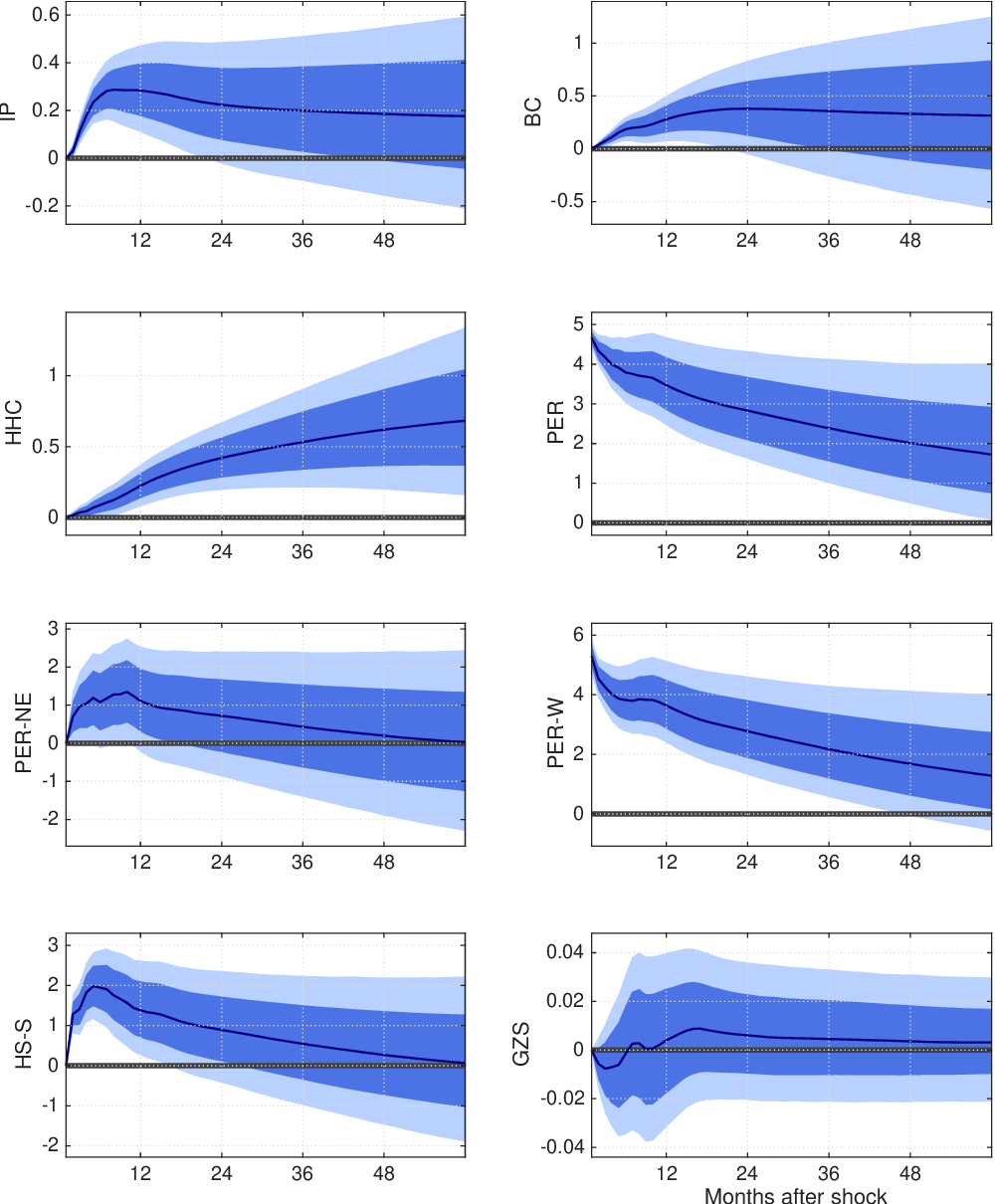}
\caption{Housing production shock and its dynamic impacts. The impulse responses are estimated from the selected 13-variable system under the recursive identification used by MSV. All impulse responses are expressed in percent. The dark band shows the 68\% posterior credible band, and the light band the 90\% posterior credible band. See Table~\ref{tab:mp_model11_variables} for complete variable descriptions.}
\label{fig:hhc_aerchol_housing}
\end{figure}

\begin{figure}[htbp]
\centering
\includegraphics[width=\textwidth]{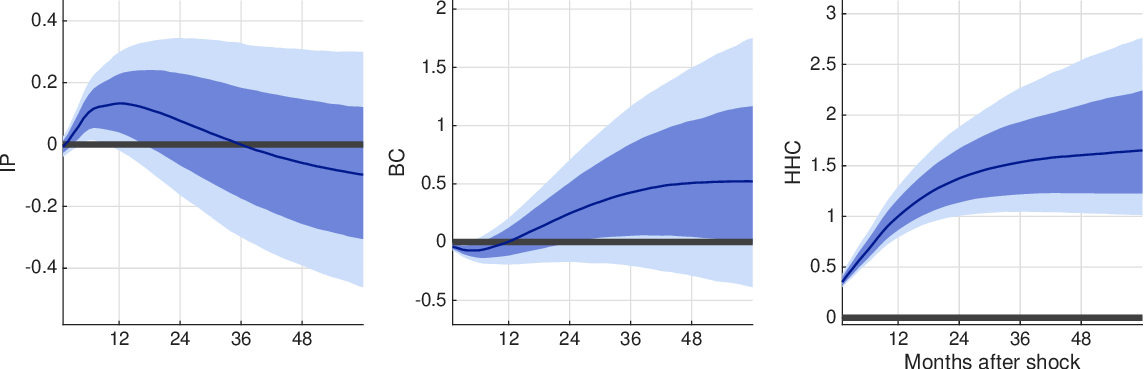}
\caption{Household credit shock and its dynamic impacts. This figure replicates the IP, BC, and HHC responses in the second column of BPSS's Figure 2 using updated monthly data and the monthly Sims--Zha prior. All impulse responses are expressed in percent. The dark band shows the 68\% posterior credible band, and the light band the 90\% posterior credible band. See Table~\ref{tab:mp_model11_variables} for complete variable descriptions.}
\label{fig:hhc_aer_dup}
\end{figure}

\begin{figure}[htbp]
\centering
\includegraphics[width=\textwidth]{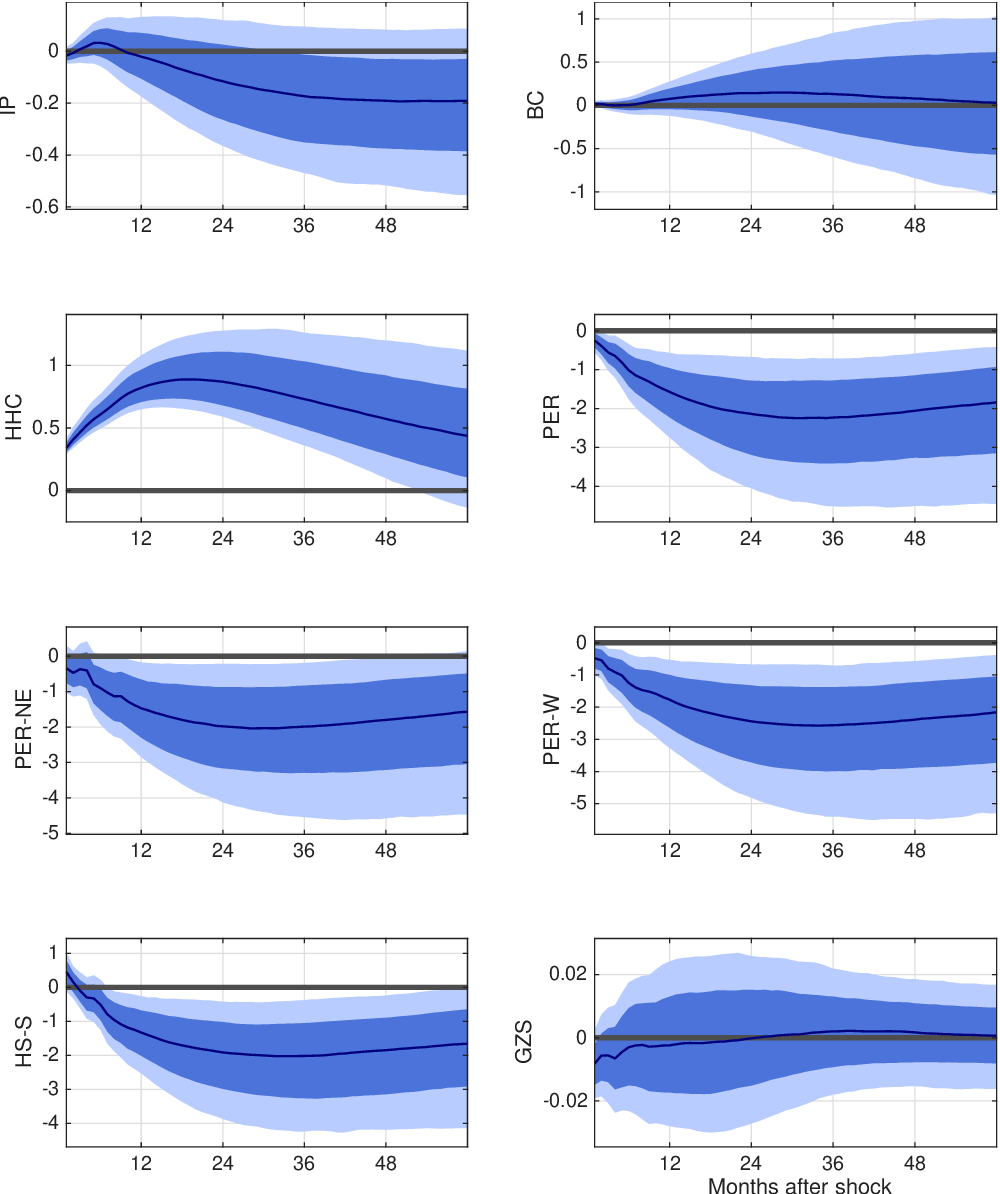}
\caption{Household credit shock and its dynamic effects. The impulse responses are estimated from the selected 13-variable system using the heteroskedasticity-based identification of BPSS. All impulse responses are expressed in percent. The dark band shows the 68\% posterior credible band and the light band the 90\% posterior credible band. See Table~\ref{tab:mp_model11_variables} for complete variable descriptions.}
\label{fig:hhc_t_shk_aer_hhc}
\end{figure}

\begin{figure}[htbp]
\centering
\includegraphics[width=\textwidth]{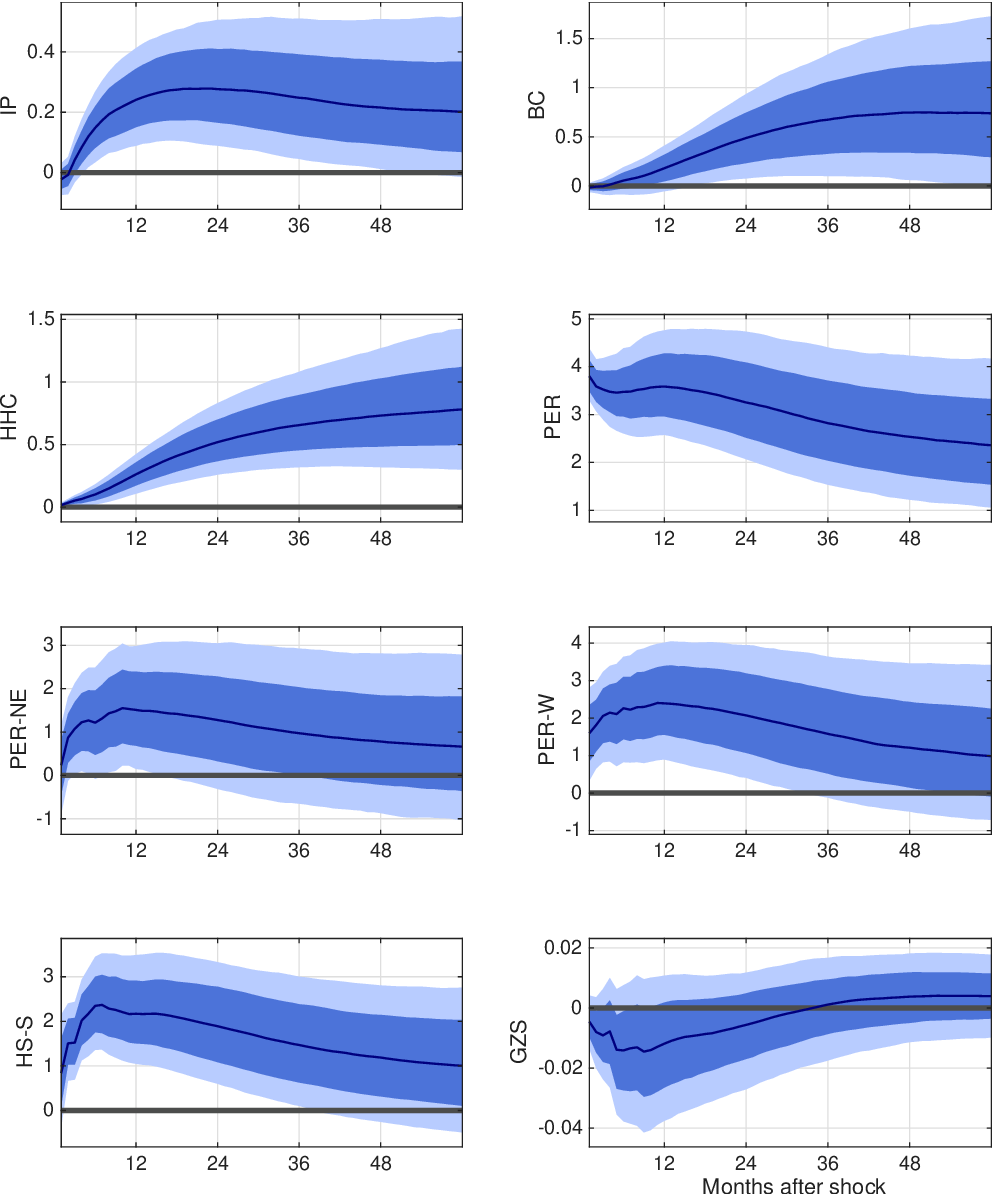}
\caption{Housing production shock and its dynamic impacts. The impulse responses are estimated from the selected 13-variable system using the heteroskedasticity-based identification of BPSS.  All impulse responses are expressed in percent. The dark band shows the 68\% posterior credible band, and the light band the 90\% posterior credible band. See Table~\ref{tab:mp_model11_variables} for complete variable descriptions.}
\label{fig:hhc_t_shk_aer_housing}
\end{figure}

\begin{figure}[htbp]
\centering
\includegraphics[width=\textwidth]{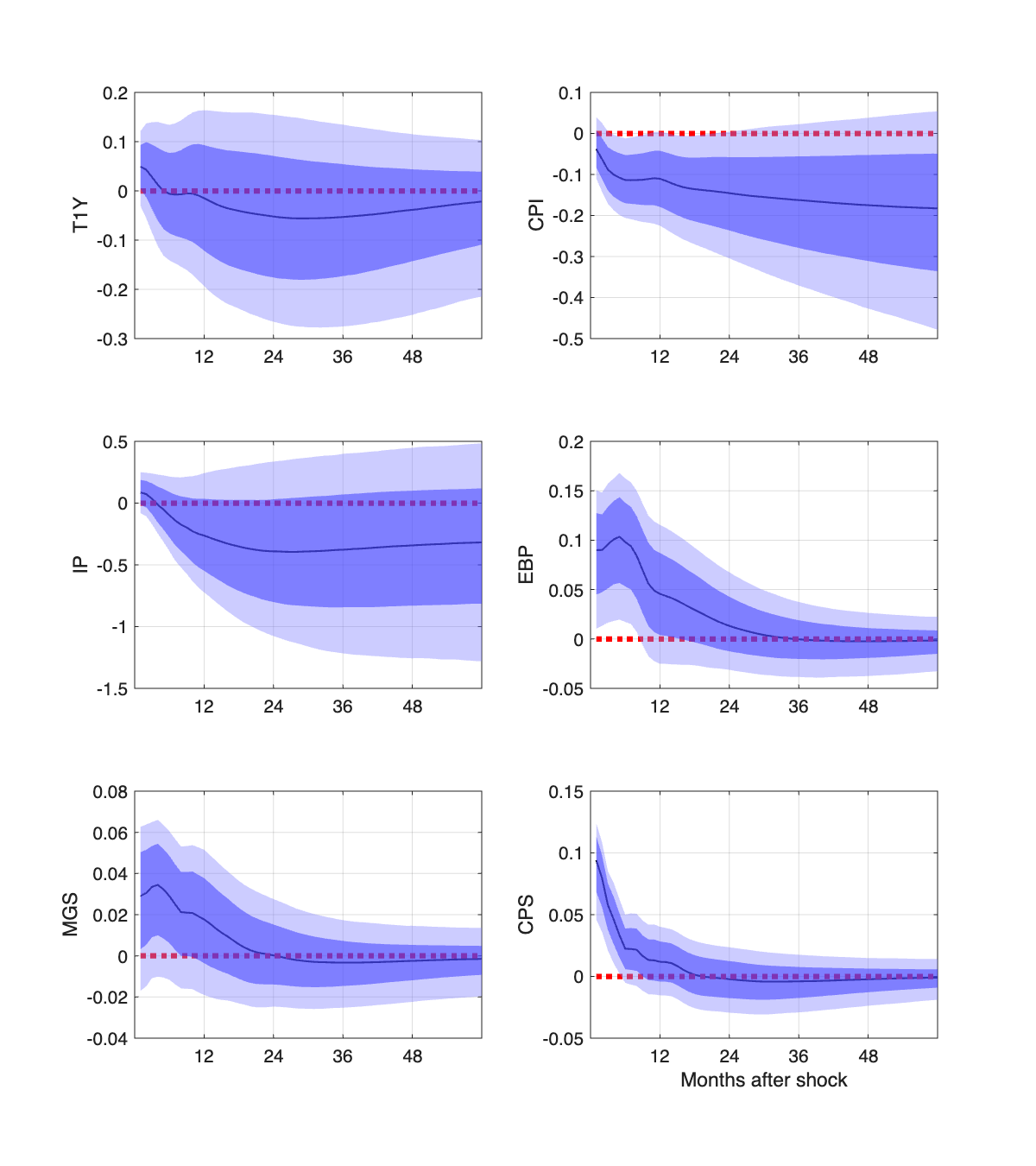}
\caption{Monetary policy shock and its transmission through the core variables. Following GK, all dynamic responses are expressed in percent. ``T1Y'' denotes the one-year Treasury rate, ``CPI'' the consumer price index, ``IP'' industrial production, ``EBP'' the excess bond premium, ``MGS'' the mortgage spread, and ``CPS'' the commercial paper spread. The dark band shows the 68\% posterior credible band and the light band the 90\% posterior credible band. See Table~\ref{tab:mp_model17_variables} for complete variable descriptions.}
\label{fig:mp_core}
\end{figure}

\begin{figure}[htbp]
\centering
\includegraphics[width=\textwidth]{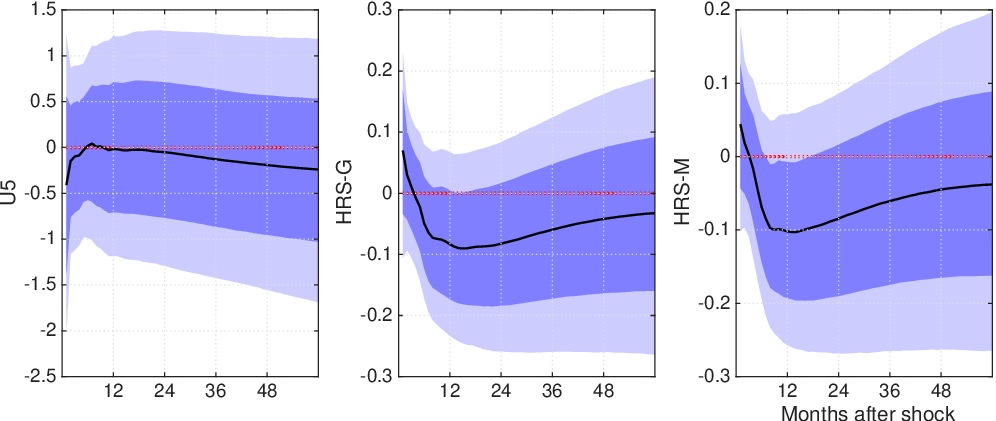}
\caption{Monetary policy shock and its transmission through selected labor market variables. Following GK, all dynamic responses are expressed in percent. ``U5'' denotes unemployed for less than 5 weeks, ``HRS-G'' hours in goods-producing industries, and ``HRS-M'' hours in manufacturing. The dark band shows the 68\% posterior credible band and the light band the 90\% posterior credible band. See Table~\ref{tab:mp_model17_variables} for complete variable descriptions.}
\label{fig:mp_labor}
\end{figure}

\begin{figure}[htbp]
\centering
\includegraphics[width=\textwidth]{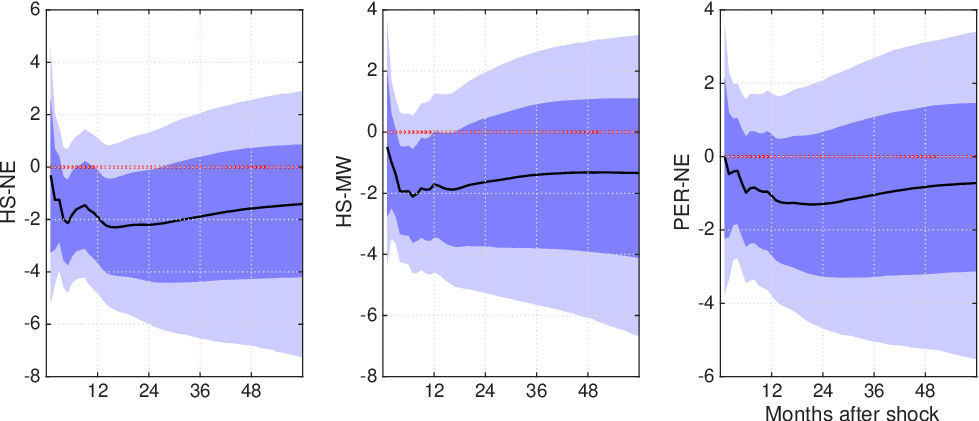}
\caption{Monetary policy shock and its transmission through selected housing variables. Following GK, all dynamic responses are expressed in percent. ``HS-NE'' denotes housing starts in the Northeast, ``HS-MW'' housing starts in the Midwest, and ``PER-NE'' new housing permits in the Northeast. The dark band shows the 68\% posterior credible band and the light band the 90\% posterior credible band. See Table~\ref{tab:mp_model17_variables} for complete variable descriptions.}
\label{fig:mp_housing}
\end{figure}

\begin{figure}[htbp]
\centering
\includegraphics[
  width=\textwidth
]{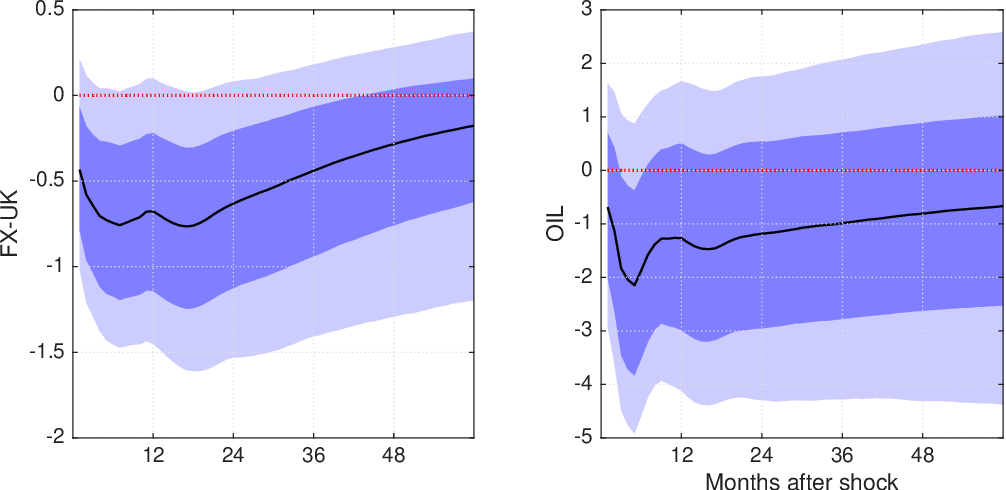}
\caption{Monetary policy shock and its transmission through selected external and commodity variables. Following GK, all dynamic responses are expressed in percent. ``FX-UK'' denotes the U.S./U.K. foreign exchange rate, and ``OIL'' denotes the oil price. The dark band shows the 68\% posterior credible band and the light band the 90\% posterior credible band. See Table~\ref{tab:mp_model17_variables} for complete variable descriptions.}
\label{fig:mp_external_commodity}
\end{figure}

\begin{figure}[htbp]
\centering
\includegraphics[
  width=\textwidth
]{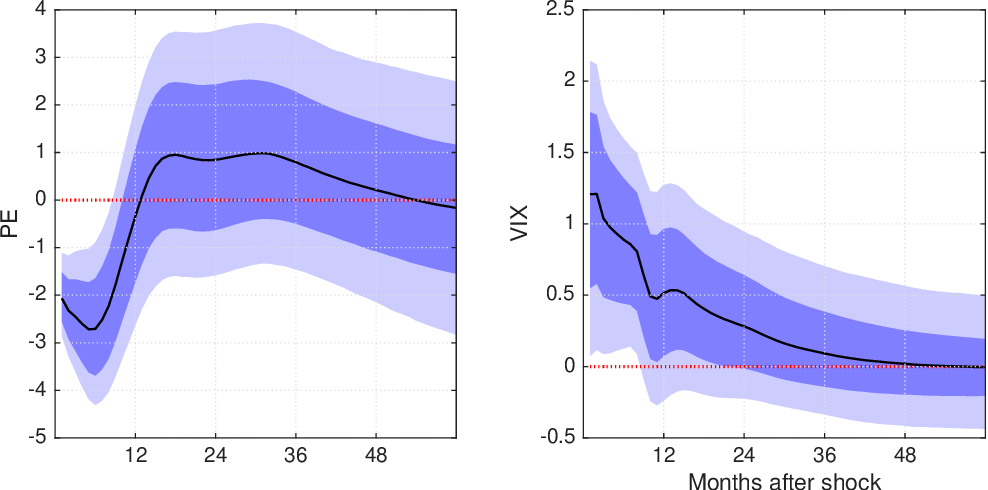}
\caption{Monetary policy shock and its transmission through selected equity and volatility variables. Following GK, all dynamic responses are expressed in percent. ``PE'' denotes the price-earnings ratio, and ``VIX'' denotes stock market volatility. The dark band shows the 68\% posterior credible band and the light band the 90\% posterior credible band. See Table~\ref{tab:mp_model17_variables} for complete variable descriptions.}
\label{fig:mp_equity_uncertainty}
\end{figure}

\begin{figure}[htbp]
\centering
\includegraphics[width=\textwidth]{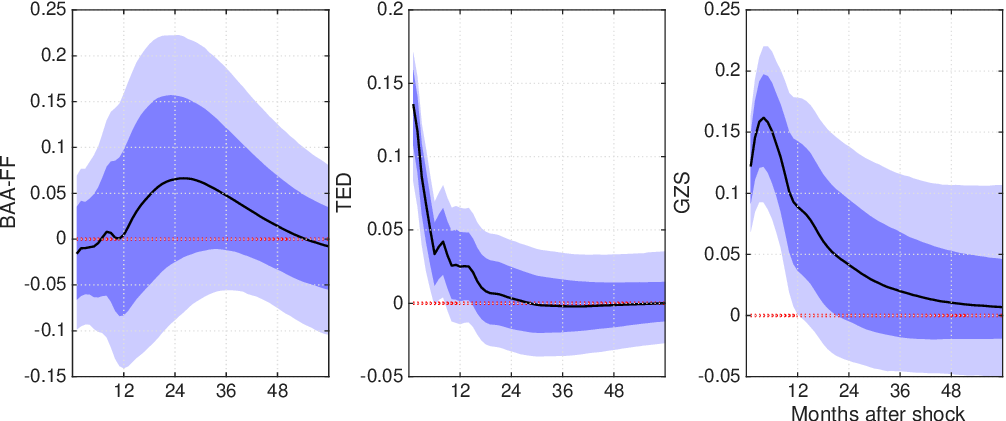}
\caption{Monetary policy shock and its transmission through selected credit and liquidity variables. Following GK, all dynamic responses are expressed in percent. ``BAA-FF'' denotes the Baa corporate bond yield spread over the federal funds rate, ``TED'' the Eurodollar deposit rate spread over the three-month Treasury bill rate, and ``GZS'' the GZ corporate bond spread. The dark band shows the 68\% posterior credible band and the light band the 90\% posterior credible band. See Table~\ref{tab:mp_model17_variables} for complete variable descriptions.}
\label{fig:mp_credit_liquidity}
\end{figure}

\begin{figure}[htbp]
\centering
\includegraphics[
  width=0.6\textwidth
]{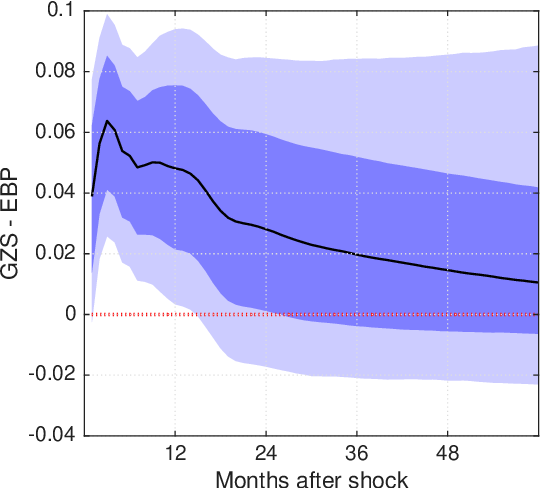}
\caption{Monetary policy shock and its transmission through expected default risk. The responses are estimated from the six-variable core system augmented with the selected GZ spread. ``EBP'' denotes the excess bond premium, and ``GZS'' denotes the GZ corporate bond spread. The difference $\texttt{GZS}-\texttt{EBP}$ measures the spread component attributable to expected default risk. All responses are expressed in percentage points. The dark and light bands depict the 68\% and 90\% posterior credible bands, respectively. See Table~\ref{tab:mp_model17_variables} for complete variable descriptions.}
\label{fig:mp_expected_default}
\end{figure}

\clearpage
\newpage 
\appendix
\setcounter{section}{0}
\setcounter{equation}{0}
\setcounter{figure}{0}
\setcounter{table}{0}
\setcounter{proposition}{0}
\setcounter{lemma}{0}
\setcounter{definition}{0}
\renewcommand{\thesection}{S\arabic{section}}
\renewcommand{\theequation}{S\arabic{equation}}
\renewcommand{\thefigure}{S\arabic{figure}}
\renewcommand{\thetable}{S\arabic{table}}
\renewcommand{\theproposition}{S\arabic{proposition}}
\renewcommand{\thelemma}{S\arabic{lemma}}
\renewcommand{\thedefinition}{S\arabic{definition}}

In this appendix, sections, tables, figures, and equations are labeled with the prefix ``S'' to denote supplemental material.

\section{Proofs for Model Construction and OOS Selection}\label{sec:proofs}

\begin{proof}[Proof of Proposition~\ref{prop:convergence}]
By \eqref{eq:selected_set}, $\widetilde{\cS}^{(k)}(\lambda)$ is a subset of $\widetilde{\cI}^{(k)}=\cN\setminus\cI^{(k)}$. Hence, $\widetilde{\cS}^{(k)}(\lambda)$ and $\cI^{(k)}$ are disjoint. The update in \eqref{eq:update} therefore gives
\[
  \cI^{(k+1)}=\cI^{(k)}\cup\widetilde{\cS}^{(k)}(\lambda),
  \qquad
  |\cI^{(k+1)}|=|\cI^{(k)}|+|\widetilde{\cS}^{(k)}(\lambda)|.
\]
The first equality implies $\cI^{(k)}\subseteq\cI^{(k+1)}$. It also implies that an index already contained in $\cI^{(k)}$ cannot be removed by any later update. If $\widetilde{\cS}^{(k)}(\lambda)\neq\emptyset$, the second equality implies $|\cI^{(k+1)}|>|\cI^{(k)}|$. Every nonterminal update therefore adds at least one index from the finite set $\cN\setminus\cI^{(0)}$, which contains $n-n^{(0)}$ indices. There can be no more than $n-n^{(0)}$ strict updates. The terminal index consequently satisfies $\bar k\leq n-n^{(0)}$. The procedure evaluates the initial system and then evaluates one enlarged system after each strict update. Hence, the number of current systems estimated along the path is $\bar k+1\leq n-n^{(0)}+1$.
\end{proof}

\begin{proof}[Proof of Proposition~\ref{prop:uniqueness}]
The initial index set $\cI^{(0)}$ is fixed by the economic application. Suppose that $\cI^{(k)}$ is uniquely determined. The corresponding variable vector $y_t^{(k)}$, the remaining index set $\widetilde{\cI}^{(k)}$, the index mapping $j\mapsto\iota_j^{(k)}$, and the standardized regressor vector $x_{t,\mathrm{std}}^{(k)}$ are then uniquely determined because the standardization moments are computed once and held fixed. 

Assumption~\ref{ass:well_defined} ensures that the argmax set in \eqref{eq:mode} is nonempty, and the prespecified deterministic rule selects one posterior mode $\widehat\theta^{(k)}$ and hence one fitted composite disturbance vector $\widehat u_t^{(k)}$. The same assumption determines one minimizer of \eqref{eq:l1_obj} for every $q\in\cQ$. The coefficient blocks in \eqref{eq:coef_block}, the selected set $\widetilde{\cS}^{(k)}(\lambda)$ in \eqref{eq:selected_set}, and the updated set $\cI^{(k+1)}$ in \eqref{eq:update} are consequently unique. Induction from $k=0$ establishes uniqueness of the entire path. Proposition~\ref{prop:convergence} establishes finite termination, and the terminal index set and variable vector are therefore unique.
\end{proof}

\begin{proof}[Proof of Lemma~\ref{lem:centering}]
For fixed $\beta_q$, the criterion in \eqref{eq:l1_obj} is a strictly convex quadratic function of $d_q$. Differentiating with respect to $d_q$ and setting the derivative equal to zero gives
\[
  -\frac{1}{T_{\cL}}\bm 1'
  \left(\widehat{\bm u}_q^{(k)}-d_q\bm 1-X^{(k)}\beta_q\right)=0,
\]
which gives the stated expression for $d_q(\beta_q)$. The residual after substitution is
\[
  \widehat{\bm u}_q^{(k)}-d_q(\beta_q)\bm 1-X^{(k)}\beta_q
  =
  M\left(\widehat{\bm u}_q^{(k)}-X^{(k)}\beta_q\right).
\]
The $\ell_1$ term does not involve the intercept. Substitution therefore yields \eqref{eq:centered_auxiliary}, and the minimizing coefficient vector is unchanged.
\end{proof}

\begin{proof}[Proof of Proposition~\ref{prop:stopping_condition}]
If $p_k=0$, no candidate remains, $\widetilde{\cS}^{(k)}(\lambda)=\emptyset$, and the left-hand side of \eqref{eq:stopping_score} is zero by convention. The result therefore holds. Suppose $p_k>0$. By Lemma~\ref{lem:centering}, the estimated coefficients on the regressors are unchanged when the dependent variable and regressors are centered. For equation $q$, the auxiliary optimization is therefore equivalent to
\[
  \min_{\beta_q}
  \left\{
  \frac{1}{2T_{\cL}}
  \left\|
    M\widehat{\bm u}_q^{(k)}
    -
    MX^{(k)}\beta_q
  \right\|_2^2
  +
  \lambda\|\beta_q\|_1
  \right\}.
\]
The objective is convex. The vector $\beta_q=0$ is a minimizer if and only if zero belongs to its subdifferential at $\beta_q=0$. This condition is
\[
  0
  \in
  -\frac{1}{T_{\cL}}(X^{(k)})'M\widehat{\bm u}_q^{(k)}
  +
  \lambda\,\partial\|\beta_q\|_1\big|_{\beta_q=0}.
\]
Since
\[
  \partial\|\beta_q\|_1\big|_{\beta_q=0}
  =
  \left\{
    z:\|z\|_\infty\leq 1
  \right\},
\]
the zero vector is a minimizer if and only if
\[
  \left\|
    \frac{1}{T_{\cL}}(X^{(k)})'M\widehat{\bm u}_q^{(k)}
  \right\|_\infty
  \leq
  \lambda.
\]
Under Assumption~\ref{ass:well_defined}, the auxiliary minimizer is unique. Hence, the condition above is equivalent to $\widehat\beta_q^{(k)}=0$. By \eqref{eq:selected_set}, no remaining candidate is admitted if and only if $\widehat\beta_q^{(k)}=0$ for every $q \in \cQ$. Therefore,
\[
  \widetilde{\cS}^{(k)}(\lambda)=\emptyset
\]
if and only if
\[
  \max_{q \in \cQ}
  \left\|
    \frac{1}{T_{\cL}}(X^{(k)})'M\widehat{\bm u}_q^{(k)}
  \right\|_\infty
  \leq
  \lambda.
\]
\end{proof}   
   
\begin{proof}[Proof of Proposition~\ref{prop:unit_invariance}]
For a transformed remaining candidate $y_{i,t}^{\mathrm{new}}=\alpha_i+d_i y_{i,t}$, its mean and standard deviation satisfy
\[
  \bar y_i^{\mathrm{new}}
  =
  \alpha_i+d_i\bar y_i,
  \qquad
  \widehat\sigma_i^{\mathrm{new}}
  =
  |d_i|\widehat\sigma_i.
\]
Its standardized value is therefore
\[
  \frac{y_{i,t}^{\mathrm{new}}-\bar y_i^{\mathrm{new}}}{\widehat\sigma_i^{\mathrm{new}}}
  =
  \sgn(d_i)y_{i,t}^{\mathrm{std}}.
\]
Consequently, the new standardized regressor matrix has the form $X_{\mathrm{new}}^{(k)}=X^{(k)}H$, where $H$ is diagonal, every diagonal entry equals $1$ or $-1$, and the same sign is repeated for the contemporaneous value and lags of a given variable. The matrix satisfies $H^{-1}=H$ and preserves the $\ell_1$ norm. For any coefficient vector $\beta_q$, define $\beta_{q,\mathrm{new}}=H\beta_q$. Then
\[
  X_{\mathrm{new}}^{(k)}\beta_{q,\mathrm{new}}
  =
  X^{(k)}HH\beta_q
  =
  X^{(k)}\beta_q,
  \qquad
  \|\beta_{q,\mathrm{new}}\|_1
  =
  \|\beta_q\|_1.
\]
The original and transformed auxiliary criteria consequently have the same values under this one to one mapping. Uniqueness in Assumption~\ref{ass:well_defined} implies that their minimizers are related by $\widehat\beta_{q,\mathrm{new}}=H\widehat\beta_q$. Multiplication by $H$ changes only coefficient signs. Every coefficient block is zero before the transformation if and only if the corresponding block is zero after the transformation. Equation~\eqref{eq:selected_set} therefore gives the same selected variable indices.
\end{proof}

\begin{proof}[Proof of Proposition~\ref{prop:update_stability}]
Assumption~\ref{ass:well_defined} ensures that the argmax set in \eqref{eq:mode} is nonempty and that the auxiliary minimizers $\widehat\beta_q^{(k)}$ are unique for every $q\in\cQ$. The prespecified deterministic rule selects one posterior mode $\widehat\theta^{(k)}$ and hence one fitted composite disturbance vector. Each nonzero coefficient index set $\cA_q^{(k)}$ is therefore uniquely determined before the perturbation. Assumption~\ref{ass:selection_margins} provides the conditions under which these index sets are preserved under sufficiently small changes in the fitted composite disturbances.

If $p_k=0$, no regressor remains in the auxiliary optimization. Every nonzero coefficient index set and the selected variable set are empty, and the next index set is unchanged under any perturbation. The conclusion holds for any $\epsilon_k>0$. Suppose $p_k>0$.

Fix an equation $q$ and write
\[
  A=\cA_q^{(k)}.
\]
First suppose that $A$ is nonempty. Let
\[
  s_A=\sgn(\widehat\beta_{q,A}^{(k)})
\]
and define
\[
  Q_A
  =
  \frac{1}{T_{\cL}}
  (\widetilde X_A^{(k)})'
  \widetilde X_A^{(k)}.
\]
Assumption~\ref{ass:selection_margins} makes $Q_A$ positive definite and therefore invertible. The first order condition for the original auxiliary problem on the nonzero coordinates is
\[
  \frac{1}{T_{\cL}}
  (\widetilde X_A^{(k)})'
  \left(
    \widetilde{\bm u}_q^{(k)}
    -
    \widetilde X_A^{(k)}
    \widehat\beta_{q,A}^{(k)}
  \right)
  =
  \lambda s_A.
\]

Replace $\widehat{\bm u}_q^{(k)}$ by $\widehat{\bm u}_q^{(k)}+h_q$. After centering, the dependent variable becomes $\widetilde{\bm u}_q^{(k)}+Mh_q$. Holding the coefficient support equal to $A$ and the sign vector equal to $s_A$, the first order condition for the perturbed problem is
\[
  \frac{1}{T_{\cL}}
  (\widetilde X_A^{(k)})'
  \left(
    \widetilde{\bm u}_q^{(k)}
    +Mh_q
    -
    \widetilde X_A^{(k)}\beta_A(h_q)
  \right)
  =
  \lambda s_A.
\]
Subtracting the original first order condition and using the definition of $Q_A$ gives
\begin{equation}
\label{eq:restricted_perturbation}
  \beta_A(h_q)
  =
  \widehat\beta_{q,A}^{(k)}
  +
  Q_A^{-1}
  \frac{1}{T_{\cL}}
  (\widetilde X_A^{(k)})'Mh_q.
\end{equation}

The mapping $h_q\mapsto\beta_A(h_q)$ is linear and continuous, and
\[
  \beta_A(0)=\widehat\beta_{q,A}^{(k)}.
\]
Assumption~\ref{ass:selection_margins} gives
\[
  \min_{j\in A}
  \left|
    \widehat\beta_{q,j}^{(k)}
  \right|
  \geq b_q>0.
\]
Consequently, there exists $\epsilon_{q,1}>0$ such that
\[
  \|h_q\|_2<\epsilon_{q,1}
\]
implies
\[
  \sgn\bigl(\beta_A(h_q)\bigr)=s_A.
\]
Thus, every coefficient indexed by $A$ remains nonzero and retains its original sign.

Define the residual associated with \eqref{eq:restricted_perturbation} by
\[
  \bm v_q(h_q)
  =
  \widetilde{\bm u}_q^{(k)}
  +Mh_q
  -
  \widetilde X_A^{(k)}\beta_A(h_q).
\]
If $A^c$ is nonempty, define for each index $j\notin A$ the dual quantity
\[
  d_j(h_q)
  =
  \frac{1}{T_{\cL}}
  (\widetilde X_j^{(k)})'\bm v_q(h_q).
\]
Equation~\eqref{eq:restricted_perturbation} implies that each mapping $h_q\mapsto d_j(h_q)$ is continuous. Assumption~\ref{ass:selection_margins} gives the strict inactive-coordinate bound
\[
  |d_j(0)|
  \leq
  \lambda-\delta_q
  \qquad
  \text{for every }j\notin A,
\]
where $\delta_q>0$. Because there are finitely many inactive coordinates, there exists $\epsilon_{q,2}>0$ such that
\[
  \|h_q\|_2<\epsilon_{q,2}
\]
implies
\[
  |d_j(h_q)|<\lambda
  \qquad
  \text{for every }j\notin A.
\]
If $A^c$ is empty, set $\epsilon_{q,2}=+\infty$.

Let
\[
  \epsilon_q^{A}
  =
  \min\{\epsilon_{q,1},\epsilon_{q,2}\}.
\]
When $\|h_q\|_2<\epsilon_q^{A}$, the coefficient vector whose entries equal $\beta_A(h_q)$ on $A$ and zero outside $A$ satisfies the complete optimality conditions for the perturbed auxiliary problem. On $A$, the first order condition holds with sign vector $s_A$. Outside $A$, the strict subgradient inequalities hold.

It remains to establish that this perturbed minimizer is unique. Let $\beta(h_q)$ denote the coefficient vector just constructed, and define
\[
  z(h_q)
  =
  \frac{1}{\lambda T_{\cL}}
  (\widetilde X^{(k)})'\bm v_q(h_q).
\]
The optimality conditions give
\[
  z_A(h_q)=s_A.
\]
When $A^c$ is nonempty, they also give
\[
  \|z_{A^c}(h_q)\|_\infty<1.
\]
For any coefficient vector $\beta$, convexity of the squared-error term and the subgradient inequality for the $\ell_1$ norm imply
\begin{align*}
  \cL_q(\beta;h_q)-\cL_q(\beta(h_q);h_q)
  &=
  \frac{1}{2T_{\cL}}
  \left\|
    \widetilde X^{(k)}
    \bigl(\beta-\beta(h_q)\bigr)
  \right\|_2^2\\
  &\quad
  +
  \lambda
  \left[
    \|\beta\|_1
    -
    \|\beta(h_q)\|_1
    -
    z(h_q)'
    \bigl(\beta-\beta(h_q)\bigr)
  \right],
\end{align*}
where $\cL_q(\cdot;h_q)$ denotes the perturbed auxiliary criterion for equation $q$. Both terms on the right-hand side are nonnegative. When $A^c$ is nonempty, the strict inequality $\|z_{A^c}(h_q)\|_\infty<1$ implies that equality can hold only if
\[
  \beta_{A^c}=0.
\]
When $A^c$ is empty, this restriction holds trivially. Conditional on this restriction, equality in the first term requires
\[
  \widetilde X_A^{(k)}
  \bigl(\beta_A-\beta_A(h_q)\bigr)=0.
\]
Since $Q_A$ is positive definite, $\widetilde X_A^{(k)}$ has full column rank. Hence,
\[
  \beta_A=\beta_A(h_q).
\]
The perturbed minimizer is therefore unique and has the same nonzero coefficient index set $A$.

Now suppose that $\cA_q^{(k)}$ is empty. In this case, the original minimizer is the zero vector. Define the perturbed score
\[
  g_q(h_q)
  =
  \frac{1}{T_{\cL}}
  (\widetilde X^{(k)})'
  \left(
    \widetilde{\bm u}_q^{(k)}+Mh_q
  \right).
\]
Assumption~\ref{ass:selection_margins} gives
\[
  \|g_q(0)\|_\infty
  \leq
  \lambda-\delta_q
\]
for some $\delta_q>0$. The mapping $h_q\mapsto g_q(h_q)$ is continuous. There therefore exists $\epsilon_q^{0}>0$ such that
\[
  \|h_q\|_2<\epsilon_q^{0}
\]
implies
\[
  \|g_q(h_q)\|_\infty<\lambda.
\]

For any nonzero coefficient vector $\beta$, the difference between the perturbed criterion at $\beta$ and at zero is
\begin{align*}
  \cL_q(\beta;h_q)-\cL_q(0;h_q)
  &=
  \frac{1}{2T_{\cL}}
  \left\|
    \widetilde X^{(k)}\beta
  \right\|_2^2
  -
  g_q(h_q)'\beta
  +
  \lambda\|\beta\|_1\\
  &\geq
  \frac{1}{2T_{\cL}}
  \left\|
    \widetilde X^{(k)}\beta
  \right\|_2^2
  +
  \left(
    \lambda-\|g_q(h_q)\|_\infty
  \right)\|\beta\|_1\\
  &>
  0.
\end{align*}
The zero vector therefore remains the unique minimizer, and its nonzero coefficient index set remains empty.

For each equation $q$, let $\epsilon_q$ equal $\epsilon_q^{A}$ when $\cA_q^{(k)}$ is nonempty and $\epsilon_q^{0}$ when $\cA_q^{(k)}$ is empty. Each $\epsilon_q$ is strictly positive. Since the current system contains finitely many equations, define
\[
  \epsilon_k
  =
  \min_{q \in \cQ}\epsilon_q
  >
  0.
\]
If $\|h_q\|_2<\epsilon_k$ for every equation $q$, each nonzero coefficient index set $\cA_q^{(k)}$ remains unchanged. The coefficient blocks in \eqref{eq:coef_block} are therefore unchanged in their zero and nonzero pattern. It follows from \eqref{eq:selected_set} that $\widetilde{\cS}^{(k)}(\lambda)$ is unchanged, and \eqref{eq:update} then implies that $\cI^{(k+1)}$ is unchanged.
\end{proof}   
   
\begin{proof}[Proof of Corollary~\ref{cor:path_stability}]
Let $\{\cI^{(k)}\}_{k=0}^{\bar k}$ denote the original construction path, and let $\{\cI_h^{(k)}\}$ denote the path generated by the perturbed fitted composite disturbances. Both procedures begin from the same core index set. Hence,
\[
  \cI_h^{(0)}=\cI^{(0)}.
\]

Suppose for some $k\leq\bar k$ that
\[
  \cI_h^{(k)}=\cI^{(k)}.
\]
The current system and the set of remaining candidates are then the same along the two paths. Because
\[
  \|h_q^{(k)}\|_2<\epsilon_k
  \qquad
  \text{for every }q \in \cQ,
\]
Proposition~\ref{prop:update_stability} implies
\[
  \widetilde{\cS}_h^{(k)}(\lambda)
  =
  \widetilde{\cS}^{(k)}(\lambda).
\]
Applying the update rule \eqref{eq:update} therefore gives
\[
  \cI_h^{(k+1)}
  =
  \cI^{(k+1)}
\]
whenever the original procedure continues beyond iteration $k$. By induction, the perturbed and original index sets coincide at every iteration through $\bar k$.

It remains to show that the terminal iteration is also unchanged. For every $k<\bar k$, the definition of $\bar k$ gives
\[
  \widetilde{\cS}^{(k)}(\lambda)\neq\emptyset.
\]
Since the perturbed selected set is identical, the perturbed procedure cannot terminate before $\bar k$. At iteration $\bar k$,
\[
  \widetilde{\cS}^{(\bar k)}(\lambda)=\emptyset.
\]
Proposition~\ref{prop:update_stability} gives
\[
  \widetilde{\cS}_h^{(\bar k)}(\lambda)=\emptyset,
\]
Thus, the perturbed procedure also terminates at $\bar k$. Proposition~\ref{prop:convergence} ensures that $\bar k$ is finite. The complete construction path and its terminal index set are therefore unchanged. Hence,
\[
  y_{t,h}^\dagger(\lambda)
  =
  y_t^\dagger(\lambda).
\]
\end{proof}

\begin{proof}[Proof of Lemma~\ref{lem:perloss}]
Fix the loss $g$. By Assumption~\ref{ass:oos_regularity}(iii), for every $\lambda\in\Lambda$ the posterior mean loss
\[
  \overline{\cL}^{(T)}_{g}(\lambda)
  =
  \frac{1}{S}\sum_{s=1}^{S}\cL^{(s,T)}_{g}(\lambda)
\]
is a finite average of finite numbers and hence a finite real number. The image $\{\overline{\cL}^{(T)}_{g}(\lambda):\lambda\in\Lambda\}$ is a finite subset of $\mathbb{R}$ and therefore attains a minimum $m_{g}$. The minimizer set
\[
  M_{g}
  =
  \left\{
    \lambda\in\Lambda:
    \overline{\cL}^{(T)}_{g}(\lambda)=m_{g}
  \right\}
\]
is nonempty and finite. By the strict total order $\prec$ in Assumption~\ref{ass:oos_regularity}(i), $M_{g}$ has a unique $\prec$-least element. Set
\[
  \lambda^{\ast}_{g}
  =
  \min_{\prec}M_{g},
\]
as in \eqref{eq:benchmark} applied under loss $g$. This proves existence and uniqueness of the reference model complexity parameter. Assumption~\ref{ass:oos_regularity}(ii) then ensures that the associated reference terminal system is also unique.

Given $\lambda^{\ast}_{g}$, for every $\lambda\in\Lambda\setminus\{\lambda^{\ast}_{g}\}$ the paired differences in \eqref{eq:loss_diff} are
\[
  \Delta^{(s,T)}_{g}(\lambda)
  =
  \cL^{(s,T)}_{g}(\lambda)
  -
  \cL^{(s,T)}_{g}(\lambda^{\ast}_{g}),
  \qquad
  s=1,\ldots,S.
\]
The common index $s$ labels a pair consisting of one posterior draw from each terminal system. By Assumption~\ref{ass:oos_regularity}(iii), both loss values in each pair are finite. Hence, $\Delta^{(s,T)}_{g}(\lambda)$ is finite for every $\lambda\in\Lambda\setminus\{\lambda^{\ast}_{g}\}$ and every $s=1,\ldots,S$.

Because $\alpha\in(0,1)$,
\[
  \frac{1-\alpha}{2}\in\left(0,\frac{1}{2}\right),
  \qquad
  \frac{1+\alpha}{2}\in\left(\frac{1}{2},1\right).
\]
By Assumption~\ref{ass:oos_regularity}(iv), the endpoints $q^{g}_{(1-\alpha)/2}(\lambda)$ and $q^{g}_{(1+\alpha)/2}(\lambda)$ are uniquely determined for every $\lambda\in\Lambda\setminus\{\lambda^{\ast}_{g}\}$. Consequently, each of the conditions
\[
  q^{g}_{(1+\alpha)/2}(\lambda)<0
  \qquad\text{and}\qquad
  q^{g}_{(1-\alpha)/2}(\lambda)
  \leq 0
  \leq
  q^{g}_{(1+\alpha)/2}(\lambda)
\]
is either true or false for every such $\lambda$, with no indeterminacy. The sets
\[
  \Lambda_{\alpha}^{g,-}
  =
  \left\{
    \lambda\in\Lambda\setminus\{\lambda^{\ast}_{g}\}:
    q^{g}_{(1+\alpha)/2}(\lambda)<0
  \right\}
\]
and
\[
  \Lambda_{\alpha}^{g,0}
  =
  \{\lambda^{\ast}_{g}\}
  \cup
  \left\{
    \lambda\in\Lambda\setminus\{\lambda^{\ast}_{g}\}:
    q^{g}_{(1-\alpha)/2}(\lambda)
    \leq 0
    \leq
    q^{g}_{(1+\alpha)/2}(\lambda)
  \right\}
\]
are therefore uniquely determined subsets of $\Lambda$.

The explicit inclusion of $\lambda^{\ast}_{g}$ in $\Lambda_{\alpha}^{g,0}$ reflects that the OOS forecast loss difference between the reference system and itself is identically zero. No empirical quantile comparison is required for the reference system. Hence, $\lambda^{\ast}_{g}\in\Lambda_{\alpha}^{g,0}$ and $\Lambda_{\alpha}^{g,0}\neq\emptyset$.

Finally, by \eqref{eq:equiv_set} applied under loss $g$,
\[
  \Lambda_{\alpha}^{g}
  =
  \begin{cases}
    \Lambda_{\alpha}^{g,-}, & \Lambda_{\alpha}^{g,-}\neq\emptyset,\\[3pt]
    \Lambda_{\alpha}^{g,0}, & \Lambda_{\alpha}^{g,-}=\emptyset.
  \end{cases}
\]
In the first case, $\Lambda_{\alpha}^{g}=\Lambda_{\alpha}^{g,-}$ is nonempty by the case condition. In the second case, $\Lambda_{\alpha}^{g}=\Lambda_{\alpha}^{g,0}$ is nonempty because it contains $\lambda^{\ast}_{g}$. Both branches are uniquely determined by the sets constructed above. Therefore, $\Lambda_{\alpha}^{g}$ exists, is unique, and is nonempty.
\end{proof}

\begin{proof}[Proof of Proposition~\ref{prop:welldefined}]
Apply Lemma~\ref{lem:perloss} first with $g=\mathrm{SE}$ and then with $g=\mathrm{AE}$. The retained sets $\Lambda_{\alpha}^{\mathrm{SE}}$ and $\Lambda_{\alpha}^{\mathrm{AE}}$ therefore exist, are unique, and are nonempty. Their union
\[
  \Lambda_{\alpha}^{\cup}
  =
  \Lambda_{\alpha}^{\mathrm{SE}}
  \cup
  \Lambda_{\alpha}^{\mathrm{AE}}
\]
is consequently a uniquely determined subset of $\Lambda$. It is nonempty because it contains the nonempty sets $\Lambda_{\alpha}^{\mathrm{SE}}$ and $\Lambda_{\alpha}^{\mathrm{AE}}$. It is finite because it is a subset of the finite grid $\Lambda$. Hence, the combined retained set $\Lambda_{\alpha}^{\cup}$ exists, is unique, is nonempty, and is finite.

By Assumption~\ref{ass:oos_regularity}(ii), the terminal dimension $n^\dagger(\lambda)$ is well defined for every $\lambda\in\Lambda$. It is therefore well defined for every $\lambda\in\Lambda_{\alpha}^{\cup}$. Because $\Lambda_{\alpha}^{\cup}$ is finite and nonempty, the image
\[
  \left\{
    n^\dagger(\lambda):
    \lambda\in\Lambda_{\alpha}^{\cup}
  \right\}
\]
is a finite nonempty set of nonnegative integers. Every finite nonempty set of integers has a unique largest value. Thus,
\[
  n_{\alpha,\cup}^{\max}
  =
  \max_{\lambda\in\Lambda_{\alpha}^{\cup}}
  n^\dagger(\lambda)
\]
exists and is unique.

Given the uniquely determined value $n_{\alpha,\cup}^{\max}$, define
\[
  \Lambda_{\alpha,\cup}^{\max}
  =
  \left\{
    \lambda\in\Lambda_{\alpha}^{\cup}:
    n^\dagger(\lambda)=n_{\alpha,\cup}^{\max}
  \right\}.
\]
This set is uniquely determined because both $\Lambda_{\alpha}^{\cup}$ and the function $\lambda\mapsto n^\dagger(\lambda)$ are uniquely determined. It is nonempty because the maximum $n_{\alpha,\cup}^{\max}$ is attained by at least one member of the finite nonempty set $\Lambda_{\alpha}^{\cup}$. It is finite because it is a subset of the finite set $\Lambda_{\alpha}^{\cup}$. Hence, $\Lambda_{\alpha,\cup}^{\max}$ exists, is unique, is nonempty, and is finite.

The strict total order $\prec$ on $\Lambda$ in Assumption~\ref{ass:oos_regularity}(i) restricts to a strict total order on the finite nonempty subset $\Lambda_{\alpha,\cup}^{\max}$. Every finite nonempty set endowed with a strict total order has a unique least element. Therefore,
\[
  \widehat\lambda_{\alpha}
  =
  \min_{\prec}\Lambda_{\alpha,\cup}^{\max}
\]
exists and is unique.

Finally, Assumption~\ref{ass:oos_regularity}(ii) states that the terminal variable vector $y_t^\dagger(\lambda)$ is unique for every $\lambda\in\Lambda$. Since $\widehat\lambda_{\alpha}\in\Lambda_{\alpha,\cup}^{\max}\subseteq\Lambda$ exists and is unique, the terminal system $y_t^\dagger(\widehat\lambda_{\alpha})$ also exists and is unique. Thus, the selected terminal system is uniquely determined by the data and the prespecified procedure.
\end{proof}

\section{Proxy Identification and the Joint Bayesian Method}\label{sec:mp_appendix}

This appendix proves the results that Section~\ref{sec:mp} defers to it. It proves Proposition~\ref{prop:mp_anchor} and Theorem~\ref{thm:mp_identification}, derives the two-stage coefficients in \eqref{eq:mp_gk_ratio} and proves Proposition~\ref{prop:mp_gk_equivalence}, shows the rotational invariance of the remaining equations in posterior computation, and states the efficiency comparison under correct parametric specification. 

\begin{proof}[Proof of Proposition~\ref{prop:mp_anchor}]
By \eqref{eq:mp_reduced_form},
\[
  v_t
  =
  A_0^{-1\,\prime}\eps_t.
\]
Therefore,
\[
  G
  \equiv
  \mathbb E(v_tm_t')
  =
  A_0^{-1\,\prime}\mathbb E(\eps_tm_t').
\]

Assumption~\ref{ass:mp_proxy} restricts the rows of $\mathbb E(\eps_tm_t')$. Its $j$th row is
\[
  \mathbb E(\eps_{jt}m_t')
  =
  \bigl(\mathbb E(m_t\eps_{jt})\bigr)'.
\]
By \eqref{eq:mp_proxy_moments}, the first row equals $\gamma_m'$, while every remaining row equals $0_{1\times n_m}$. Stacking these rows gives
\[
  \mathbb E(\eps_tm_t')
  =
  \begin{pmatrix}
    \gamma_m'\\
    0_{1\times n_m}\\
    \vdots\\
    0_{1\times n_m}
  \end{pmatrix}
  =
  e_1\gamma_m'.
\]
This result holds for every $n_m\geq1$, because Assumption~\ref{ass:mp_proxy} restricts which shocks may be correlated with the instruments, not how many instruments there are.

Using the definition of $s$ in \eqref{eq:mp_a1_s},
\[
  G
  =
  A_0^{-1\,\prime}e_1\gamma_m'
  =
  s\gamma_m'.
\]
Recall that
\[
  G'
  =
  [g_1\ \cdots\ g_n].
\]
The $i$th row of $G=s\gamma_m'$ is $s_i\gamma_m'$. Hence,
\[
  g_i
  =
  \gamma_ms_i,
  \qquad
  i=1,\ldots,n.
\]

Let $p$ be an admissible anchor. By definition,
\[
  g_p
  \neq
  0_{n_m\times1}.
\]
Since $g_p=\gamma_ms_p$ and $\gamma_m\neq0_{n_m\times1}$ under Assumption~\ref{ass:mp_proxy}, admissibility implies
\[
  s_p
  \neq
  0.
\]

Let
\[
  W
  =
  \mathbb E(m_tm_t')^{-1}.
\]
Assumption~\ref{ass:mp_proxy} states that $\mathbb E(m_tm_t')$ is positive definite. Hence, $W$ is positive definite and
\[
  \gamma_m'W\gamma_m
  >
  0.
\]

By \eqref{eq:mp_gk_ratio}, the two-stage coefficient for variable $i$ using anchor $p$ is
\[
  \beta_{i\mid p}
  =
  \frac{g_p'Wg_i}{g_p'Wg_p}.
\]
Substituting $g_i=\gamma_ms_i$ and $g_p=\gamma_ms_p$ gives
\[
  \beta_{i\mid p}
  =
  \frac{s_ps_i\gamma_m'W\gamma_m}
       {s_p^2\gamma_m'W\gamma_m}
  =
  \frac{s_i}{s_p}.
\]
In particular,
\[
  \beta_{p\mid p}
  =
  1.
\]
Collecting these coefficients gives
\[
  d_p
  =
  \begin{pmatrix}
    \beta_{1\mid p}\\
    \vdots\\
    \beta_{n\mid p}
  \end{pmatrix}
  =
  \frac{1}{s_p}
  \begin{pmatrix}
    s_1\\
    \vdots\\
    s_n
  \end{pmatrix}
  =
  \frac{s}{s_p}.
\]

Let $r$ be another admissible anchor. The same argument gives
\[
  d_r
  =
  \frac{s}{s_r}.
\]
Because $s_p\neq0$ and $s_r\neq0$,
\[
  d_r
  =
  \frac{s}{s_r}
  =
  \frac{s_p}{s_r}\frac{s}{s_p}
  =
  \frac{s_p}{s_r}d_p.
\]
Thus, ratio vectors obtained from different admissible anchors are proportional. The anchor determines which component of the ratio vector is normalized to one.

It remains to show that the scale normalization removes this anchor-specific representation. By \eqref{eq:mp_reduced_form},
\[
  \Sigma
  =
  A_0^{-1\,\prime}A_0^{-1},
\]
and hence
\[
  \Sigma^{-1}
  =
  A_0A_0'.
\]
Using $s=A_0^{-1\,\prime}e_1$ from \eqref{eq:mp_a1_s},
\[
  s'\Sigma^{-1}s
  =
  e_1'A_0^{-1}A_0A_0'A_0^{-1\,\prime}e_1
  =
  e_1'e_1
  =
  1.
\]
Since $d_p=s/s_p$,
\[
  d_p'\Sigma^{-1}d_p
  =
  \frac{1}{s_p^2}s'\Sigma^{-1}s
  =
  \frac{1}{s_p^2}.
\]
Therefore,
\[
  \sqrt{d_p'\Sigma^{-1}d_p}
  =
  \frac{1}{|s_p|}.
\]
The unit-variance normalization in \eqref{eq:mp_gk_scale} then gives
\[
  \frac{d_p}{\sqrt{d_p'\Sigma^{-1}d_p}}
  =
  \frac{s/s_p}{1/|s_p|}
  =
  \frac{|s_p|}{s_p}s
  =
  \sgn(s_p)s.
\]
Thus, every admissible anchor recovers the same impact vector up to sign. Applying the same economic sign normalization selects the common orientation $s$ for every anchor.

The matrices $\Psi_h$ in \eqref{eq:mp_Psi} are generated by the reduced-form matrices $B_\ell=A_\ell A_0^{-1}$ in \eqref{eq:mp_B}, which are identified from the reduced form without reference to any anchor. The impulse responses therefore satisfy
\[
  \mathrm{IRF}_1(h)
  =
  s'\Psi_h.
\]
Because every admissible anchor recovers the same sign-normalized impact vector $s$, every admissible anchor produces the same impulse responses at every horizon. Moreover, $a_1=\Sigma^{-1}s$ by \eqref{eq:mp_a1_s}, so every admissible anchor also recovers the same first equation.
\end{proof}

\begin{proof}[Proof of Theorem~\ref{thm:mp_identification}]
Transposing \eqref{eq:mp_reduced_form} gives
\[
  v_t
  =
  A_0^{-1\prime}\eps_t.
\]
The maintained exclusion restrictions imply
\[
  \mathbb E(\eps_tm_t')
  =
  e_1\gamma_m'.
\]
Therefore,
\[
  G
  =
  \mathbb E(v_tm_t')
  =
  A_0^{-1\prime}\mathbb E(\eps_tm_t')
  =
  A_0^{-1\prime}e_1\gamma_m'
  =
  s\gamma_m'.
\]
Since
\[
  \Sigma^{-1}
  =
  A_0A_0'
  \qquad\text{and}\qquad
  s
  =
  A_0^{-1\prime}e_1,
\]
we have
\[
  \Sigma^{-1}s
  =
  A_0A_0'A_0^{-1\prime}e_1
  =
  A_0e_1
  =
  a_1.
\]
Hence,
\[
  \Sigma^{-1}G
  =
  a_1\gamma_m'.
\]
Because $A_0$ is invertible, both $a_1$ and $s$ are nonzero. It follows that
\[
  \gamma_m\neq0_{n_m\times1}
  \quad\Longleftrightarrow\quad
  G\neq0_{n\times n_m}
  \quad\Longleftrightarrow\quad
  \rank(G)=1.
\]
Since $\Sigma^{-1}$ is invertible,
\[
  \rank(\Sigma^{-1}G)
  =
  \rank(G).
\]
Thus, the relevance condition is equivalent to the two rank conditions stated in the theorem.

Suppose first that $\gamma_m\neq0_{n_m\times1}$. Let $q\in\Real^{n_m}$ satisfy $Gq\neq0_{n\times1}$. From \eqref{eq:mp_rank_one},
\[
  Gq
  =
  (\gamma_m'q)s,
  \qquad
  \Sigma^{-1}Gq
  =
  (\gamma_m'q)a_1.
\]
Since $s\neq0_{n\times1}$, the condition $Gq\neq0_{n\times1}$ implies $\gamma_m'q\neq0$. The unit-variance normalization of the first shock gives
\[
  a_1'\Sigma a_1
  =
  e_1'A_0'A_0^{-1\prime}A_0^{-1}A_0e_1
  =
  e_1'e_1
  =
  1.
\]
Therefore,
\[
  (\Sigma^{-1}Gq)'\Sigma(\Sigma^{-1}Gq)
  =
  (\gamma_m'q)^2,
\]
and
\[
  \frac{\Sigma^{-1}Gq}
  {\sqrt{(\Sigma^{-1}Gq)'\Sigma(\Sigma^{-1}Gq)}}
  =
  \sgn(\gamma_m'q)a_1.
\]
This proves the expression for $a_1$ in \eqref{eq:mp_identified_objects}. Moreover,
\[
  \Sigma a_1
  =
  A_0^{-1\prime}A_0^{-1}A_0e_1
  =
  A_0^{-1\prime}e_1
  =
  s,
\]
which implies
\[
  s'\Sigma^{-1}s
  =
  a_1'\Sigma a_1
  =
  1.
\]
Using $Gq=(\gamma_m'q)s$ gives
\[
  (Gq)'\Sigma^{-1}(Gq)
  =
  (\gamma_m'q)^2,
\]
and hence
\[
  \frac{Gq}
  {\sqrt{(Gq)'\Sigma^{-1}(Gq)}}
  =
  \sgn(\gamma_m'q)s.
\]
This proves the expression for $s$ in \eqref{eq:mp_identified_objects}.

Let
\[
  B_\ell
  =
  A_\ell A_0^{-1},
  \qquad
  \ell=1,\ldots,\cL,
\]
and let
\[
  \widetilde c'
  =
  c'A_0^{-1}.
\]
Since $a_1=A_0e_1$,
\[
  a_{\ell,1}
  =
  A_\ell e_1
  =
  B_\ell a_1,
  \qquad
  \ell=1,\ldots,\cL,
\]
and
\[
  c_1
  =
  c'e_1
  =
  \widetilde c'a_1.
\]
Thus, identification of $a_1$ identifies all contemporaneous, lagged, and constant coefficients in the first equation under the same sign normalization.

The factorization $G=s\gamma_m'$ and the normalization $s'\Sigma^{-1}s=1$ imply
\[
  \gamma_m'
  =
  s'\Sigma^{-1}G.
\]
Hence,
\[
  \delta_m
  =
  \Omega_m^{-1}\gamma_m
\]
is also identified under the common sign normalization. The impulse responses to the first shock satisfy
\[
  \mathrm{IRF}_1(h)
  =
  s'\Psi_h,
  \qquad
  h\geq0.
\]
Since $s$ and the reduced-form dynamic coefficients determining $\Psi_h$ are identified, the impulse responses to the first shock are identified at every horizon. This proves sufficiency.

It remains to prove necessity. Suppose that
\[
  \gamma_m
  =
  0_{n_m\times1}.
\]
Then
\[
  \mathbb E(\eps_tm_t')
  =
  0_{n\times n_m}
  \qquad\text{and}\qquad
  G
  =
  0_{n\times n_m}.
\]
Because $n\geq2$, there exists an orthogonal $n\times n$ matrix $P$ such that $Pe_1\neq e_1$ and $Pe_1\neq-e_1$. Define
\[
  A_\ell^*
  =
  A_\ell P,
  \qquad
  \ell=0,\ldots,\cL,
\]
\[
  c^*
  =
  P'c,
  \qquad
  \eps_t^*
  =
  P'\eps_t.
\]
Orthogonality gives
\[
  \mathbb E(\eps_t^*\eps_t^{*\prime})
  =
  I_n,
\]
and
\[
  \mathbb E(\eps_t^*m_t')
  =
  P'\mathbb E(\eps_tm_t')
  =
  0_{n\times n_m}.
\]
Thus, the transformed representation satisfies the same proxy restrictions. Its reduced-form coefficients are unchanged because
\[
  A_\ell^*(A_0^*)^{-1}
  =
  A_\ell PP'A_0^{-1}
  =
  A_\ell A_0^{-1},
  \qquad
  \ell=1,\ldots,\cL,
\]
and
\[
  c^{*\prime}(A_0^*)^{-1}
  =
  c'PP'A_0^{-1}
  =
  c'A_0^{-1}.
\]
The reduced-form innovation is also unchanged:
\[
  \eps_t^{*\prime}(A_0^*)^{-1}
  =
  \eps_t'PP'A_0^{-1}
  =
  \eps_t'A_0^{-1}.
\]
Hence, the original and transformed systems imply the same distribution for the observed variables and instruments. Their first contemporaneous coefficient vectors nevertheless differ:
\[
  a_1^*
  =
  A_0^*e_1
  =
  A_0Pe_1.
\]
Because $A_0$ is invertible and $Pe_1\neq\pm e_1$,
\[
  a_1^*
  \neq
  \pm a_1.
\]
Thus, when $\gamma_m=0_{n_m\times1}$, two observationally equivalent parameterizations satisfy the same proxy restrictions but have first equations that differ by more than a common sign. The first equation is therefore not identified. This proves necessity and establishes the if-and-only-if result.

Let $Q$ be any orthogonal $(n-1)\times(n-1)$ matrix and define
\[
  H
  =
  \diag(1,Q).
\]
Postmultiplying $A_\ell$ by $H$ for $\ell=0,\ldots,\cL$, replacing $c$ by $H'c$, and replacing $\eps_t$ by $H'\eps_t$ leaves the reduced-form unchanged. Since $He_1=e_1$, this transformation also leaves the first equation, the first shock, the proxy restrictions, and the impulse responses to the first shock unchanged, while it generally changes the remaining equations. The remaining equations are therefore not identified without additional restrictions and are not needed for the identified impulse responses.

Last, Proposition~\ref{prop:mp_normalization} shows that orthogonal transformations of columns $2,\ldots,n$ leave the joint covariance restriction, the first equation, and the impulse responses to the first shock unchanged. Thus, the remaining equations are not identified and are not needed for these impulse responses.
\end{proof}

\runinhead{The two-stage coefficients} The main text states the two-stage coefficient in \eqref{eq:mp_gk_ratio} without derivation. We derive it here from population moments, using $G$, $g_i$, and $W=\Omega_m^{-1}$ as defined in Section~\ref{subsec:mp_proxy}. Fix an admissible anchor $p$. The linear projection of the anchor innovation on the instrument vector is
\[
  v_{p,t}
  =
  m_t'\pi_p+\zeta_{p,t},
  \qquad
  \pi_p
  =
  Wg_p,
  \qquad
  \mathbb E(m_t\zeta_{p,t})
  =
  0_{n_m\times1},
\]
where $\zeta_{p,t}$ is the projection residual. The fitted value is $\widetilde v_{p,t}=m_t'Wg_p$. The population two-stage least squares coefficient from regressing $v_{i,t}$ on $v_{p,t}$ using $m_t$ as the instrument vector is the ratio of the covariance of $v_{i,t}$ with the fitted value to the covariance of $v_{p,t}$ with the fitted value. Since
\[
  \mathbb E(\widetilde v_{p,t}v_{i,t})
  =
  g_p'Wg_i
  \qquad\text{and}\qquad
  \mathbb E(\widetilde v_{p,t}v_{p,t})
  =
  g_p'Wg_p,
\]
we obtain
\[
  \beta_{i\mid p}
  =
  \frac{\mathbb E(\widetilde v_{p,t}v_{i,t})}
       {\mathbb E(\widetilde v_{p,t}v_{p,t})}
  =
  \frac{g_p'Wg_i}{g_p'Wg_p},
\]
which is \eqref{eq:mp_gk_ratio}. The denominator is positive because $W$ is positive definite and $g_p\neq0_{n_m\times1}$. Under Assumption~\ref{ass:mp_proxy}, an anchor is admissible if and only if $s_p\neq0$, because \eqref{eq:mp_rank_one} gives $g_p=\gamma_ms_p$ and $\gamma_m\neq0_{n_m\times1}$.

\begin{proof}[Proof of Proposition~\ref{prop:mp_gk_equivalence}]
By \eqref{eq:mp_rank_one}, $G=s\gamma_m'$, which gives $g_i=\gamma_ms_i$ for every $i$. For any admissible anchor $p$,
\[
  \beta_{i\mid p}
  =
  \frac{g_p'Wg_i}{g_p'Wg_p}
  =
  \frac{s_ps_i\,\gamma_m'W\gamma_m}{s_p^{2}\,\gamma_m'W\gamma_m}
  =
  \frac{s_i}{s_p},
\]
where $\gamma_m'W\gamma_m>0$ because $W$ is positive definite and $\gamma_m\neq0_{n_m\times1}$. Collecting these coefficients gives $d_p=s/s_p$.

The unit-variance normalization gives $s'\Sigma^{-1}s=1$, as established in the proof of Theorem~\ref{thm:mp_identification}. Therefore,
\[
  d_p'\Sigma^{-1}d_p
  =
  \frac{s'\Sigma^{-1}s}{s_p^{2}}
  =
  \frac{1}{s_p^{2}},
  \qquad
  \frac{d_p}{\sqrt{d_p'\Sigma^{-1}d_p}}
  =
  \sgn(s_p)\,s.
\]
The scale normalization in \eqref{eq:mp_gk_scale} therefore returns $s$ up to sign, and the economic sign normalization selects the same $s$ for every admissible anchor.

For the joint procedure, Theorem~\ref{thm:mp_identification} identifies $s$ directly from $G$ and $\Sigma$. For any $q\in\Real^{n_m}$ such that $Gq\neq0_{n\times1}$, the identity $Gq=(\gamma_m'q)s$ and the normalization $s'\Sigma^{-1}s=1$ reduce the second expression in \eqref{eq:mp_identified_objects} to $\pm s$. The same economic sign normalization selects the same $s$. Both procedures therefore identify the same impact vector.

The two procedures also identify the same first equation. Since $a_1=\Sigma^{-1}s$, they identify the same contemporaneous coefficient vector. The reduced-form matrices $B_\ell=A_\ell A_0^{-1}$ are common to both procedures, and hence
\[
  a_{\ell,1}
  =
  A_\ell e_1
  =
  B_\ell a_1,
  \qquad
  \ell=1,\ldots,\cL,
\]
is the same under both. Likewise,
\[
  c_1
  =
  c'e_1
  =
  c'A_0^{-1}a_1,
\]
and $c'A_0^{-1}$ is the reduced-form constant common to both procedures. Thus, every coefficient in the first equation agrees after the common sign normalization.

Finally, the moving-average matrices $\Psi_h$ in \eqref{eq:mp_Psi} are functions of the common reduced-form matrices. Equation~\eqref{eq:mp_irf} therefore gives
\[
  \mathrm{IRF}_1(h)
  =
  s'\Psi_h
\]
for both procedures at every horizon $h\geq0$.
\end{proof}

Two features of this equivalence deserve emphasis. First, the anchor enters the two-stage construction even though the identified object does not depend on it. The fitted value $\widetilde v_{p,t}=s_pm_t'W\gamma_m$ and the ratio vector $d_p=s/s_p$ vary with $p$, and a second admissible anchor $r$ gives $d_r=(s_p/s_r)d_p$. Under the rank-one restriction, these representations are proportional, and the common scale and sign normalization maps each of them to the same $s$. Second, Proposition~\ref{prop:mp_anchor} characterizes the role of this rank-one restriction. When it fails, ratio vectors based on different anchors need not be proportional and, in general, no common normalization reconciles them. The equivalence in Proposition~\ref{prop:mp_gk_equivalence} is therefore a consequence of the maintained proxy restrictions, not a generic property of the two-stage procedure.

The same distinction has a finite-sample counterpart. Let $\widehat G$ denote the sample counterpart of $G$, formed from the estimated reduced-form innovations, and let $\widehat g_i$ denote the transpose of its $i$th row. With one instrument, each $\widehat g_i$ is a scalar and
\[
  \widehat\beta_{i\mid p}
  =
  \frac{\widehat g_i}{\widehat g_p}
\]
whenever $\widehat g_p\neq0$, so changing the anchor only rescales the same estimated direction. With more than one instrument, sampling variation generally makes $\widehat G$ have rank greater than one even when $G$ has rank one. The vectors $\widehat g_1,\ldots,\widehat g_n$ then need not share a common direction, and
\[
  \widehat\beta_{i\mid p}
  =
  \frac{\widehat g_p'\widehat W\widehat g_i}
       {\widehat g_p'\widehat W\widehat g_p}
\]
depends on the anchor through $\widehat g_p$. Different anchors therefore generally produce estimated directions that are not proportional. This finite-sample anchor dependence can arise even when the population proxy restrictions are correct, and it is distinct from a population failure of the rank-one restriction. The joint method avoids anchor dependence because it neither selects a residual anchor nor divides by an estimated anchor loading. It does not remove the need for valid and relevant instruments or for the rank-one restriction to hold in population.

The following propositions establish two implications of the joint formulation beyond identification. Proposition~\ref{prop:mp_normalization} shows that rotations of the remaining equations leave the identified equation, shock, and impulse responses unchanged. Proposition~\ref{prop:mp_efficiency} shows that joint estimation attains the asymptotic efficiency bound and weakly dominates two-stage estimation of the same object under the stated regularity conditions.

\begin{proposition} 
\label{prop:mp_normalization}
Suppose that $n\geq2$, the joint system satisfies \eqref{eq:mp_joint_covariance}, and $\gamma_m\neq0_{n_m\times1}$. Let $Q$ be any orthogonal $(n-1)\times(n-1)$ matrix and define
\[
  H
  =
  \diag(1,Q).
\]
For $\ell=0,\ldots,\cL$, define
\[
  A_\ell^*
  =
  A_\ell H,
  \qquad
  c^*
  =
  H'c,
  \qquad
  \eps_t^*
  =
  H'\eps_t.
\]
Then the transformed coefficients and shocks imply the same reduced form and satisfy the same joint covariance restriction as the original joint system. The transformation leaves the first equation, the first shock, and the impulse responses to that shock unchanged at every horizon. Consequently, the equations associated with columns $2,\ldots,n$ are not identified without additional restrictions, and their rotational indeterminacy does not affect identification of the first equation or its impulse responses.
\end{proposition}

\begin{proof}
Theorem~4 of \citet{tZ99} establishes the block-rotation invariance underlying this result. The first equation forms one block, while the equations associated with columns $2,\ldots,n$ form the transformed block. Because the SVAR is written in row-vector form, the premultiplication in \citet{tZ99} appears here as postmultiplication of the coefficient matrices by $H$. We verify the result directly and then impose the external-instrument covariance restriction.

The matrix $H$ is orthogonal and leaves the first coordinate unchanged:
\[
  H'H
  =
  HH'
  =
  I_n,
  \qquad
  He_1
  =
  H'e_1
  =
  e_1.
\]
Postmultiplying \eqref{eq:svar_full} by $H$ gives the transformed structural representation because
\[
  A_\ell^*
  =
  A_\ell H,
  \qquad
  c^{*\prime}
  =
  c'H,
  \qquad
  \eps_t^{*\prime}
  =
  \eps_t'H.
\]

For each $\ell=1,\ldots,\cL$,
\[
  A_\ell^*(A_0^*)^{-1}
  =
  A_\ell H(A_0H)^{-1}
  =
  A_\ell HH'A_0^{-1}
  =
  A_\ell A_0^{-1}.
\]
The reduced-form constant is also unchanged:
\[
  c^{*\prime}(A_0^*)^{-1}
  =
  c'HH'A_0^{-1}
  =
  c'A_0^{-1}.
\]
Under the transformed representation, the reduced-form innovation satisfies
\[
  v_t^{*\prime}
  =
  \eps_t^{*\prime}(A_0^*)^{-1}
  =
  \eps_t'HH'A_0^{-1}
  =
  \eps_t'A_0^{-1}
  =
  v_t'.
\]
Thus, the observed-data reduced form is unchanged.

We next verify the joint covariance restriction. Since $\eps_t^*=H'\eps_t$,
\[
  \mathbb E(\eps_t^*\eps_t^{*\prime})
  =
  H'\mathbb E(\eps_t\eps_t')H
  =
  I_n.
\]
Equation~\eqref{eq:mp_joint_covariance} gives $\mathbb E(\eps_tm_t')=e_1\gamma_m'$, and hence
\[
  \mathbb E(\eps_t^*m_t')
  =
  H'\mathbb E(\eps_tm_t')
  =
  H'e_1\gamma_m'
  =
  e_1\gamma_m'.
\]
Because $m_t$ is not transformed, its covariance remains $\Omega_m$. Therefore,
\[
  \mathbb E
  \left(
  \begin{bmatrix}
    \eps_t^*\\
    m_t
  \end{bmatrix}
  \begin{bmatrix}
    \eps_t^{*\prime} & m_t'
  \end{bmatrix}
  \right)
  =
  \begin{bmatrix}
    I_n & e_1\gamma_m'\\
    \gamma_m e_1' & \Omega_m
  \end{bmatrix}.
\]
Thus, the transformed representation satisfies the same joint covariance restriction.

The first column of every coefficient matrix is unchanged. For $\ell=0,\ldots,\cL$,
\[
  A_\ell^*e_1
  =
  A_\ell He_1
  =
  A_\ell e_1.
\]
In particular,
\[
  a_1^*
  =
  a_1,
  \qquad
  a_{\ell,1}^*
  =
  a_{\ell,1},
  \qquad
  \ell=1,\ldots,\cL.
\]
The constant in the first equation is also unchanged:
\[
  c_1^*
  =
  e_1'c^*
  =
  e_1'H'c
  =
  e_1'c
  =
  c_1.
\]
The first shock is unchanged as well:
\[
  \eps_{1t}^*
  =
  e_1'\eps_t^*
  =
  e_1'H'\eps_t
  =
  \eps_{1t}.
\]

Because $\gamma_m$ and $\Omega_m$ are unchanged, $\delta_m=\Omega_m^{-1}\gamma_m$ is unchanged. Hence the projection residual
\[
  w_{1t}
  =
  \eps_{1t}-m_t'\delta_m
\]
is also unchanged. The augmented first equation \eqref{eq:mp_joint_first_equation} is therefore identical under the original and transformed representations. Moreover, \eqref{eq:mp_joint_covariance} implies
\[
  \mathbb E(w_{1t}\eps_{-1,t}')
  =
  0_{1\times(n-1)},
\]
where $\eps_{-1,t}$ collects shocks $2,\ldots,n$. Since $\eps_{-1,t}^*=Q'\eps_{-1,t}$,
\[
  \mathbb E(w_{1t}\eps_{-1,t}^{*\prime})
  =
  \mathbb E(w_{1t}\eps_{-1,t}')Q
  =
  0_{1\times(n-1)}.
\]
Thus, the orthogonality conditions in the joint system are preserved without invoking any assumption beyond \eqref{eq:mp_joint_covariance}.

It remains to verify the impulse responses. Let $s'=e_1'A_0^{-1}$. Under the transformed representation,
\[
  s^{*\prime}
  =
  e_1'(A_0^*)^{-1}
  =
  e_1'H'A_0^{-1}
  =
  e_1'A_0^{-1}
  =
  s'.
\]
The reduced-form coefficient matrices are unchanged, so the moving-average matrices satisfy $\Psi_h^*=\Psi_h$ for every $h\geq0$. Therefore,
\[
  \mathrm{IRF}_1^*(h)
  =
  s^{*\prime}\Psi_h^*
  =
  s'\Psi_h
  =
  \mathrm{IRF}_1(h)
  \qquad
  \text{for every }h\geq0.
\]
This is the specialization of the block invariance result in Theorem~4 of \citet{tZ99} to the first shock and the block of remaining equations. The external-instrument covariance restriction is preserved because the transformation fixes the first coordinate. The remaining equations can vary across observationally equivalent representations. Under the maintained relevance condition, Theorem~\ref{thm:mp_identification} identifies the first equation up to sign, and the transformation above leaves that equation, the first shock, and its impulse responses unchanged.
\end{proof}

\begin{proposition} 
\label{prop:mp_efficiency}
Suppose that the Gaussian joint proxy SVAR is correctly specified and admits a finite-dimensional, locally identified regular parameterization $\vartheta$ in a neighborhood of an interior true value $\vartheta_0$. Assume that the likelihood is locally asymptotically normal with nonsingular Fisher information $\mathcal I(\vartheta_0)$, that the posterior satisfies the Bernstein--von Mises conclusion \citep[Theorem~10.1]{vanderVaart1998} centered at the maximum-likelihood estimator $\widehat\vartheta_{\mathrm{ML}}$, and that the prior density is positive and continuously differentiable in a neighborhood of $\vartheta_0$. Assume further that there exists a sequence $\delta_T\downarrow0$ with $\sqrt{T}\delta_T\rightarrow\infty$ such that, with probability approaching one, an interior posterior mode $\widehat\vartheta_{\mathrm{PM}}$ and $\widehat\vartheta_{\mathrm{ML}}$ lie in $B(\vartheta_0,\delta_T)$, that $T^{-1}\nabla^2\ell_T(\vartheta)$ converges uniformly to $-\mathcal I(\vartheta_0)$ on this neighborhood, and that the log-prior score is bounded there. Let $\psi(\vartheta)$ be any continuously differentiable identified object common to the joint and two-stage methods, including $a_1$, $s$, or a finite collection of impulse responses to the first shock. Then the joint posterior for $\sqrt{T}\{\psi(\vartheta)-\psi(\widehat\vartheta_{\mathrm{ML}})\}$ converges to
\[
  N(0,V_{\mathrm{eff}}),
  \qquad
  V_{\mathrm{eff}}
  =
  D\psi(\vartheta_0)\mathcal I(\vartheta_0)^{-1}D\psi(\vartheta_0)'.
\]
Under the mode-localization and uniform-Hessian conditions,
\[
  \widehat\vartheta_{\mathrm{PM}}-\widehat\vartheta_{\mathrm{ML}}
  =
  O_p(T^{-1})
  =
  o_p(T^{-1/2}).
\]
If a two-stage estimator $\widehat\psi_{\mathrm{2S}}$ is regular and asymptotically linear for the same object under the same joint model and has asymptotic covariance $V_{\mathrm{2S}}$, then
\[
  V_{\mathrm{2S}}-V_{\mathrm{eff}}
\]
is positive semidefinite. Equality holds if and only if the influence function of $\widehat\psi_{\mathrm{2S}}$ equals the efficient influence function almost surely.
\end{proposition}

\begin{proof}
Let $b(\vartheta)=\nabla\log\pi(\vartheta)$ denote the log-prior score. The first-order conditions for the posterior mode and the maximum-likelihood estimator imply
\[
  0
  =
  \nabla\ell_T(\widehat\vartheta_{\mathrm{PM}})
  +
  b(\widehat\vartheta_{\mathrm{PM}})
  =
  \nabla^2\ell_T(\widetilde\vartheta)
  (\widehat\vartheta_{\mathrm{PM}}-\widehat\vartheta_{\mathrm{ML}})
  +
  b(\widehat\vartheta_{\mathrm{PM}})
\]
for some $\widetilde\vartheta$ on the line segment joining $\widehat\vartheta_{\mathrm{PM}}$ and $\widehat\vartheta_{\mathrm{ML}}$. The uniform Hessian limit, nonsingularity of $\mathcal I(\vartheta_0)$, and bounded log-prior score therefore yield
\[
  \widehat\vartheta_{\mathrm{PM}}-\widehat\vartheta_{\mathrm{ML}}
  =
  O_p(T^{-1})
  =
  o_p(T^{-1/2}).
\]
The assumed Bernstein--von Mises limit and the delta method give the stated posterior limit for $\psi(\vartheta)$. The parametric information bound for regular estimators under the same correctly specified joint model gives
\[
  V_{\mathrm{2S}}
  -
  D\psi(\vartheta_0)\mathcal I(\vartheta_0)^{-1}D\psi(\vartheta_0)'
  \succeq
  0.
\]
For a regular asymptotically linear estimator, equality holds if and only if its influence function equals the efficient influence function almost surely.
\end{proof}

\section{Full Sets of Impulse Responses}\label{sec:full_irfs}
Figures~\ref{fig:selected_irfs1}-\ref{fig:selected_irfs4} display the full set of impulse responses for the selected 13-variable system for the boom and bust application. Figures~\ref{fig:bpss_orig_irfs1}-\ref{fig:bpss_orig_irfs4} display the full set of replicated impulse responses for the original 10-variable heteroskedastic system of BPSS. The replication uses updated monthly data and the monthly Sims-Zha prior, which is similar to the prior used by BPSS. To be compatible with BPSS's original article, the scale of impulse responses is not adjusted to represent percent or percentage points. 

\begin{figure}[htbp]
\centering
\includegraphics[width=\textwidth]{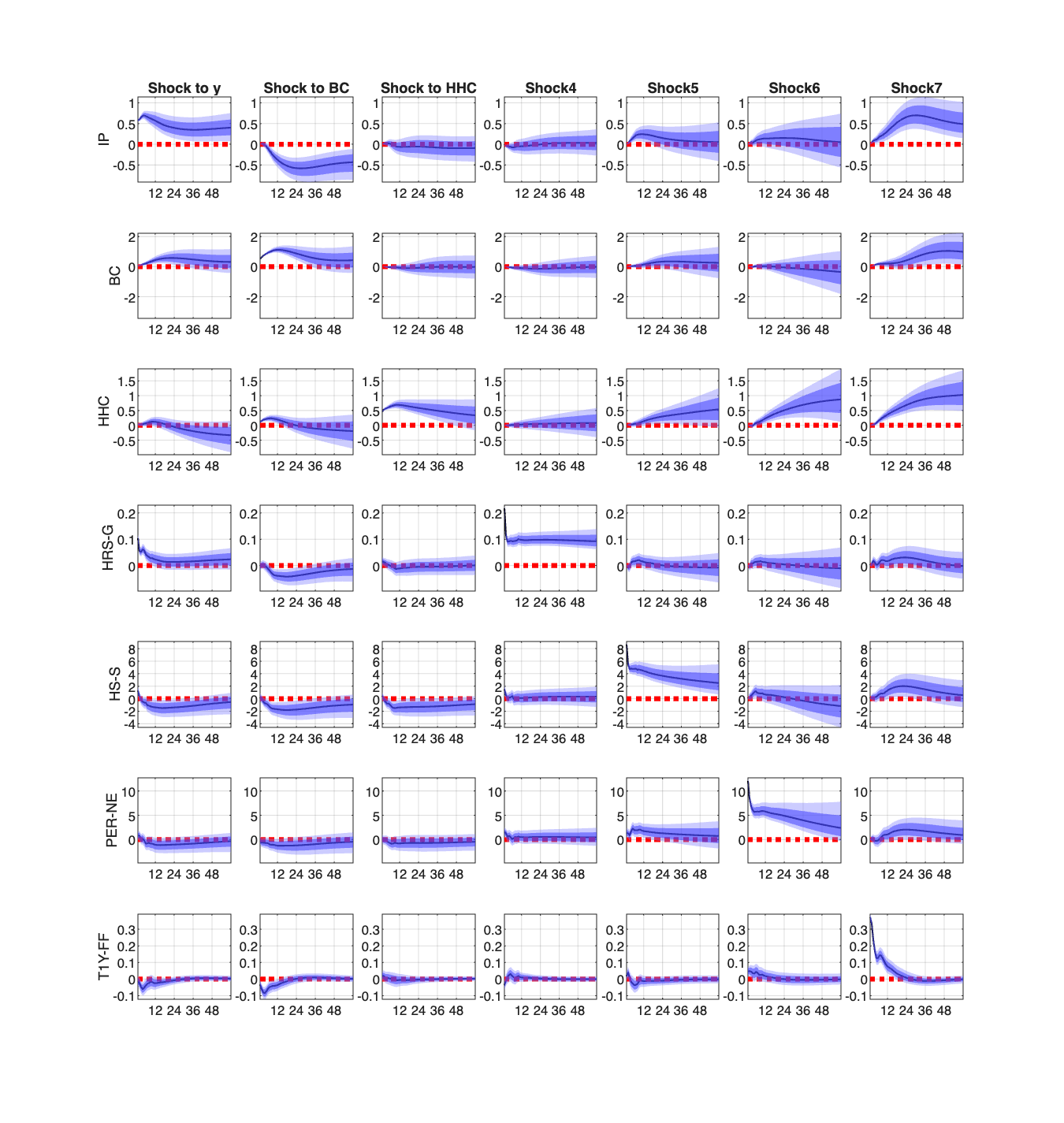}
\caption{Part 1 of the full set of impulse responses for the selected 13-variable system. All impulse responses are expressed in percent. The dark band shows the 68\% posterior credible band and the light band the 90\% posterior credible band. Table~\ref{tab:mp_model11_variables} provides complete descriptions of the variable labels shown on the y-axis.}
\label{fig:selected_irfs1}
\end{figure}

\begin{figure}[htbp]
\centering
\includegraphics[width=\textwidth]{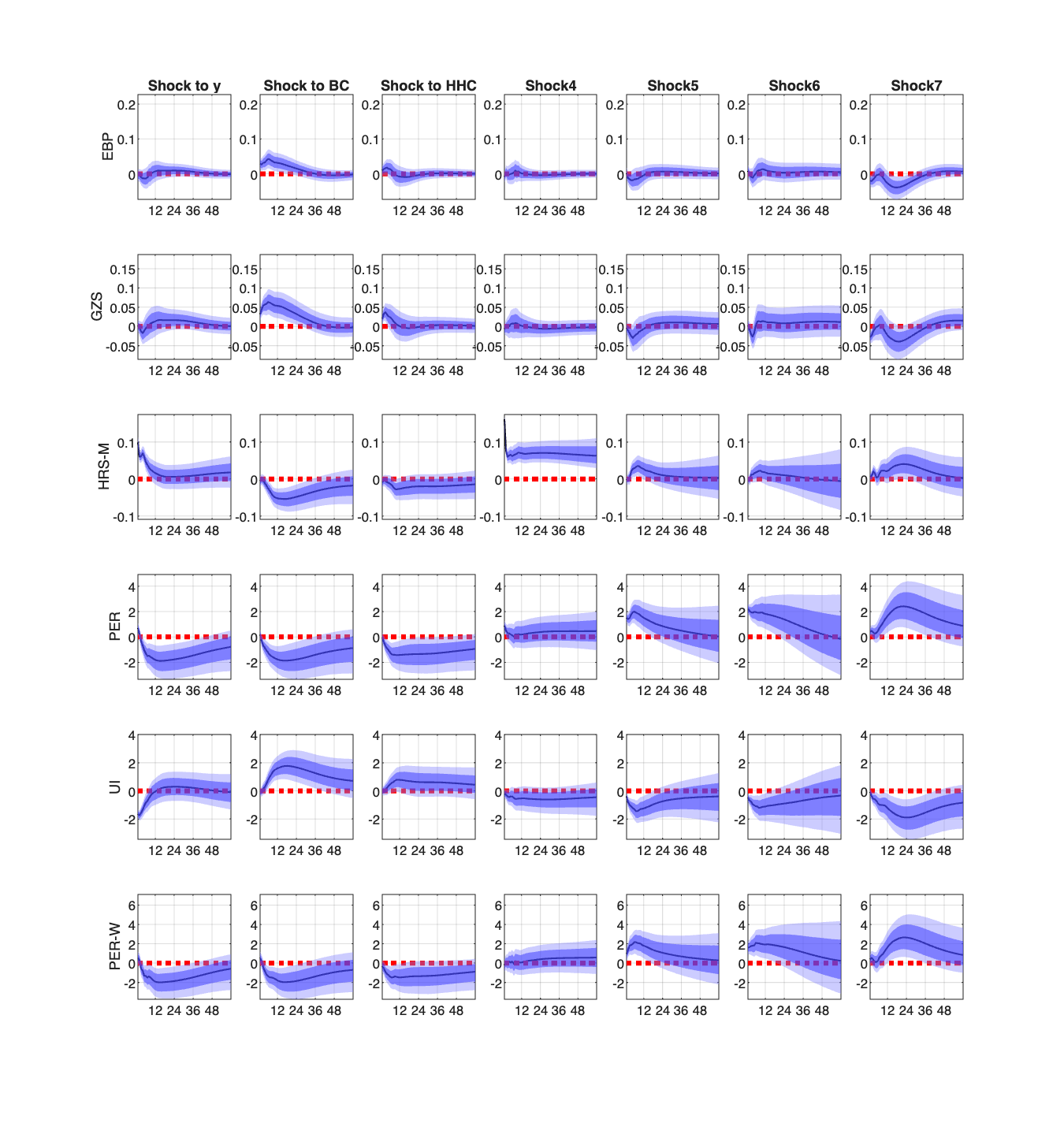}
\caption{Part 2 of the full set of impulse responses for the selected 13-variable system. All impulse responses are expressed in percent. The dark band shows the 68\% posterior credible band and the light band the 90\% posterior credible band. Table~\ref{tab:mp_model11_variables} provides complete descriptions of the variable labels shown on the y-axis.}
\label{fig:selected_irfs2}
\end{figure}

\begin{figure}[htbp]
\centering
\includegraphics[width=\textwidth]{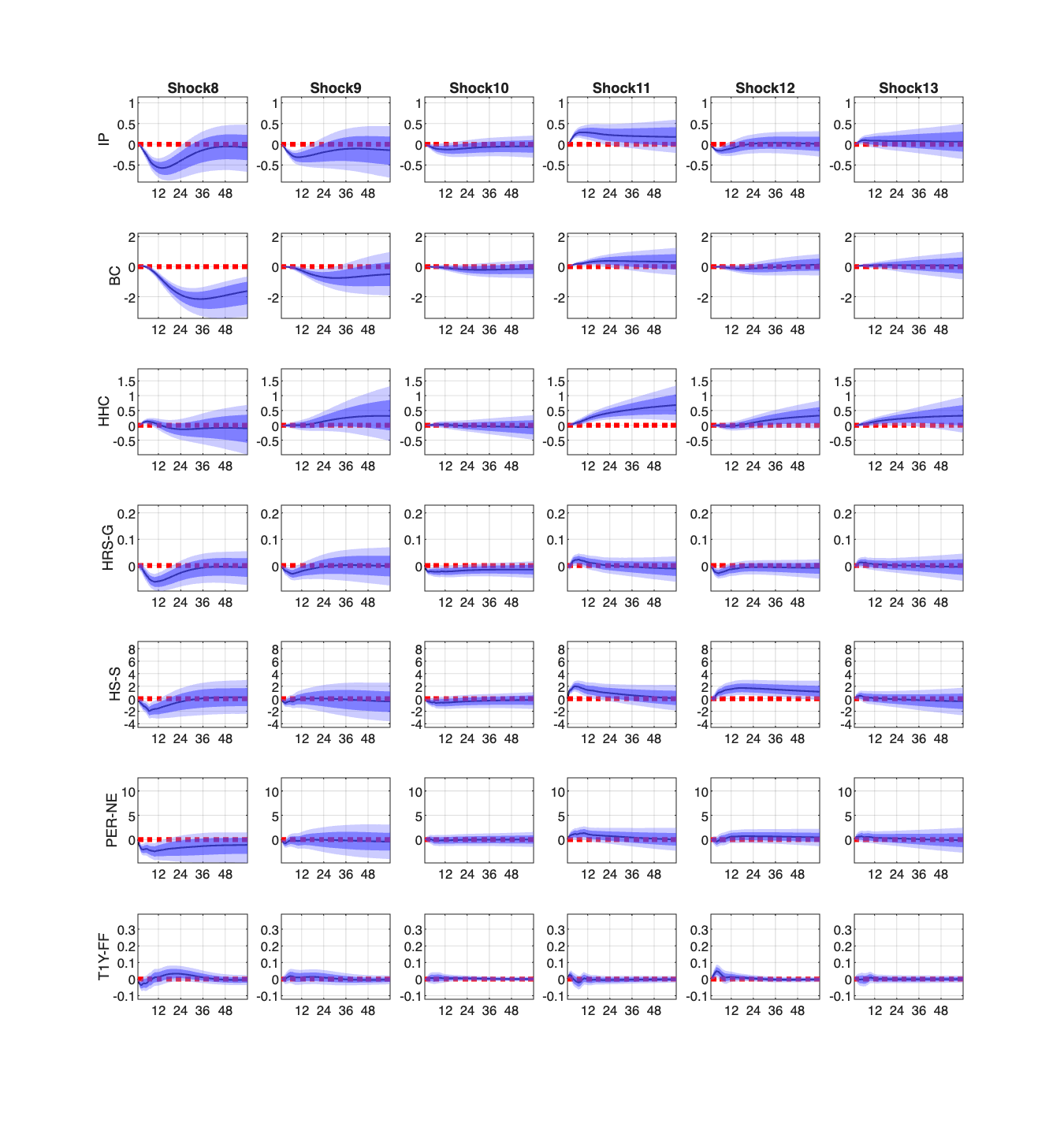}
\caption{Part 3 of the full set of impulse responses for the selected 13-variable system. All impulse responses are expressed in percent. The dark band shows the 68\% posterior credible band and the light band the 90\% posterior credible band. Table~\ref{tab:mp_model11_variables} provides complete descriptions of the variable labels shown on the y-axis.}
\label{fig:selected_irfs3}
\end{figure}

\begin{figure}[htbp]
\centering
\includegraphics[width=\textwidth]{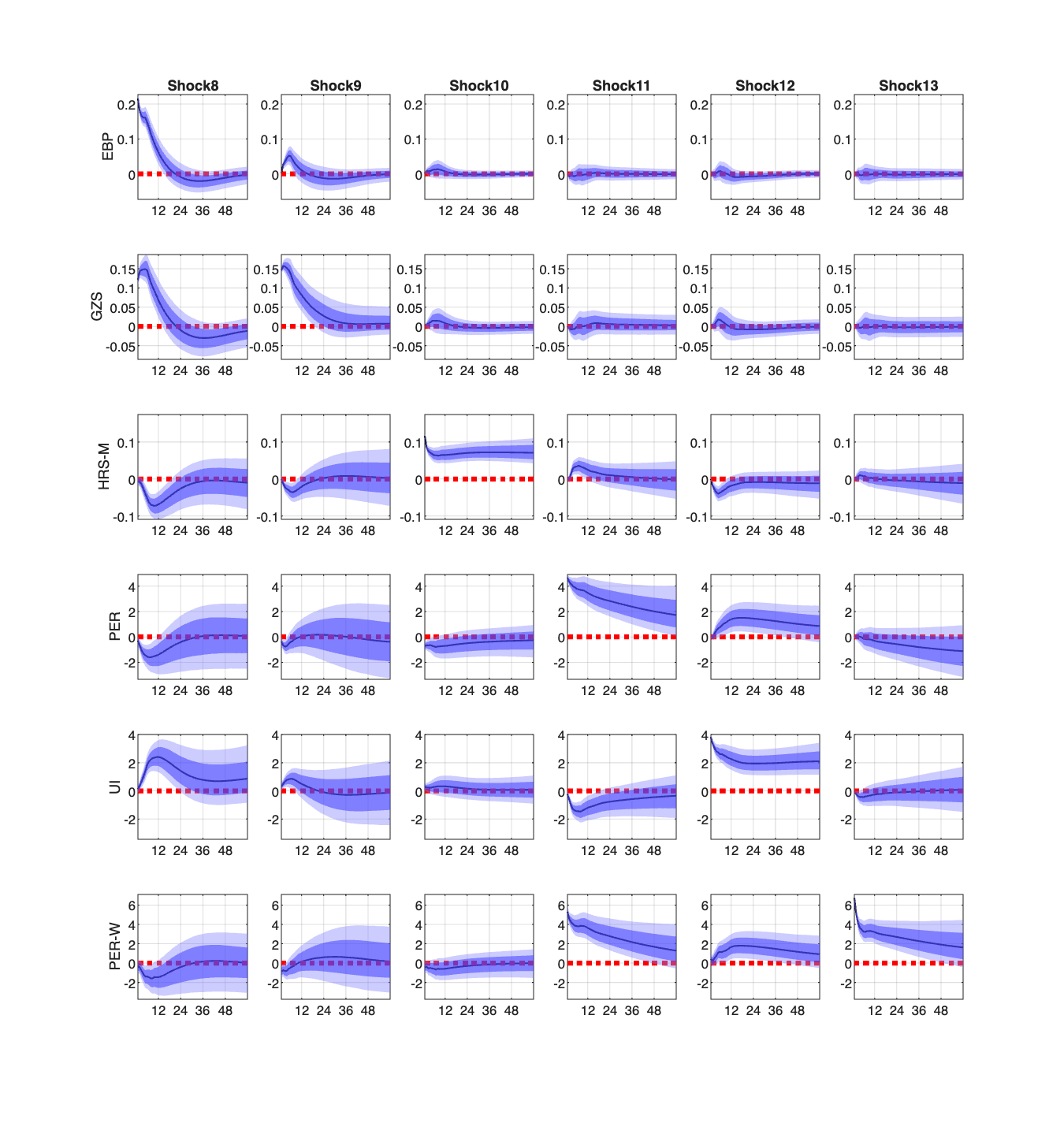}
\caption{Part 4 of the full set of impulse responses for the selected 13-variable system. All impulse responses are expressed in percent. The dark band shows the 68\% posterior credible band and the light band the 90\% posterior credible band. Table~\ref{tab:mp_model11_variables} provides complete descriptions of the variable labels shown on the y-axis.}
\label{fig:selected_irfs4}
\end{figure}

\begin{figure}[htbp]
\centering
\includegraphics[width=\textwidth]{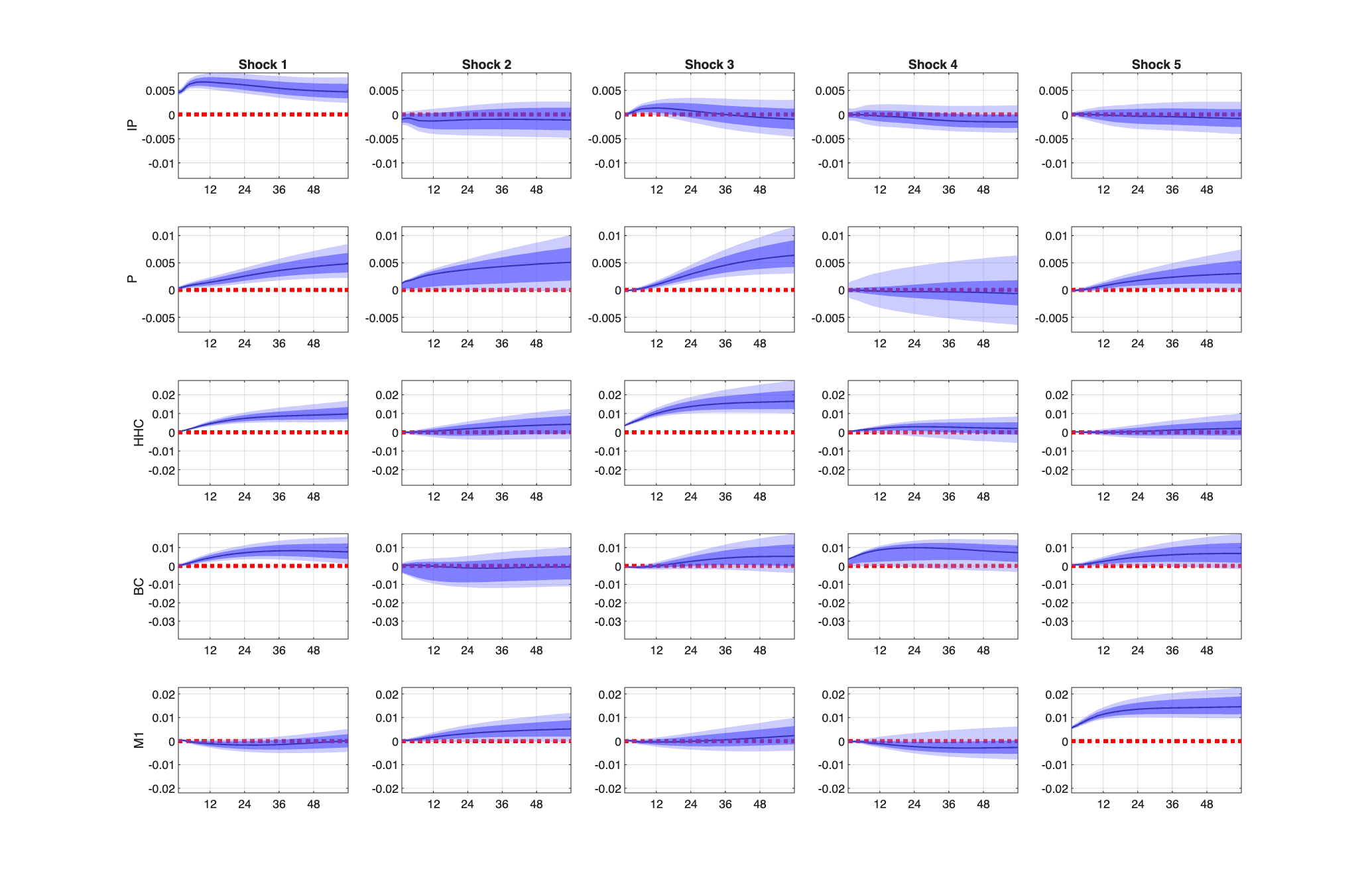}
\caption{Part 1 of the full set of impulse responses for the original 10-variable heteroskedastic system of BPSS. The dark band shows the 68\% posterior credible band and the light band the 90\% posterior credible band. The y-axis labels follow BPSS, who provide complete variable descriptions.}
\label{fig:bpss_orig_irfs1}
\end{figure}

\begin{figure}[htbp]
\centering
\includegraphics[width=\textwidth]{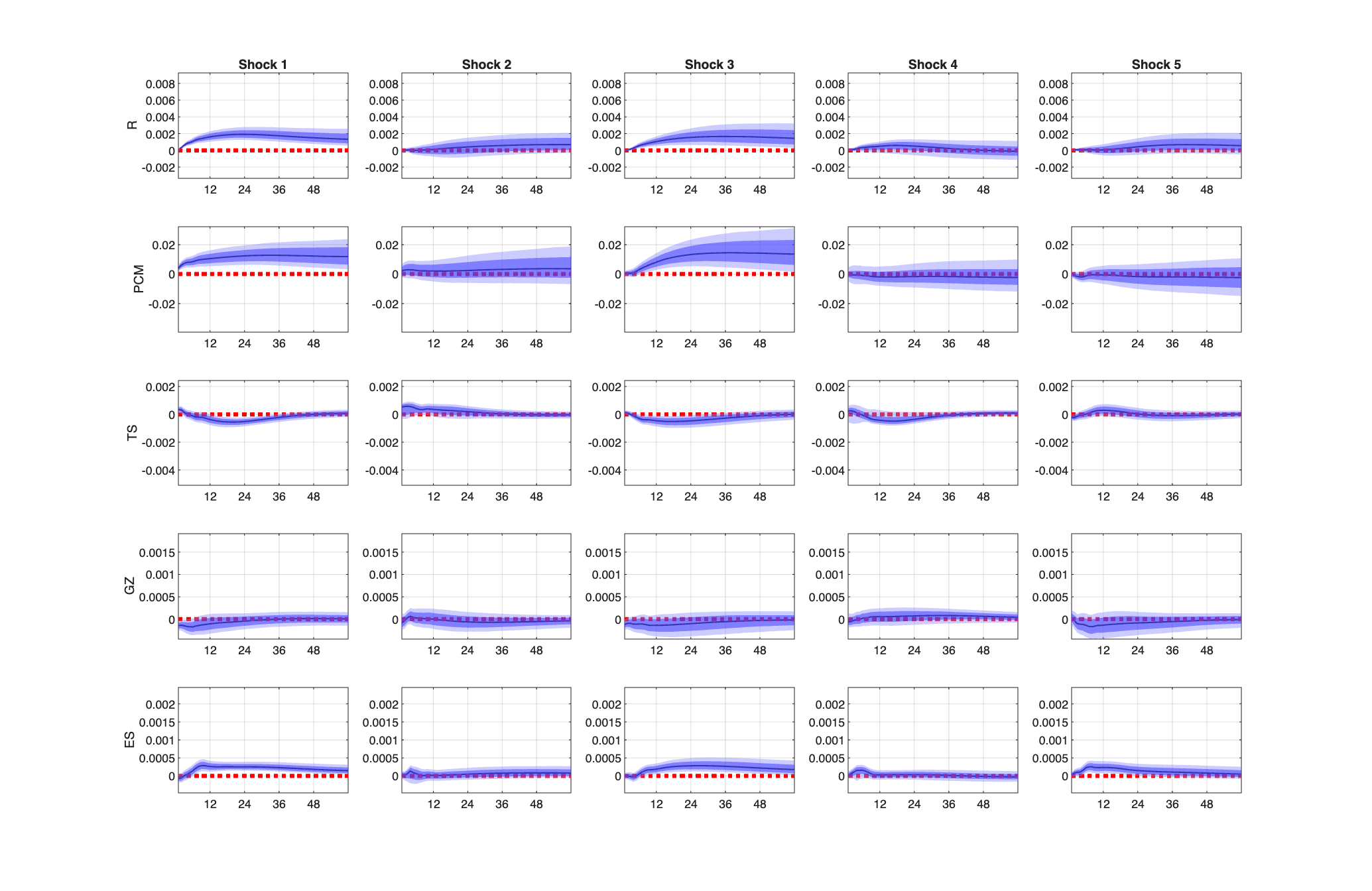}
\caption{Part 2 of the full set of impulse responses for the original 10-variable heteroskedastic system of BPSS. The dark band shows the 68\% posterior credible band and the light band the 90\% posterior credible band. The y-axis labels follow BPSS, who provide complete variable descriptions.}
\label{fig:bpss_orig_irfs2}
\end{figure}

\begin{figure}[htbp]
\centering
\includegraphics[width=\textwidth]{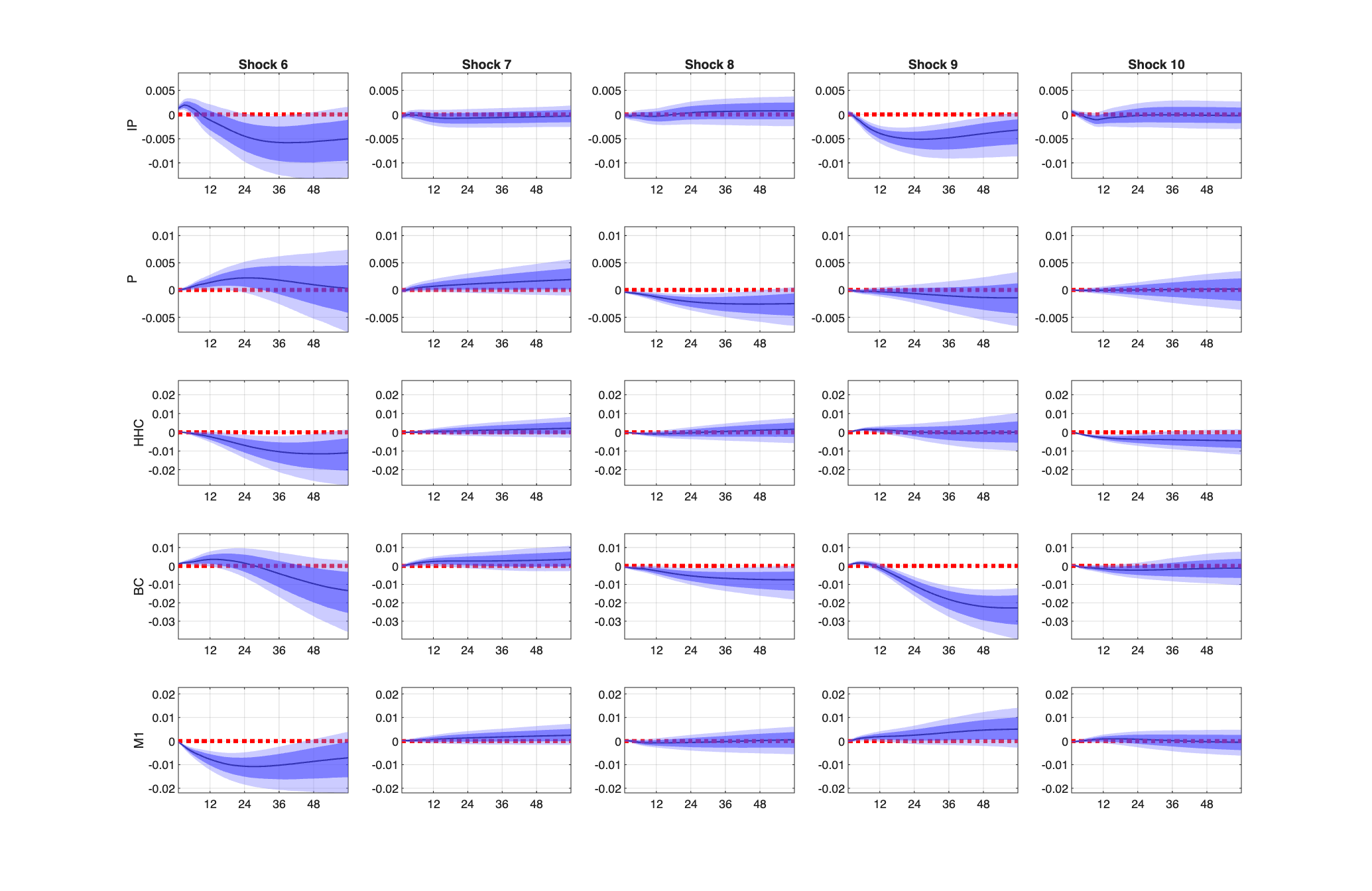}
\caption{Part 3 of the full set of impulse responses for the original 10-variable heteroskedastic system of BPSS. The dark band shows the 68\% posterior credible band and the light band the 90\% posterior credible band. The y-axis labels follow BPSS, who provide complete variable descriptions.}
\label{fig:bpss_orig_irfs3}
\end{figure}

\begin{figure}[htbp]
\centering
\includegraphics[width=\textwidth]{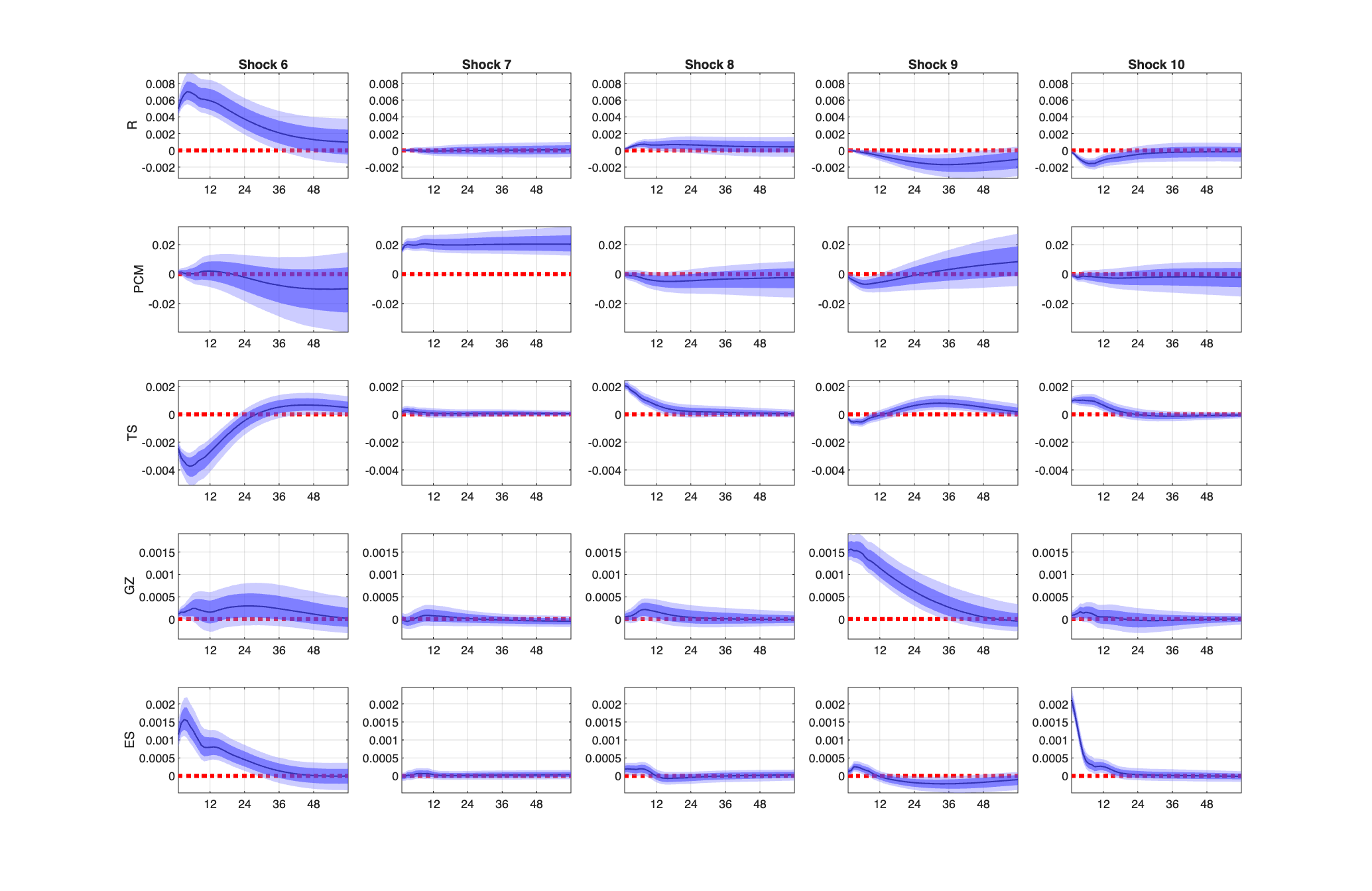}
\caption{Part 4 of the full set of impulse responses for the original 10-variable heteroskedastic system of BPSS. The dark band shows the 68\% posterior credible band and the light band the 90\% posterior credible band. The y-axis labels follow BPSS, who provide complete variable descriptions.}
\label{fig:bpss_orig_irfs4}
\end{figure}

\clearpage
\bibliographystyle{plainnat}
\bibliography{tzref3}
\end{document}